\documentclass[accepted]{uai2025} 
                        
\usepackage[american]{babel}

\usepackage{natbib} 
\usepackage{mathtools} 
\usepackage{booktabs} 
\usepackage{tikz} 
\usepackage{amsthm}
\usepackage{amssymb}
\usepackage{algorithm}
\usepackage{algorithmic}
\usepackage{thm-restate}
\usepackage{cleveref}
\usepackage{subfig}
\usepackage{thmtools}

\newtheorem{definition}{Definition}
\newtheorem{lemma}{Lemma}

\newtheorem{proposition}{Proposition}
\newtheorem{condition}{Condition}

\newcommand{\anc}{\text{Anc}}
\newcommand{\ant}{\text{Ant}}

\usetikzlibrary{arrows,arrows.meta}
\usetikzlibrary{shapes,shapes.arrows}

\tikzstyle{ndint} = [draw, semithick, shape=rectangle, minimum size=20pt,inner sep=0pt]
\tikzstyle{ndout} = [draw, semithick, shape=circle, minimum size=20pt,inner sep=0pt]
\tikzstyle{ndlat} = [draw, semithick, shape=circle, minimum size=20pt,inner sep=0pt, fill=lightgray]
\tikzstyle{ndsel} = [draw, semithick, shape=circle, minimum size=20pt,inner sep=0pt, fill=lightgray, double]
\tikzstyle{arint} = [style={->,>=Latex,thick}] 
\tikzstyle{arout} = [style={->,>=Latex,thick}] 
\tikzstyle{arlout} = [style={->,>=Latex,thick}] 
\tikzstyle{arlat} = [style={<->,>=Latex,thick}] 

\tikzstyle{out} = [style={o->,>=stealth',thick,draw,fill}] 
\tikzstyle{hut} = [style={<-,>=Latex,style=semithick}]
\tikzstyle{tut} = [style=semithick]
\tikzstyle{tuh} = [style={->,>=Latex,style=semithick}]
\tikzstyle{huh} = [style={<->,>=Latex,style=semithick}]
\tikzstyle{hus} = [style={<-{Rays[n=6]},>=Latex,style=semithick}]
\tikzstyle{tus} = [style={-{Rays[n=6]},style=semithick}]
\tikzstyle{sut} = [style={{Rays[n=6]}-,style=semithick}]
\tikzstyle{suh} = [style={{Rays[n=6]}->,>=Latex,style=semithick}]
\tikzstyle{sus} = [style={{Rays[n=6]}-{Rays[n=6]},style=semithick}]

\newcommand{\tuh}{\mathrel{\rightarrow}}
\newcommand{\hut}{\mathrel{\leftarrow}}
\newcommand{\huh}{\mathrel{\leftrightarrow}}
\newcommand{\hus}{\mathrel{{\leftarrow\mkern-11mu\ast}}}
\newcommand{\suh}{\mathrel{{\ast\mkern-11mu\to}}}
\newcommand{\sut}{\mathrel{{\ast\mkern-11mu\relbar\mkern-9mu\relbar}}}
\newcommand{\tus}{\mathrel{{\relbar\mkern-11mu\relbar\mkern-11mu\ast}}}
\newcommand{\sus}{\mathrel{{\ast\mkern-11mu\relbar\mkern-9mu\relbar\mkern-11mu\ast}}}
\newcommand{\tut}{\mathrel{{\relbar\mkern-9mu\relbar}}}

\title{$\sigma$-Maximal Ancestral Graphs}

\author[1]{\href{mailto:<binghuayao71210@gmail.com>}{Binghua Yao}{}}
\author[1]{\href{mailto:<J.M.Mooij@uva.nl>}{Joris M. Mooij}{}}

\affil[1]{%
    Korteweg-de Vries Institute\\
    University of Amsterdam\\
    Amsterdam, The Netherlands
}

\begin{document}
\maketitle

\begin{abstract}%
Maximal Ancestral Graphs (MAGs) provide an abstract representation of Directed Acyclic Graphs (DAGs) with latent (selection) variables. These graphical objects encode information about ancestral relations and d-separations of the DAGs they represent. This abstract representation has been used amongst others to prove the soundness and completeness of the FCI algorithm for causal discovery, and to derive a do-calculus for its output. One significant inherent limitation of MAGs is that they rule out the possibility of cyclic causal relationships. In this work, we address that limitation. We introduce and study a class of graphical objects that we coin ``$\sigma$-Maximal Ancestral Graphs'' (``$\sigma$-MAGs''). We show how these graphs provide an abstract representation of (possibly cyclic) Directed Graphs (DGs) with latent (selection) variables, analogously to how MAGs represent DAGs. We study the properties of these objects and provide a characterization of their Markov equivalence classes.
\end{abstract}

\section{Introduction}\label{sec:intro}
Maximal Ancestral Graphs (MAGs) were introduced by \citet{richardson2002ancestral} to provide an abstract representation of Directed Acyclic Graphs (DAGs) with latent variables (where the latent variables can either be marginalized out, or conditioned out in case of selection bias). MAGs encode information about ancestral relations and d-separations of the DAGs they represent. Key results in the theory of MAGs are various characterizations of their Markov equivalence classes \citep{spirtes1996polynomial,Ali++2005,ZhaoZhengLiu2005,AliRichardsonSpirtes2009}. MAGs provide an appropriate level of abstraction for various complex causal reasoning tasks. For example, MAGs have been used to prove the soundness and completeness of the FCI algorithm for causal discovery \citep{SMR1995,Ali++2005,Zhang2008_JMLR}, and to derive a causal do-calculus for its output \citep{Zhang2008_AI}.

One inherent limitation of MAGs is that they rule out the possibility of cyclic causal relationships.
Both feedback loops and the merging of variables may lead to directed cycles in the causal graph, and this requires the consideration of directed graphs (DGs) rather than directed acyclic graphs (DAGs).
A first attempt to generalize MAGs to cyclic models are the Maximal Almost Ancestral Graphs (MAAGs) introduced by \citet{Strobl2015}.
While MAAGs do represent ancestral relations of DGs with latent (selection) variables, they fail to generalize the other key feature of MAGs: MAAGs do \emph{not} capture the $d$-separation properties of the DGs that they represent.
Another approach was taken by \citet{claassen2023establishing}.
Already in \citeyear{richardson1997characterization},
\citeauthor{richardson1997characterization} characterized Markov equivalence for DGs (without latent (selection) variables), providing necessary and sufficient conditions for when two DGs entail the same $d$-separations.
Building on this, \citet{claassen2023establishing} proposed Cyclic Maximal Ancestral Graphs (CMAGs) as MAG-like representations of DGs without latent (selection) variables. 
While CMAGs contain enough information to characterize the Markov equivalence of the DGs they represent, it is not clear at present whether the $d$-separation properties can be read off easily from the CMAG.
Furthermore, it is not clear how CMAGs can be generalized to allow for latent (selection) variables.
In conclusion, these proposals provide only partial solutions at best.

Another limitation of existing approaches is that they focus on representing $d$-separations.
While the $d$-separation Markov property holds for DAG models (e.g., acylic structural causal models and causal Bayesian networks \citep{Lauritzen90}) and specific DG models (e.g., linear structural equation models \citep{Spirtes94,Koster96}), it fails to hold for DG models in general \citep{spirtes95,neal2000}.

Recently, a more general theory of cyclic structural causal models (SCMs) has been developed \citep{BFPM21}. The subclass of \emph{simple} SCMs was identified as having particularly convenient mathematical properties, while providing a substantive extension of the class of acyclic SCMs. In particular, the $d$-separation Markov property for acyclic SCMs extends to a more generally applicable $\sigma$-separation Markov property for simple SCMs \citep{forre2017markov,BFPM21}. 
In general, a $\sigma$-separation in the DG of a simple SCM implies the corresponding $d$-separation, but not vice versa.
The $d$-separation Markov property, which may imply additional conditional independence relations, only holds for certain subclasses of simple SCMs, in particular for acyclic SCMs, linear SCMs, and SCMs with discrete variables \citep{forre2017markov,BFPM21}.

In this context, a natural question is whether the notion of MAGs can be extended to represent the $\sigma$-separation properties of (possibly cyclic) DGs with latent (selection) variables. In this work, we answer this question affirmatively and in full generality.

Our contributions are as follows. We introduce and study a class of graphical objects that we coin ``$\sigma$-Maximal Ancestral Graphs'' (``$\sigma$-MAGs''). We show how these graphs provide an abstract representation of (possibly cyclic) directed graphs with latent (selection) variables, similarly to how MAGs represent DAGs with latent (selection) variables. By drawing inspiration from \citep{spirtes1996polynomial}, we build up the theory of $\sigma$-MAGs culminating in a characterization of their Markov equivalence classes.

In this paper, we define how a \( \sigma \)-MAG represents a Directed Mixed Graph (DMG), rather than a DG with latent variables. The reason is that marginalizing over latent variables in DGs is already well-established: DGs can be marginalized into DMGs by a graphical procedure known as ``marginalization'' or ``latent projection'' \citep{BFPM21}. To avoid redundancy and to streamline our exposition, we focus directly on the abstraction from DMGs to \( \sigma \)-MAGs, which highlights our main contribution: capturing the effects of conditioning on latent selection variables in cyclic settings.

\section{Preliminaries}
In this work, we propose a new class of mixed graphs, referred to as \emph{$\sigma$-Maximal Ancestral Graphs} ($\sigma$-MAGs), which provide a representation for sets of Directed Mixed Graphs (DMGs). We begin by establishing some terminologies.

$\sigma$-MAGs have multiple edge types. In addition to the three edge types for DMGs (\(\tuh, \hut, \huh\)), there is an undirected edge (\(\tut\)). Each edge \( a \sus b \) has two \emph{edge marks}, one at each node, with each edge mark being either a \emph{tail} or an \emph{arrowhead}. We use the ``\(*\)'' symbol to denote any of these two edge marks. Thus, the notation \( a \sus b \) represents all four possible edge types between \( a \) and \( b \). Edges of the form \( a \hus b \) and \( b \suh a \) are called \emph{into} \( a \). Edges of the form \( a \tus b \) and \( b \sut a \) are called \emph{out of} \( a \).
In order to define $\sigma$-MAGs, we extend the common definitions in DMGs to accommodate the new set of edge types.

\begin{definition}
Let $\mathcal{H}=(\mathcal{V},\mathcal{E})$ be a mixed graph with nodes $\mathcal{V}$ and edges $\mathcal{E}$ of the types $\{\tuh, \hut, \huh,\tut\}$. Let $a,b,c,v_0,v_n\in\mathcal{V}$.
\begin{enumerate}
    \item If there is an edge $a\sus b$ between $a$ and $b$, we say that $a$ and $b$ are \emph{adjacent} in $\mathcal{H}$.
    \item A triple of distinct nodes $(a, b, c)$ in $\mathcal{H}$ is called \emph{unshielded} if $b$ is adjacent to both $a$ and $c$ in $\mathcal{H}$, but $a$ is not adjacent to $c$ in $\mathcal{H}$.
    \item A \emph{walk} between $v_0$ and $v_n$ in $\mathcal{H}$ is a finite alternating sequence of nodes and edges:\footnote{Without extra explanations, we consider a walk (path) between $v_0$ and $v_n$ in the form of $v_0\sus v_1\sus\cdots \sus v_{n-1}\sus v_n$.}
    $$
    \langle v_0,e_1,v_1,\ldots,v_{n-1},e_n,v_n\rangle,
    $$
    for some $n\geq0$, such that for every $i=1,\ldots,n$, we have $e_i=v_{i-1}\sus v_i\in\mathcal{E}$.
    \item A walk is called a \emph{path} if no node appears more than once on the walk.
    \item A \emph{directed walk (path)} from $v_0$ to $v_n$ in $\mathcal{H}$ is a walk (path) of the form:
    $$
    v_0\rightarrow v_1\rightarrow\cdots\rightarrow v_{n-1}\rightarrow v_n,
    $$
    for some $n\geq0$.
    \item An \emph{anterior walk (path)} from $v_0$ to $v_n$ in $\mathcal{H}$ is a walk (path) of the form:
    $$
    v_0 \tus v_1 \tus \cdots\tus v_{n-1}\tus v_n,
    $$
    for some $n\geq0$.
    \item A \emph{directed cycle} is a directed walk $a\rightarrow\cdots\rightarrow b$, concatenated with the directed edge $b\rightarrow a$.
    \item An \emph{almost directed cycle} is a directed walk $a\rightarrow\cdots\rightarrow b$, concatenated with the bidirected edge $b\huh  a$.
    \item If there is a directed walk from $a$ to $b$ in $\mathcal{H}$, we say that $a$ is an ancestor of $b$ in $\mathcal{H}$, and we write $a\in\anc_\mathcal{H}(b)$. For a set $A\subseteq\mathcal{V}$, we define $\anc_\mathcal{H}(A)=\bigcup_{v\in A}\anc_\mathcal{H}(v)$.
    \item A triple of consecutive nodes $v_{k-1}\sus v_k\sus v_{k+1}$ on a walk in $\mathcal{H}$ is called a \emph{collider} if it is of the form $v_{k-1}\suh v_k \hus v_{k+1}$.
    \item A triple of consecutive nodes $v_{k-1}\sus v_k\sus v_{k+1}$ on a walk in $\mathcal{H}$ is called a \emph{non-collider} if it is of the form $v_{k-1}\sus v_k\tus v_{k+1}$ or $v_{k-1}\sut v_k\sus v_{k+1}$. Furthermore, we also refer to the endpoints $v_0$ and $v_n$ of a walk between $v_0$ and $v_n$ in $\mathcal{H}$ as \emph{non-colliders}.
    \item Let \( v_0, v_n \in \mathcal{V} \) and \( Z \subseteq \mathcal{V} \). If \( \pi \) is a path between \( v_0 \) and \( v_n \) in \( \mathcal{H} \), then the \emph{Collider Distance Sum to Z}  of \( \pi \) is given by  
\[
\sum_{i=0}^{n} l_\mathcal{H}(v_i, Z) \cdot 1_{\{ v_i \text{ is a collider on } \pi \}},
\]
where \( l_\mathcal{H}(v_i,Z) \) denotes the length of the shortest directed path from \( v_i \) to \( Z \) in \( \mathcal{H} \).

\end{enumerate}
\end{definition}

\section{Definition \& Characterization} \label{section:1}
In this section, we introduce the class of $\sigma$-MAGs, which encodes certain ancestral relationships and $\sigma$-connections of DMGs. Leveraging this property, we will later introduce the corresponding separation criterion for $\sigma$-MAGs and characterize their Markov equivalence.

To facilitate the definition of $\sigma$-MAGs, we first define the notion of inducing paths in mixed graphs in analogy to inducing paths in DAGs \citep{VermaPearl1990}. This differs from the related notion of $\sigma$-inducing path in DMGs (Definition~\ref{sigmaSeparation}).

\begin{definition}[Inducing Walk (Path)]
In a mixed graph $\mathcal{H}$ with nodes $ \mathcal{V} $, and edges $ \mathcal{E} $ of the types $\{\rightarrow, \leftarrow, \huh , -\}$, a walk (path) \( \pi \) in \(\mathcal{H}\) between \( v_0 \) and \( v_n \) is called an \emph{inducing walk (path)} if every node on \( \pi \), except for the endpoints, is a collider and belongs to \( \anc_\mathcal{H}(\{v_0, v_n\}) \).  
\end{definition}

Whereas in MAGs, undirected edges appear in case of conditioning (for example when modeling selection bias), in $\sigma$-MAGs they can also encode the existence of directed cycles in the DMGs that the $\sigma$-MAG represents. One can sometimes distinguish the two by looking at the connectivities of the undirected edges in the $\sigma$-MAG, which motivates the following definition. 

\begin{definition}[Complete and Incomplete Neighborhoods]
In a mixed graph $\mathcal{H}$ with nodes $ \mathcal{V} $, and edges $ \mathcal{E} $ of the types $\{\rightarrow, \leftarrow, \huh , -\}$, the \emph{neighborhood} of a node $a$ is the set of nodes connected to $a$ by undirected edges, denoted \(\text{Nbh}_\mathcal{H}(a) = \{v \in \mathcal{V} : a - v \in \mathcal{E}\}\). 

We say that the neighborhood of \(a\) is \emph{complete} if it forms a clique, meaning that for any two nodes \(b, c \in \text{Nbh}_\mathcal{H}(a)\), we have \(b - c \in \mathcal{E}\). Otherwise, the neighborhood of \(a\) is said to be \emph{incomplete}.
\end{definition}

As we will demonstrate later, complete neighborhoods in a $\sigma$-MAG may correspond to strongly connected components (See Definition~\ref{family relationships} for more details) in DMGs, while incomplete neighborhoods must correspond to ancestors of selection variables in DMGs.

We can now state the proposed abstract definition of $\sigma$-MAGs:

\begin{definition}[\(\sigma\)-MAG]  
A mixed graph \( \mathcal{H} \) with nodes \( \mathcal{V} \) and edges $\mathcal{E}$ of the types \( \{\rightarrow, \leftarrow, \huh , -\} \) is called a \emph{\( \sigma \)-maximal ancestral graph (\(\sigma\)-MAG)} if all of the following conditions hold:  
\begin{enumerate}  
    \item Between any two distinct nodes, there is at most one edge, and no node is adjacent to itself.  
    \item $\mathcal{H}$ is \emph{ancestral}: If \( \mathcal{H} \) contains an anterior path \( a \tus \cdots \tus b \) then it does not contain an edge \( a \hus b \).
    \item $\mathcal{H}$ is \emph{\( \sigma \)-maximal}:
    \begin{enumerate}
        \item $\mathcal{H}$ is \emph{maximal}: There is no inducing path between any two distinct non-adjacent nodes.
        \item \( \mathcal{H} \) is \emph{\(\sigma\)-complete}: If \( \mathcal{H} \) contains a triple \( a \suh b \tut c \), then \( a \) and \( c \) must be adjacent in \( \mathcal{H} \). Furthermore, if \( \mathcal{H} \) also contains an edge \( b \tut d \), then \( c \) and \( d \) must be adjacent in \( \mathcal{H} \).
    \end{enumerate}
\end{enumerate}  
\end{definition}

It follows immediately from this definition that MAGs (as defined in \citep{richardson2002ancestral})
can be seen as a subclass of \( \sigma \)-MAGs, namely those that have no edges into an undirected edge (the pattern \( v_i \suh v_j \tut v_k\) cannot occur in a MAG by definition).
One can show that the assumptions concerning such patterns also imply certain orientations.

\begin{restatable}[Fundamental Property of $\sigma$-MAGs]{lemma}{lemmadef} \label{oldDef}
Let \( \mathcal{H} \) be a \( \sigma \)-MAG. If \( \mathcal{H} \) contains a triple of the form \( a \suh b \tut c \), then the edge between \( a \) and \( c \) is of the same type as the edge between \( a \) and \( b \), and the neighborhoods of \( b \) and $c$ are both complete.
\end{restatable}

Further, it follows that a \( \sigma \)-MAG contains no directed cycles or almost directed cycles.

An example of a mixed graph that is not a valid $\sigma$-MAG is given in Figure~\ref{fig:nonmaximal_ancestral_graph}. Various examples of valid $\sigma$-MAGs are given in Figure~\ref{fig:representing_examples} (on the left).

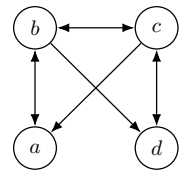
\begin{figure}[t]\centering
  \scalebox{0.8}{\begin{tikzpicture}
    \begin{scope}
        \node[ndout] (X1) at (-1,0) {$a$};
        \node[ndout] (X2) at (-1,2) {$b$};
        \node[ndout] (X3) at (1,2) {$c$};
        \node[ndout] (X4) at (1,0) {$d$};
        \draw[huh] (X1) edge (X2);
        \draw[huh] (X2) edge (X3);
        \draw[huh] (X3) edge (X4);
        \draw[tuh] (X2) edge (X4);
        \draw[tuh] (X3) edge (X1);
    \end{scope}
  \end{tikzpicture}}
  \caption{This mixed graph is not a valid $\sigma$-MAG, because it has an inducing path $a \huh b \huh  c \huh d$ while $a$ and $d$ are non-adjacent.\label{fig:nonmaximal_ancestral_graph}}
\end{figure}

The essential connection between $\sigma$-MAGs and DMGs is obtained by the following definition, which expresses in which way a $\sigma$-MAG ``represents'' a class of DMGs:

\begin{definition}[Representing DMGs by $\sigma$-MAGs] \label{representing}
Let $\mathcal{H}$ be a mixed graph with nodes $ \mathcal{V} $, and edges \(\mathcal{E}\) of the types \(\{\tuh, \hut, \huh, \tut\}\). Let \(\mathcal{G}\) be a DMG with nodes \(\mathcal{V}^+ = \mathcal{V} \cup \mathcal{S}\). We say that \(\mathcal{H}\) represents \(\mathcal{G}\) given \(\mathcal{S}\) if all of the following conditions hold:
\begin{enumerate}
    \item Between any two distinct nodes in $\mathcal{V}$, there is at most one edge, and no edge connects a node to itself in $\mathcal{H}$.
    \item Two distinct nodes \(a,b \in  \mathcal{V}\) are adjacent in \(\mathcal{H}\) if and only if there exists a \(\sigma\)-inducing path between \(a\) and \(b\) given \(\mathcal{S}\) in \(\mathcal{G}\).
    \item If \(a \hus  b\) in \(\mathcal{H}\), then \(a \notin \anc_\mathcal{G}(\{b\} \cup \mathcal{S})\).
    \item If \(a \tus  b\) in \(\mathcal{H}\), then \(a \in \anc_\mathcal{G}(\{b\} \cup \mathcal{S})\).
\end{enumerate}
\end{definition}

The above definition yields a systematic procedure for constructing a $\sigma$-MAG $\mathcal{H}$ from a given DMG $\mathcal{G}$ that includes selection variables. First, the adjacencies in $\mathcal{H}$ are determined based on the existence of $\sigma$-inducing paths given $\mathcal{S}$ between pairs of nodes in $\mathcal{G}$. Next, the orientations of the edges in $\mathcal{H}$ are derived by analyzing the ancestral relationships among the observed nodes in $\mathcal{G}$. Figure~\ref{construction} illustrates this procedure by showing how to construct a $\sigma$-MAG from a Directed Graph (DG) with latent (selection) variables.

\begin{figure}[t]
    \centering
    \subfloat[DG with latent (selection) variables\label{dg}]{
      \scalebox{0.8}{\begin{tikzpicture}
        \node[ndout] (a) at (-2,0) {$a$};
        \node[ndout] (b) at (0,0) {$b$};
        \node[ndout] (c) at (2,0) {$c$};
        \node[ndout] (d) at (0,2) {$d$};
        \node[ndlat] (u) at (-1,1) {$u$};
        \node[ndsel] (s) at (1,1) {$s$};
        \draw[tuh, bend left = 20] (a) edge (b);
        \draw[hut, bend right = 20] (a) edge (b);
        \draw[tuh] (b) edge (s);
        \draw[tuh] (c) edge (s);
        \draw[tuh] (u) edge (d);
        \draw[tuh] (u) edge (b);
      \end{tikzpicture}}
    }\\

    \subfloat[DMG after marginalizing out latent variables\label{dmg}]{
      \scalebox{0.8}{\begin{tikzpicture}
        \node[ndout] (a) at (-2,0) {$a$};
        \node[ndout] (b) at (0,0) {$b$};
        \node[ndout] (c) at (2,0) {$c$};
        \node[ndout] (d) at (0,2) {$d$};
        \node[ndsel] (s) at (1,1) {$s$};
        \draw[tuh, bend left = 20] (a) edge (b);
        \draw[hut, bend right = 20] (a) edge (b);
        \draw[tuh] (b) edge (s);
        \draw[hut] (s) edge (c);
        \draw[huh] (d) edge (b);
      \end{tikzpicture}}
    }\\

    \subfloat[$\sigma$-MAG constructed from the DMG\label{sigmamag}]{
      \scalebox{0.8}{\begin{tikzpicture}
        \node[ndout] (a) at (-2,0) {$a$};
        \node[ndout] (b) at (0,0) {$b$};
        \node[ndout] (c) at (2,0) {$c$};
        \node[ndout] (d) at (0,2) {$d$};
        \draw[tut] (a) edge (b);
        \draw[tut] (b) edge (c);
        \draw[tuh] (b) edge (d);
        \draw[hut] (d) edge (a);
      \end{tikzpicture}}
    }
    \caption{Constructing a $\sigma$-MAG from a DG with latent variable $u$ and latent selection variable $s$.}
    \label{construction}
\end{figure}
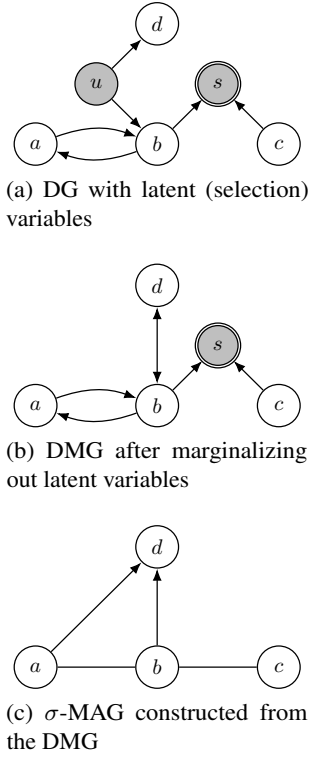

In this example, by marginalizing over latent variables (see Definition~\ref{margilization} for details), we obtain a DMG (Figure~\ref{dmg}) from a DG with latent and selection variables (Figure~\ref{dg}). Then, according to Definition~\ref{representing}, we construct the corresponding $\sigma$-MAG, as shown in Figure~\ref{sigmamag}. For instance, consider the path $a \tuh b \huh d$ in the DMG, which is $\sigma$-inducing given the selection variable $\mathcal{S} = \{s\}$. This implies that $a$ and $d$ are adjacent in the resulting $\sigma$-MAG. Moreover, since $a$ is an ancestor of $s$ in the DMG, while $d$ is not an ancestor of either $a$ or $s$, we can determine the orientation of the edge between $a$ and $d$ in the $\sigma$-MAG to be $a \tuh d$. By repeating this process for all relevant pairs of nodes in the DMG, we can systematically construct the entire $\sigma$-MAG.

Some examples illustrating how $\sigma$-MAGs represent DMGs are shown in Figure~\ref{fig:representing_examples}. In both $\mathcal{H}_1$ and $\mathcal{H}_2$, the neighborhood of \( a \) is complete; the key difference lies in whether there is an arrowhead pointing to \( a \). In particular, $\mathcal{H}_1$ contains the edge \( a \tuh d \), implying that in the DMGs represented by $\mathcal{H}_1$, the set \( \{a,b,c\} \) may either form a strongly connected component or simply be ancestors of the set \( \mathcal{S} \) (as illustrated in $\mathcal{G}_1$ and $\mathcal{G}_2$, respectively). In contrast, $\mathcal{H}_2$ contains the edge \( a \huh b \), and Lemma~\ref{oldDef} guarantees that \( \{a,b,c\} \) must form a strongly connected component, as in $\mathcal{G}_3$ and $\mathcal{G}_4$. Furthermore, if the neighborhood of \( a \) is incomplete, as in \( \mathcal{H}_3 \), then the set \( \{a,b,c\} \) must be a subset of \( \anc_{\mathcal{H}_3}(\mathcal{S}) \) in the corresponding DMGs represented by it (namely, \( \mathcal{G}_5 \) and \( \mathcal{G}_6 \)), although a partial strongly connected component may still exist, as seen in \( \mathcal{G}_6 \).

\begin{figure}[t]
    \centering
    \subfloat[$\mathcal{H}_1$ (left), $\mathcal{G}_1$ (middle), $\mathcal{G}_2$ (right)]{
    \scalebox{0.8}{\begin{tikzpicture}
        \begin{scope}
        \node[ndout] (a) at (-1,0) {$a$};
        \node[ndout] (b) at (1,0) {$b$};
        \node[ndout] (c) at (0,2) {$c$};
        \node[ndout] (d) at (0,-1) {$d$};
        \draw[tut] (a) edge (b);
        \draw[tut] (b) edge (c);
        \draw[tut] (c) edge (a);
        \draw[tuh] (a) edge (d);
        \end{scope}
    \end{tikzpicture}
    \qquad
    \begin{tikzpicture}
        \begin{scope}
        \node[ndout] (a) at (-1,0) {$a$};
        \node[ndout] (b) at (1,0) {$b$};
        \node[ndout] (c) at (0,2) {$c$};
        \node[ndout] (d) at (0,-1) {$d$};
        \draw[tuh] (a) edge (d);
        \draw[tuh, bend left=20] (a) edge (c);
        \draw[tuh, bend left=20] (b) edge (a);
        \draw[tuh, bend left=20] (c) edge (b);
        \end{scope}
    \end{tikzpicture}
    \qquad
    \begin{tikzpicture}
        \begin{scope}
        \node[ndout] (a) at (-1,0) {$a$};
        \node[ndout] (b) at (1,0) {$b$};
        \node[ndout] (c) at (0,2) {$c$};
        \node[ndout] (d) at (0,-1) {$d$};
        \node[ndsel] (s) at (0.5,1) {$s_{bc}$};
        \draw[tuh] (a) edge (c);
        \draw[tuh] (a) edge (b);
        \draw[tuh] (c) edge (s);
        \draw[tuh] (b) edge (s);
        \draw[huh] (d) edge (a);
        \end{scope}
    \end{tikzpicture}}
    }\\

    \subfloat[$\mathcal{H}_2$ (left), $\mathcal{G}_3$ (middle), $\mathcal{G}_4$ (right)]{
      \scalebox{0.8}{\begin{tikzpicture}
        \begin{scope}
        \node[ndout] (a) at (-1,0) {$a$};
        \node[ndout] (b) at (1,0) {$b$};
        \node[ndout] (c) at (0,2) {$c$};
        \node[ndout] (d) at (0,-1) {$d$};
        \draw[tut] (a) edge (b);
        \draw[tut] (b) edge (c);
        \draw[tut] (c) edge (a);
        \draw[huh] (d) edge (a);
        \draw[huh] (d) edge (b);
        \draw[huh] (d) .. controls (2,-0.5) and (2,1.5) .. (c);
        \end{scope}
    \end{tikzpicture}
    \quad
    \begin{tikzpicture}
        \begin{scope}
        \node[ndout] (a) at (-1,0) {$a$};
        \node[ndout] (b) at (1,0) {$b$};
        \node[ndout] (c) at (0,2) {$c$};
        \node[ndout] (d) at (0,-1) {$d$};
        \draw[huh] (d) edge (a);
        \draw[tuh, bend left=20] (a) edge (b);
        \draw[tuh, bend left=20] (a) edge (c);
        \draw[tuh, bend left=20] (b) edge (a);
        \draw[tuh, bend left=20] (c) edge (a);
        \end{scope}
    \end{tikzpicture}
    \qquad
    \begin{tikzpicture}
        \begin{scope}
        \node[ndout] (a) at (-1,0) {$a$};
        \node[ndout] (b) at (1,0) {$b$};
        \node[ndout] (c) at (0,2) {$c$};
        \node[ndout] (d) at (0,-1) {$d$};
        \draw[tuh, bend left=20] (a) edge (c);
        \draw[tuh, bend left=20] (c) edge (b);
        \draw[tuh, bend left=20] (b) edge (a);
        \draw[huh] (d) edge (a);
        \draw[huh] (d) edge (b);
        \end{scope}
    \end{tikzpicture}}
    }\\

    \subfloat[$\mathcal{H}_3$ (left), $\mathcal{G}_5$ (middle), $\mathcal{G}_6$ (right)]{
    \scalebox{0.8}{\begin{tikzpicture}
        \begin{scope}
        \node[ndout] (a) at (-1,0) {$a$};
        \node[ndout] (b) at (1,0) {$b$};
        \node[ndout] (c) at (0,2) {$c$};
        \node[ndout] (d) at (0,-1) {$d$};
        \draw[tut] (a) edge (b);
        \draw[tut] (c) edge (a);
        \draw[tuh] (a) edge (d);
        \end{scope}
    \end{tikzpicture}
    \qquad
    \begin{tikzpicture}
        \begin{scope}
        \node[ndout] (a) at (-1,0) {$a$};
        \node[ndout] (b) at (1,0) {$b$};
        \node[ndout] (c) at (0,2) {$c$};
        \node[ndout] (d) at (0,-1) {$d$};
        \node[ndsel] (s1) at (0,0) {$s_{ab}$};
        \node[ndsel] (s2) at (-0.5,1) {$s_{ac}$};
        \draw[tuh] (a) edge (d);
        \draw[tuh] (a) edge (s1);
        \draw[tuh] (b) edge (s1);
        \draw[tuh] (a) edge (s2);
        \draw[tuh] (c) edge (s2);
        \end{scope}
    \end{tikzpicture}
    \qquad
    \begin{tikzpicture}
        \begin{scope}
        \node[ndout] (a) at (-1,0) {$a$};
        \node[ndout] (b) at (1,0) {$b$};
        \node[ndout] (c) at (0,2) {$c$};
        \node[ndout] (d) at (0,-1) {$d$};
        \node[ndsel] (sa) at (0,0) {$s_{ab}$};
        \draw[tuh] (a) edge (d);
        \draw[tuh] (a) edge (sa);
        \draw[tuh, bend left=20] (a) edge (c);
        \draw[tuh] (b) edge (sa);
        \draw[tuh, bend left=20] (c) edge (a);
        \end{scope}
    \end{tikzpicture}}
    }
    \caption{Examples of $\sigma$-MAGs (left) representing DMGs (middle and right).
    (a) $\mathcal{H}_1$ represents $\mathcal{G}_1$ and represents $\mathcal{G}_2$ given $\mathcal{S} = \{s_{bc}\}$; 
    (b) $\mathcal{H}_2$ represents both $\mathcal{G}_3$ and $\mathcal{G}_4$;
    (c) $\mathcal{H}_3$ represents $\mathcal{G}_5$ given $\mathcal{S} = \{s_{ac}, s_{ab}\}$ and represents $\mathcal{G}_6$ given $\mathcal{S} = \{s_{ab}\}$.}

    \label{fig:representing_examples}
\end{figure}
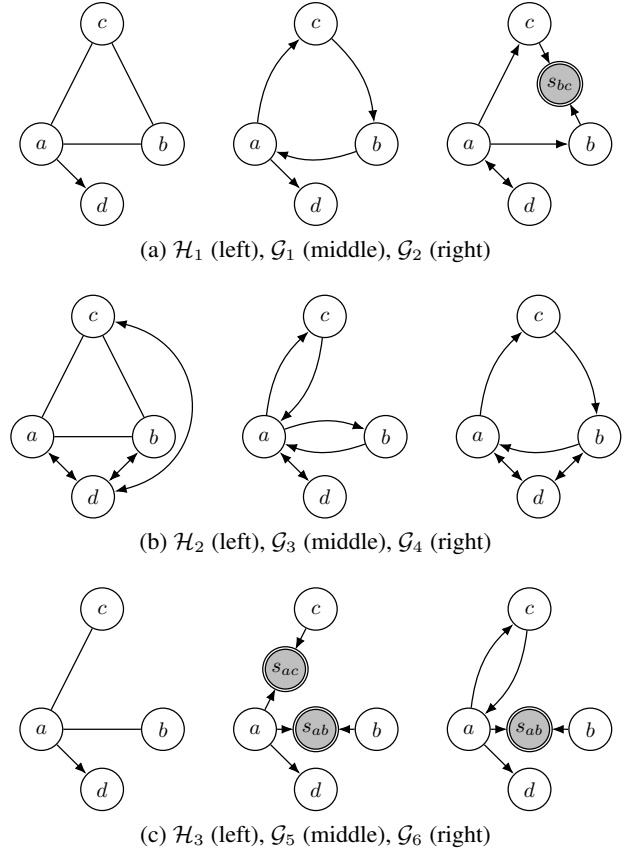

In the rest of this section we will show that the mixed graphs that represent DMGs are characterized as \( \sigma \)-MAGs. The proofs are provided in the Appendix.

A useful property of $\sigma$-MAGs that comes in handy often is:
\begin{restatable}{lemma}{lemmaant} \label{lemmaant}
Let $\mathcal{H}$ be a $\sigma$-MAG. If $\mathcal{H}$ contains an anterior path that starts with a directed edge:
\[
v_0 \rightarrow v_1 \tus  \cdots \tus  v_n
\]
for $n\geq2$, then $v_0$ belongs to the ancestors of $v_n$ in $\mathcal{H}$.
\end{restatable}

The following result demonstrates that every $\sigma$-MAG can represent at least one DMG, given an additional set of selection nodes.

\begin{restatable}{lemma}{lemmaone} \label{lemma1}
Let \(\mathcal{H}\) be a \(\sigma\)-MAG with nodes \( \mathcal{V} \). Then, there exists a DMG \( \mathcal{G} \) with nodes \( \mathcal{V}^+ = \mathcal{V} \cup \mathcal{S} \), such that \(\mathcal{H}\) represents \( \mathcal{G} \) given \( \mathcal{S} \).
\end{restatable}

The following elementary properties demonstrate that certain key characteristics of DMGs are encoded in a $\sigma$-MAG that represents them.

\begin{restatable}{lemma}{lemmazero} \label{lemma0}
Let $\mathcal{H}$ be a $\sigma$-MAG that represents a DMG $\mathcal{G}$ given $\mathcal{S}$. Let $a$ and $b$ be distinct nodes in $\mathcal{H}$.
\begin{enumerate}
    \item $a\in\ant_\mathcal{H}(b)$ implies $a\in\anc_\mathcal{G}(\{b\}\cup \mathcal{S})$.
    \item If $\mathcal{H}$ contains the edge $a \suh b$, then there exists a $\sigma$-inducing path given $\mathcal{S}$ in $\mathcal{G}$ between $a$ and $b$ that is into $b$.
    \item If $\mathcal{H}$ contains the edge $a \huh b$, then there exists a $\sigma$-inducing path given $\mathcal{S}$ in $\mathcal{G}$ between $a$ and $b$ that is into both $a$ and $b$.
\end{enumerate}
\end{restatable}

On the other hand, every mixed graph that represents a DMG must be a $\sigma$-MAG.

\begin{restatable}{lemma}{lemmatwo} \label{lemma2}
Let $\mathcal{H}$ be a mixed graph with nodes $ \mathcal{V} $ and edges \(\mathcal{E}\) of the types \(\{\rightarrow, \leftarrow, \huh , -\}\). Let $\mathcal{G}$ be a DMG with nodes $\mathcal{V}^+ = \mathcal{V} \cup \mathcal{S}$. If $\mathcal{H}$ represents $\mathcal{G}$ given $\mathcal{S}$, then $\mathcal{H}$ is a $\sigma$-MAG.
\end{restatable} 

Finally, we obtain the following result characterizing $\sigma$-MAGs as the mixed graphs that represent DMGs in the sense of Definition~\ref{representing}.

\begin{restatable}{theorem}{thmone} \label{thm1}
Let \( \mathcal{H} \) be a mixed graph with nodes \( \mathcal{V} \) and edges \( \mathcal{E} \) of the types \(\{\rightarrow, \leftarrow, \huh , -\}\). The following equivalence holds:  
\[
\begin{aligned}
&\text{There exists a DMG }\mathcal{G} \text{ with nodes } \mathcal{V}^+ = \mathcal{V} \cup \mathcal{S}, \\
&\quad \text{such that } \mathcal{H} \text{ represents } \mathcal{G} \text{ given }\mathcal{S}. \\
&\iff \mathcal{H} \text{ is a } \sigma\text{-MAG}.
\end{aligned}
\]
\end{restatable}

\section{Separation Criterion} \label{section:2}
In this section, we extend the commonly used $m$-separation criterion for MAGs \citep{richardson2002ancestral} to $\sigma$-MAGs and demonstrate how it encodes $\sigma$-separation (See Definition~\ref{sigma separation} for more details) in the represented DMGs.

\begin{definition}[$m$-separation for $\sigma$-MAGs]
Let $\mathcal{H}$ be a $\sigma$-MAG with nodes $\mathcal{V}$. Let $X,Y,Z\subseteq  \mathcal{V}$ be subsets of the nodes.
\begin{enumerate}
    \item Consider a walk (path) $\pi$ in $\mathcal{H}$ with $n\geq 0$ edges:
      $$v_0\sus \cdots\sus v_n.$$
    We say that $\pi$ is \emph{$Z$-$m$-blocked} or \emph{$m$-blocked by $Z$} if any of the following conditions holds:
    \begin{enumerate}
        \item there exists a non-collider $v_k$ on $\pi$ in $Z$, or
        \item there exists a collider $v_k$ on $\pi$ that is not in $\anc_\mathcal{H}(Z)$, or
        \item there exists a subwalk (subpath) of the form \( v_{k-1} \suh v_k - v_{k+1} \) or \( v_{k-1} - v_k \hus  v_{k+1} \) on $\pi$.
    \end{enumerate}
    Conversely, we say that $\pi$ is \emph{$Z$-$m$-open} or \emph{$m$-open given $Z$} if:
    \begin{enumerate}
        \item all non-colliders on $\pi$ are not in $Z$, and
        \item all colliders on $\pi$ are in $\anc_\mathcal{H}(Z)$, and
        \item it contains no subwalk (subpath) of the form \( v_{k-1} \suh v_k - v_{k+1} \) or \( v_{k-1} - v_k \hus  v_{k+1} \).
    \end{enumerate}
    \item We say that \emph{$X$ is $m$-separated from $Y$ given $Z$} if every walk in $\mathcal{H}$ from a node in $X$ to a node in $Y$ is $m$-blocked by $Z$. This is denoted as:
    $$
    X\overset{m}{\underset{\mathcal{H}}{\perp}} Y\mid Z.
    $$
    If this property does not hold, we will write
    $$
    X\overset{m}{\underset{\mathcal{H}}{\not\perp}} Y\mid Z.
    $$
\end{enumerate}
\end{definition}

The following result is frequently employed to simplify proofs and to make the verification of \( m \)-separation more tractable in practice.

\begin{restatable}{proposition}{walkpath} \label{walk eq path}
Let $\mathcal{H}$ be a $\sigma$-MAG. For $Z\subseteq \mathcal{V}$, and $a,b\in\mathcal{V}$, the following are equivalent:
\begin{enumerate}
    \item there exists a $Z$-$m$-open walk between $a$ and $b$ in $\mathcal{H}$;
    \item there exists a $Z$-$m$-open path between $a$ and $b$ in $\mathcal{H}$.
\end{enumerate}
\end{restatable}

Lemmas~\ref{lemma17a} to \ref{lemma18} establish the relation between $m$-separation in $\sigma$-MAGs and $\sigma$-separation in DMGs, as stated in Theorem~\ref{thm2}.

\begin{restatable}{lemma}{lemmaseventeena} \label{lemma17a}
Let $\mathcal{G}$ be a DMG with nodes $\mathcal{V}^+=\mathcal{V}\cup \mathcal{S}$, and let $\mathcal{H}$ be a $\sigma$-MAG that represents $\mathcal{G}$ given $\mathcal{S}$. Let $a,b\in  \mathcal{V}$ and $Z\subseteq \mathcal{V}$. If $\pi$ is a $m$-open path given $Z$ between $a$ and $b$ in $\mathcal{H}$, then every node on $\pi$ is in $\anc_\mathcal{G}(\{a,b\}\cup Z\cup \mathcal{S})$.
\end{restatable}

\begin{restatable}{lemma}{lemmaseventeenb} \label{lemma17b}
Let $\mathcal{G}$ be a DMG with nodes $\mathcal{V}^+=\mathcal{V}\cup \mathcal{S}$, and let $\mathcal{H}$ be a $\sigma$-MAG that represents $\mathcal{G}$ given $\mathcal{S}$. Let $a,b\in  \mathcal{V}$ and $Z\subseteq  \mathcal{V}$. If there exists a $Z$-$m$-open path between $a$ and $b$ in $\mathcal{H}$, then there exists a $(Z\cup \mathcal{S})$-$\sigma$-open path in $\mathcal{G}$ between $a$ and $b$.
\end{restatable}

\begin{restatable}{lemma}{lemmaeighteen} \label{lemma18}
Let $\mathcal{G}$ be a DMG with nodes $\mathcal{V}^+=\mathcal{V}\cup \mathcal{S}$, and let $\mathcal{H}$ be a $\sigma$-MAG that represents $\mathcal{G}$ given $\mathcal{S}$. Let $a,b\in \mathcal{V}$ and $Z\subseteq \mathcal{V}$. If there exists a $(Z\cup \mathcal{S})$-$\sigma$-open path between $a$ and $b$ in $\mathcal{G}$, then there exists a $Z$-$m$-open path in $\mathcal{H}$ between $a$ and $b$.
\end{restatable}

For proving the latter two lemmata, we drew inspiration from the corresponding proofs for the corresponding statements for MAGs by \citet[Lemma 17--18]{spirtes1996polynomial}.
Combining these yields the following fundamental result.

\begin{restatable}{theorem}{thmtwo} \label{thm2}
Let $\mathcal{G}$ be a DMG with nodes $\mathcal{V}^+=\mathcal{V}\cup \mathcal{S}$, and let $\mathcal{H}$ be a $\sigma$-MAG that represents $\mathcal{G}$ given $\mathcal{S}$. Let $X,Y,Z\subseteq  \mathcal{V}$ be subsets of the nodes. We have the following equivalence:
$$
X\overset{m}{\underset{\mathcal{H}}{\perp}} Y\mid Z\qquad\iff\qquad X\overset{\sigma}{\underset{\mathcal{G}}{\perp}} Y\mid Z\cup \mathcal{S}.
$$
\end{restatable}

The following proposition establishes fundamental properties of \( m \)-separation and forms the basis for the concept of inducing paths: specifically, two distinct nodes in a \( \sigma \)-MAG are connected by an inducing path if and only if they cannot be \( m \)-separated by any subset of the remaining nodes.

\begin{restatable}{proposition}{inducingPropertyOne}\label{inducing property one}
Let $\mathcal{H}$ be a $\sigma$-MAG with nodes $\mathcal{V}$ and let $a,b\in \mathcal{V}$ be distinct nodes. Then the following are equivalent:
\begin{enumerate}
    \item There is an inducing path in $\mathcal{H}$ between $a$ and $b$;
    \item There is an inducing walk in $\mathcal{H}$ between $a$ and $b$;
    \item $a\overset{m}{\underset{\mathcal{H}}{\not\perp}} b\mid Z$ for all $Z\subseteq  \mathcal{V}\backslash\{a,b\}$;
    \item $a\overset{m}{\underset{\mathcal{H}}{\not\perp}} b\mid Z$ for $Z=\anc_\mathcal{H}(\{a,b\} )\backslash\{a,b\}$.
\end{enumerate}
\end{restatable}

The orientations of the outermost edges on an inducing path encode essential information about ancestral or anterior relationships.

\begin{restatable}{lemma}{inducingPropertyTwo}
Let $\mathcal{H}$ be a $\sigma$-MAG with nodes $\mathcal{V}$ and let $a,b\in  \mathcal{H}$ be distinct. If there exists an inducing path between $a$ and $b$ in $\mathcal{H}$, and all inducing paths in $\mathcal{H}$ between $a$ and $b$ are out of $b$, then $b\in\ant_\mathcal{H}(a)$.
\end{restatable}

\begin{restatable}{lemma}{inducingPropertyThree}
Let $\mathcal{H}$ be a $\sigma$-MAG with nodes $\mathcal{V}$ and let $a,b\in\mathcal{V}$ be distinct. If there exists an inducing path between $a$ and $b$ in $\mathcal{H}$ that is into $b$, and $a\notin\anc_\mathcal{H}(b)$, then there exists an inducing path between $a$ and $b$ in $\mathcal{H}$ that is both into $a$ and into $b$.
\end{restatable}

\section{Markov Equivalence} \label{section:3}
With the well-defined notion of $m$-separation, we now turn our attention to the concept of the \emph{$m$-Markov Equivalence Class} of $\sigma$-MAGs, in which every $\sigma$-MAG shares the same $m$-separation relations.

\begin{definition} [$m$-Markov Equivalence]
Two $\sigma$-MAGs $\mathcal{H}_1,\mathcal{H}_2$ with the same nodes $\mathcal{V}$ are \emph{$m$-Markov equivalent} if for any three subsets of the nodes $X,Y,Z\subseteq \mathcal{V}$, $X$ is $m$-separated from $Y$ given $Z$ in $\mathcal{H}_1$ if and only if $X$ is $m$-separated from $Y$ given $Z$ in $\mathcal{H}_2$.
\end{definition}

One important feature that will play a role in the characterization of $m$-Markov equivalence involves the concept of \emph{discriminating paths}, which can be seen as a generalization of an unshielded triple.

\begin{definition} [Discriminating Path]
A path $\pi=(a,v_0,\ldots,v_n,b,c)$ (with $n\geq0$) in a $\sigma$-MAG $\mathcal{H}$ is a \emph{discriminating path} for $b$ if:
\begin{enumerate}
    \item $a$ is not adjacent to $c$ in $\mathcal{H}$, and
    \item for $k=0,\ldots,n$: $v_k$ is a collider on $\pi$ and a parent of $c$ in $\mathcal{H}$.
\end{enumerate}
\end{definition}

An example of discriminating path is given in Figure~\ref{fig:discriminating_path}.

\begin{figure}[t]
\centering
  \scalebox{0.8}{\begin{tikzpicture}
    \begin{scope}
      \node[ndout] (A) at (0,0) {$a$};
      \node[ndout] (Q1) at (1.5,0) {$v_0$};
      \node[ndout] (Q2) at (3,0) {$v_1$};
      \node (dots) at (4.5,0) {$\cdots$};
      \node[ndout] (Qn) at (6,0) {$v_n$};
      \node[ndout] (B) at (7.5,0) {$b$};
      \node[ndout] (C) at (7.5,1.5) {$c$};
      \draw[suh] (A) -- (Q1);
      \draw[huh] (Q1) -- (Q2);
      \draw[huh] (Q2) -- (dots);
      \draw[huh] (dots) -- (Qn);
      \draw[tuh,in=-165] (Q1) edge (C);
      \draw[tuh,in=-150] (Q2) edge (C);
      \draw[tuh,in=-135] (Qn) edge (C);
      \draw[suh]  (B) -- (Qn);
      \draw[sus] (B) -- (C);
    \end{scope}
  \end{tikzpicture}}
  \caption{A discriminating path for $b$ between $a$ and $c$.\label{fig:discriminating_path}}
\end{figure}
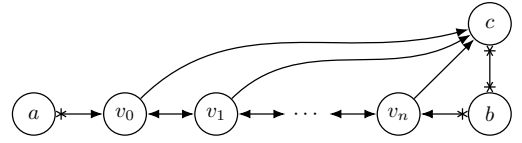

Applying $m$-separation to discriminating paths leads to the following lemma:

\begin{restatable}{lemma}{lemmasix} \label{lemma6}
In a $\sigma$-MAG $\mathcal{H}$ with nodes $\mathcal{V}$, let $\pi=(a,v_0,\ldots,v_n,b,c)$ be a discriminating path for $b$. Then, the following hold:
\begin{enumerate}
    \item If $b$ is a collider on $\pi$, then for any subset of nodes $Z \subseteq \mathcal{V}\backslash\{a,c\}$ such that $a$ and $c$ are $m$-separated given $Z$, we have  $b \notin Z$.
    \item If $b$ is a non-collider on $\pi$, then for any subset of nodes $Z \subseteq \mathcal{V}\backslash\{a,c\}$ such that $a$ and $c$ are $m$-separated given $Z$, we have $b \in Z$.
\end{enumerate}
\end{restatable}

\citet{spirtes1996polynomial} showed that two MAGs having the ``same basic colliders'' is a necessary and sufficient condition for ``d-Markov equivalence''. Here, we aim to verify whether an analogous statement applies to $\sigma$-MAGs. The assumption is restated as follows:

\begin{condition} \label{condition}
For two $\sigma$-MAGs $\mathcal{H}_1,\mathcal{H}_2$  with the same nodes $\mathcal{V}$, the following three conditions hold:
\begin{enumerate}
    \item $\mathcal{H}_1,\mathcal{H}_2$ have the same adjacencies.
    \item $\mathcal{H}_1,\mathcal{H}_2$ have the same unshielded colliders.
    \item Let \( \pi \) be a discriminating path for a node \( v \) in \( \mathcal{H}_1 \), and let \( \pi' \) be the corresponding path to \( \pi \) in \( \mathcal{H}_2 \).\footnote{Since \( \mathcal{H}_1 \) and \( \mathcal{H}_2 \) have the same nodes and adjacencies, the sequence of nodes in \( \pi \) within \( \mathcal{H}_1 \) uniquely determines a corresponding path \( \pi' \) in \( \mathcal{H}_2 \).} If \( \pi' \) is also a discriminating path for \( v \), then \( v \) is a collider on \( \pi \) in \( \mathcal{H}_1 \) if and only if it is a collider on \( \pi' \) in \( \mathcal{H}_2 \).
\end{enumerate}
\end{condition}
In the rest of this section, we mimic the proof strategy of \citet{spirtes1996polynomial} (with small modifications at various places to account for the extension of MAGs to $\sigma$-MAGs) to show that this assumption is a necessary and sufficient condition for $m$-Markov equivalence.

It is straightforward that Condition~\ref{condition} is a necessary condition to obtain $m$-Markov equivalence.

\begin{restatable}{lemma}{thmthreea} \label{thm3a}
Let $\mathcal{H}_1$ and $\mathcal{H}_2$ be two $\sigma$-MAGs with the same node set $\mathcal{V}$. If $\mathcal{H}_1$ and $\mathcal{H}_2$ are $m$-Markov equivalent, then they satisfy Condition~\ref{condition}.
\end{restatable}

Now, we consider the opposite direction and begin with the definition of \emph{covered nodes}, as follows:

\begin{definition} [Covered Node]  
In a \( \sigma \)-MAG \( \mathcal{H} \), a node \( v_k \) on a path \( \pi \) between \( v_0 \) and \( v_n \) (where \( 0 < k < n \)) is called a \emph{covered node} if its adjacent nodes on \( \pi \), \( v_{k-1} \) and \( v_{k+1} \), are also adjacent in \( \mathcal{H} \). Conversely, \( v_k \) is called an \emph{uncovered node} if \( v_{k-1} \) and \( v_{k+1} \) are non-adjacent in \( \mathcal{H} \).  
\end{definition}

The following proposition establishes an important property: every covered node on a shortest $m$-open path corresponds to a unique discriminating subpath. This result is built upon Lemma~\ref{lemma7} and Lemma~\ref{lemma8}, which provide the necessary foundation.

\begin{restatable}{proposition}{lemmanine} \label{lemma9}  
Let \( \pi \) be a shortest $m$-open path between \( v_0 \) and \( v_n \) given \( Z \) in a \( \sigma \)-MAG \( \mathcal{H} \). If \( v_j \) is a covered node on \( \pi \), then:  
\begin{enumerate}  
    \item If \( v_{j-1} \rightarrow v_{j+1} \), there exists a unique index \( i < j \) such that the subpath of $\pi$ between $v_i$ and $v_{j+1}$ is a discriminating path for \( v_j \).  
    \item If \( v_{j-1} \leftarrow v_{j+1} \), there exists a unique index \( i > j \) such that the subpath of $\pi$ between $v_{j-1}$ and $v_i$ is a discriminating path for \( v_j \).  
\end{enumerate}  
\end{restatable}  

\begin{restatable}{lemma}{lemmaseven} \label{lemma7}
Let $\pi$ be a shortest $m$-open path between \( v_0 \) and \( v_n \) given \( Z \) in a \( \sigma \)-MAG \( \mathcal{H} \). Suppose there exists an edge \( v_i \sus  v_j \) ($i<j$) in \( \mathcal{H} \) that is not part of \( \pi \). Define \( \pi' \) as the path obtained by replacing the subpath between \( v_i \) and \( v_j \) on \( \pi \) with the edge \( v_i \sus  v_j \). Then, one of the following conditions must hold:
\begin{enumerate}
    \item \( v_{i-1}\suh v_i \hus  v_{i+1} \) appears on \( \pi \) and \( v_i \rightarrow v_j \) exists in \( \mathcal{H} \).
    \item \( v_{j-1} \suh v_j\hus v_{j+1}\) appears on \( \pi \) and \( v_i \leftarrow v_j \) exists in \( \mathcal{H} \).
\end{enumerate}
\end{restatable}

\begin{restatable}{lemma}{lemmaeight} \label{lemma8}
Let \( \pi \) be a shortest $m$-open path between \( v_0 \) and \( v_n \) given \( Z \) in a \( \sigma \)-MAG \( \mathcal{H} \). Suppose \( \pi \) contains the subpath \( v_{k-1} \sus  v_k \sus  v_{k+1} \), where \( v_{k-1} \) and \( v_{k+1} \) are adjacent in \( \mathcal{H} \). Then, \( \mathcal{H} \) must contain one of the subgraphs shown in Figure~\ref{fig:lemma8}.
\end{restatable}

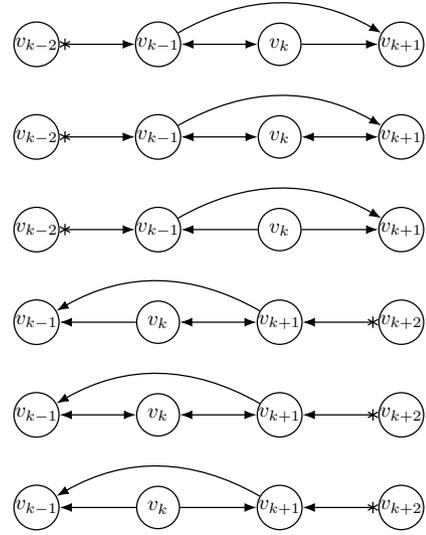
\begin{figure}[t]
    \centering
        \scalebox{0.8}{\begin{tikzpicture}
            \begin{scope}
            \node[ndout] (vk-2) at (-4,0) {$v_{k-2}$};
            \node[ndout] (vk-1) at (-2,0) {$v_{k-1}$};
            \node[ndout] (vk) at (0,0) {$v_k$};
            \node[ndout] (vk+1) at (2,0) {$v_{k+1}$};
            \draw[suh] (vk-2) edge (vk-1);
            \draw[huh] (vk-1) edge (vk);
            \draw[tuh] (vk) edge (vk+1);
            \draw[tuh, bend left] (vk-1) edge (vk+1);
            \end{scope}
        \end{tikzpicture}}

        \medskip
        \scalebox{0.8}{\begin{tikzpicture}
            \begin{scope}
            \node[ndout] (vk-2) at (-4,0) {$v_{k-2}$};
            \node[ndout] (vk-1) at (-2,0) {$v_{k-1}$};
            \node[ndout] (vk) at (0,0) {$v_k$};
            \node[ndout] (vk+1) at (2,0) {$v_{k+1}$};
            \draw[suh] (vk-2) edge (vk-1);
            \draw[huh] (vk-1) edge (vk);
            \draw[huh] (vk) edge (vk+1);
            \draw[tuh, bend left] (vk-1) edge (vk+1);
            \end{scope}
        \end{tikzpicture}}

        \medskip 
        \scalebox{0.8}{\begin{tikzpicture}
            \begin{scope}
            \node[ndout] (vk-2) at (-4,0) {$v_{k-2}$};
            \node[ndout] (vk-1) at (-2,0) {$v_{k-1}$};
            \node[ndout] (vk) at (0,0) {$v_k$};
            \node[ndout] (vk+1) at (2,0) {$v_{k+1}$};
            \draw[suh] (vk-2) edge (vk-1);
            \draw[hut] (vk-1) edge (vk);
            \draw[tuh] (vk) edge (vk+1);
            \draw[tuh, bend left] (vk-1) edge (vk+1);
            \end{scope}
        \end{tikzpicture}}

        \medskip
        \scalebox{0.8}{\begin{tikzpicture}
            \begin{scope}
            \node[ndout] (vk-1) at (-4,0) {$v_{k-1}$};
            \node[ndout] (vk) at (-2,0) {$v_{k}$};
            \node[ndout] (vk+1) at (0,0) {$v_{k+1}$};
            \node[ndout] (vk+2) at (2,0) {$v_{k+2}$};
            \draw[hut] (vk-1) edge (vk);
            \draw[huh] (vk) edge (vk+1);
            \draw[hus] (vk+1) edge (vk+2);
            \draw[hut, bend left] (vk-1) edge (vk+1);
            \end{scope}
        \end{tikzpicture}}

        \medskip
        \scalebox{0.8}{\begin{tikzpicture}
            \begin{scope}
            \node[ndout] (vk-1) at (-4,0) {$v_{k-1}$};
            \node[ndout] (vk) at (-2,0) {$v_{k}$};
            \node[ndout] (vk+1) at (0,0) {$v_{k+1}$};
            \node[ndout] (vk+2) at (2,0) {$v_{k+2}$};
            \draw[huh] (vk-1) edge (vk);
            \draw[huh] (vk) edge (vk+1);
            \draw[hus] (vk+1) edge (vk+2);
            \draw[hut, bend left] (vk-1) edge (vk+1);
            \end{scope}
        \end{tikzpicture}}

        \medskip
        \scalebox{0.8}{\begin{tikzpicture}
            \begin{scope}
            \node[ndout] (vk-1) at (-4,0) {$v_{k-1}$};
            \node[ndout] (vk) at (-2,0) {$v_{k}$};
            \node[ndout] (vk+1) at (0,0) {$v_{k+1}$};
            \node[ndout] (vk+2) at (2,0) {$v_{k+2}$};
            \draw[hut] (vk-1) edge (vk);
            \draw[tuh] (vk) edge (vk+1);
            \draw[hus] (vk+1) edge (vk+2);
            \draw[hut, bend left] (vk-1) edge (vk+1);
            \end{scope}
        \end{tikzpicture}}
  \caption{Possible subgraphs for Lemma~\ref{lemma8}.}
  \label{fig:lemma8}
\end{figure}

Next, from Lemma~\ref{lemma10}, we can infer Lemmas~\ref{lemma11} to \ref{lemma13}, which collectively provide a foundation for defining the order of covered nodes.

\begin{restatable}{lemma}{lemmaten} \label{lemma10}
In a \( \sigma \)-MAG $\mathcal{H}$, if \( \pi \) is a shortest $m$-open path between \( v_0 \) and \( v_n \) given \( Z \), then no pair of distinct covered nodes \( v_b \) and \( v_j \) on $\pi$ can satisfy both of the following conditions: \( v_b \) is a covered node on the discriminating subpath of \( \pi \) for \( v_j \), and \( v_j \) is a covered node on the discriminating subpath of \( \pi \) for \( v_b \).  
\end{restatable}

\begin{restatable}{lemma}{lemmaeleven} \label{lemma11}
In a $\sigma$-MAG $\mathcal{H}$, if $\pi$ is a shortest $m$-open path between $v_0$ and $v_n$ given $Z$, then no triple of distinct covered nodes $v_i,v_j,v_k$ on $\pi$ can satisfy all of the following conditions: $v_i$ is a covered node on the discriminating subpath of $\pi$ for $v_j$, $v_j$ is a covered node on the discriminating subpath of $\pi$ for $v_k$, and $v_k$ lies between $v_i$ and $v_j$ on $\pi$.
\end{restatable}

\begin{restatable}{lemma}{lemmatwelve} \label{lemma12}
In a \( \sigma \)-MAG \( \mathcal{H} \), if \( \pi \) is a shortest $m$-open path between \( v_0 \) and \( v_n \) given \( Z \), then no quadruple of distinct covered nodes \( v_i, v_j, v_k, v_l \) on $\pi$ can satisfy all of the following conditions: \( v_i \) is a covered node on the discriminating subpath of \( \pi \) for \( v_j \), \( v_k \) is a covered node on the discriminating subpath of \( \pi \) for \( v_l \), $v_l$ lies between $v_i$ and $v_j$ on $\pi$, and $v_j$ lies between $v_l$ and $v_k$ on $\pi$.
\end{restatable}

\begin{restatable}{lemma}{lemmathirteen} \label{lemma13}
In a \( \sigma \)-MAG \( \mathcal{H} \), if \( \pi \) is a shortest $m$-open path between \( a \) and \( b \) given \( Z \), then no sequence of distinct covered nodes \( v_0, v_1, \ldots, v_n \) on $\pi$ with \( n \geq 2 \) can satisfy the following conditions: for each \( i = 0, \ldots, n-1 \), the node \( v_i \) is a covered node on the discriminating subpath of \( \pi \) for \( v_{i+1} \), and \( v_n \) is a covered node on the discriminating subpath of \( \pi \) for \( v_0 \). 
\end{restatable}

Equipped with the above lemmata, we now introduce the \emph{order} of covered nodes on a shortest $m$-open path.  

\begin{definition} [Order of covered Node] \label{order of covered node}
In a \( \sigma \)-MAG \( \mathcal{H} \), let \( \pi \) be a shortest path between nodes $a$ and $b$ given a conditioning set $Z$. A covered node  \( v \) on \( \pi \) is called a \emph{0-order covered node} if the discriminating subpath of \( \pi \) for \( v \) contains no other covered nodes.  

Furthermore, a covered node \( v \) on \( \pi \) is called an \emph{$n$-order covered node} if the maximum order of any other covered node on the discriminating subpath of \( \pi \) for \( v \) is \( n-1 \).  
\end{definition}

Lemma~\ref{lemma13} guarantees that this recursive definition is well-founded. Given a shortest $m$-open path containing covered nodes in a \( \sigma \)-MAG, starting from an arbitrary covered node \( v_0 \), if there is no other covered node on the discriminating subpath of \( \pi \) for \( v_0 \), then \( v_0 \) is a 0-order covered node. If another covered node \( v_1 \) exists on this subpath, we proceed by checking whether there is a covered node on the discriminating subpath of \( \pi \) for \( v_1 \). Repeating this process, we must eventually reach a covered node that has no other covered nodes on its discriminating subpath, since \( \pi \) has finite length and Lemma~\ref{lemma13} ensures that we do not encounter repetitive covered nodes in this process. Similarly, there must exist a covered node such that all covered nodes on its discriminating subpath of \( \pi \) are 0-order. By induction, the definition of $n$-order covered node is well-defined.

Now we will show that Condition~\ref{condition} is also a sufficient condition for $m$-Markov equivalence.

\begin{restatable}{lemma}{lemmafifteen} \label{lemma15}  
If two \( \sigma \)-MAGs \( \mathcal{H}_1 \) and \( \mathcal{H}_2 \) with the same nodes $\mathcal{V}$ satisfy Condition~\ref{condition}, and if \( \pi=(a,v_0,\ldots,v_n,b,c) \) is a discriminating path for $b$ in \( \mathcal{H}_1 \), then let \( \pi' \) be the corresponding path to \( \pi \) in \( \mathcal{H}_2 \). If every node on \( \pi' \), except for the endpoints and \( b \), is a collider, then \( \pi' \) is a discriminating path for \( b \) in \( \mathcal{H}_2 \). 
\end{restatable}

Equipped with Definition~\ref{order of covered node}, we now show that every node on a given shortest \( m \)-open path in one \( \sigma \)-MAG retains the same collider/non-collider status on the corresponding path in another \( \sigma \)-MAG, provided that the two graphs satisfy Condition~\ref{condition}.

\begin{restatable}{lemma}{lemmasixteen} \label{lemma16}  
If two \( \sigma \)-MAGs \( \mathcal{H}_1 \) and \( \mathcal{H}_2 \) with the same nodes \( \mathcal{V} \) satisfy Condition~\ref{condition}, and if \( \pi \) is a shortest $m$-open path between \( v_0 \) and \( v_n \) given \( Z \) in \( \mathcal{H}_1 \), with \( \pi' \) being the corresponding path to \( \pi \) in \( \mathcal{H}_2 \), then \( v_k \) is a collider on \( \pi \) if and only if \( v_k \) is a collider on \( \pi' \).  
\end{restatable}

The following lemma establishes useful properties for a special case involving covered nodes.

\begin{restatable}{lemma}{lemmanineteen} \label{lemma19}  
If a \( \sigma \)-MAG \( \mathcal{H} \) contains a path \( a \suh b\tuh c \) and an edge \( a \sus c \), then:  
\begin{enumerate}  
    \item The edge between $a$ and $c$ is oriented as \( a \suh c \).  
    \item If \( a \sus c \) has a different edge mark at \( a \) than \( a \suh b \), then the edges are oriented as \( a \huh b\tuh c \) and \( a \tuh c \).  
\end{enumerate}  
\end{restatable}

By Lemma~\ref{lemma16}, all non-colliders on a shortest \( m \)-open path in one \( \sigma \)-MAG cannot block the corresponding path in another \( \sigma \)-MAG, provided that the two graphs satisfy Condition~\ref{condition}. Lemmas~\ref{lemma20} and~\ref{lemma21} further show that all colliders do not block the path either, thereby establishing the \( m \)-Markov equivalence of the two graphs.

\begin{restatable}{lemma}{lemmatwenty} \label{lemma20}
Let \( \mathcal{H}_1 \) and \( \mathcal{H}_2 \) be two \( \sigma \)-MAGs with the same node set \( \mathcal{V} \) that satisfy Condition~\ref{condition}. Suppose \( \pi \) is a shortest $m$-open path between \( v_0 \) and \( v_n \) given \( Z \) in \( \mathcal{H}_1 \) and has the smallest collider distance sum to \( Z \).\footnote{That is, there does not exist another $m$-open path \( \mu \) between \( v_0 \) and \( v_n \) given \( Z \) in \( \mathcal{H}_1 \) such that either \( \mu \) has fewer edges than \( \pi \), or \( \mu \) has the same number of edges as \( \pi \), but a smaller collider distance sum to \( Z \).} If \( v_i \) is a collider on \( \pi \), \( q \) is an ancestor of \( Z \) in \( \mathcal{H}_1 \), and there is an edge \( v_i \rightarrow q \) on a shortest directed path from \( v_i \) to \( Z \) in \( \mathcal{H}_1 \), then the edge \( v_i \rightarrow q \) must also be present in \( \mathcal{H}_2 \).
\end{restatable}

\begin{restatable}{lemma}{lemmatwentyone} \label{lemma21}
Let \( \mathcal{H}_1 \) and \( \mathcal{H}_2 \) be two \( \sigma \)-MAGs with the same node set \( \mathcal{V} \) that satisfy Condition~\ref{condition}. Suppose \( \pi \) is a shortest $m$-open path between \( v_0 \) and \( v_n \) given \( Z \) in \( \mathcal{H}_1 \) and has the smallest collider distance sum to \( Z \). If \( v_k \) is a collider on \( \pi \), and \( q \in Z \) is the endpoint of a shortest directed path \( \mu \) from \( v_k \) to \( Z \) in \( \mathcal{H}_1 \), then \( v_k \) is an ancestor of \( q \) in \( \mathcal{H}_2 \).
\end{restatable}

\begin{restatable}{lemma}{thmthreeb} \label{thm3b}
If two \( \sigma \)-MAGs \( \mathcal{H}_1, \mathcal{H}_2 \) with the same node set \( \mathcal{V} \) satisfy Condition~\ref{condition}, then they are $m$-Markov equivalent.
\end{restatable}

Consequently, we obtain the following equivalence:

\begin{restatable}{theorem}{thmthree} \label{thm3}
Two $\sigma$-MAGs $\mathcal{H}_1,\mathcal{H}_2$  with the same nodes $\mathcal{V}$ are $m$-Markov equivalent if and only if $\mathcal{H}_1,\mathcal{H}_2$ satisfy Condition~\ref{condition}.
\end{restatable}

Moreover, by combining it with the equivalence derived in Section~\ref{section:2}, we also obtain the $\sigma$-Markov equivalence under Condition~\ref{condition}.

\begin{restatable}{theorem}{thmfour} \label{thm4}
Let $\mathcal{G}_1,\mathcal{G}_2$ be two DMGs with the same nodes $\mathcal{V}^+=\mathcal{V}\cup \mathcal{S}$, and let $\mathcal{H}_1,\mathcal{H}_2$ be two $\sigma$-MAGs that represents $\mathcal{G}_1,\mathcal{G}_2$ given $\mathcal{S}$ respectively. $\mathcal{G}_1,\mathcal{G}_2$ are $\sigma$-Markov equivalent given $\mathcal{S}$ if and only if $\mathcal{H}_1,\mathcal{H}_2$ satisfy Condition~\ref{condition}.
\end{restatable}

\section{Discussion and Conclusion}

It is remarkable that with the same edge types as MAGs, $\sigma$-MAGs can account for both cycles and selection bias (via undirected edges) in such a way that the key properties of MAGs generalize in a predictable way to those of $\sigma$-MAGs, even though the undirected edges do not in general distinguish between cycles and selection bias (except in case of the presence of an edge into a clique of undirected edges).

The theory we developed here provides the foundations for possible future work on developing sound and complete extensions of FCI applicable to data generated by simple (possibly cyclic) SCMs in the presence of selection bias. Furthermore, it may prove an important step towards developing a do-calculus for the output of FCI in that general setting.

\begin{acknowledgements}
The presentation of this paper at the conference was financially supported by the Amsterdam ELLIS Unit.
We thank Tom Claassen for discussions that benefited this work.
\end{acknowledgements}

\bibliography{references}

\newpage

\onecolumn

\title{Supplementary Material to ``$\sigma$-Maximal Ancestral Graphs''}
\maketitle

\appendix

\section{Preliminaries}
\subsection{Directed Mixed Graphs (DMGs)}
\begin{definition} [Directed Mixed Graphs (DMGs)]
A \emph{Directed Mixed Graph (DMG)} is a graph $\mathcal{G}=({\mathcal{V}},\mathcal{E},\mathcal{L})$ with nodes ${\mathcal{V}}$ and two types of edges: \emph{directed edges} $\mathcal{E}\subseteq \{(v_i,v_j):v_i,v_j\in{\mathcal{V}},v_i\neq v_j\}$, and bidirected edges $\mathcal{L}\subseteq\{(v_i,v_j):v_i,v_j\in{\mathcal{V}},v_i\neq v_j\}$.
\end{definition}

We denote a directed edge \( (v_i, v_j) \in \mathcal{E} \) as \( v_i \rightarrow v_j \) or equivalently \( v_j \leftarrow v_i \). Self-loops of the form \( v_i \rightarrow v_i \) are not allowed, but multiple edges (with at most one of each type) between any pair of distinct nodes are permitted. Similarly, we denote a bidirected edge \( (v_i, v_j) \in \mathcal{L} \) as \( v_i \huh  v_j \) or equivalently \( v_j \huh  v_i \).  

Additionally, we use a star edge mark as a placeholder indicating either an arrowhead or a tail. Specifically,
\begin{enumerate}
\item \( v_i \suh v_j \) represents \( v_i \tuh v_j \) or \( v_i \huh  v_j \),
\item \( v_i \hus  v_j \) represents \( v_i \hut v_j \) or \( v_i \huh  v_j \),
\item \( v_i \sus  v_j \) represents \( v_i \tuh v_j \), \( v_i \leftarrow v_j \), or \( v_i \huh  v_j \).  
\end{enumerate}

\begin{definition}  
Let \( \mathcal{G} = (\mathcal{V}, \mathcal{E}, \mathcal{L}) \) be a DMG.  
\begin{enumerate}  
    \item If \( v_i \sus  v_j \in \mathcal{E} \cup \mathcal{L} \), then \( v_i \) and \( v_j \) are said to be \emph{adjacent} in \( \mathcal{G} \).  
    \item An edge of the form \( v_i \hus  v_j \) is said to be \emph{into} \( v_i \).  
    \item An edge of the form \( v_i \rightarrow v_j \) is said to be \emph{out of} \( v_i \).  
\end{enumerate}  
\end{definition}  

\begin{definition} [Walks]  
Let \( \mathcal{G} = (\mathcal{V},\mathcal{E},\mathcal{L}) \) be a DMG, and let \( v_0, v_n \in \mathcal{V} \).  
\begin{enumerate}  
    \item A \emph{walk} from \( v_0 \) to \( v_n \) in \( \mathcal{G} \) is a finite alternating sequence of adjacent nodes and edges,  
    \[
    \langle v_0, a_1, v_1, \ldots, v_{n-1}, a_n, v_n \rangle,
    \]  
    for some \( n \geq 0 \), such that for every \( k = 1, \ldots, n \), we have \( a_k = (v_{k-1}, v_k) \in \mathcal{E} \cup \mathcal{L} \). The \emph{trivial walk} consisting of a single node \( v_0 \) is also allowed. The walk is called \emph{into} \( v_0 \) if \( a_1 = v_0 \hus  v_1 \), and \emph{out of} \( v_0 \) if \( a_1 = v_0 \rightarrow v_1 \). Similarly, it is called \emph{into} \( v_n \) if \( a_n = v_{n-1} \suh v_n \) and \emph{out of} \( v_n \) if \( a_n = v_{n-1} \leftarrow v_n \).  

    \item A \emph{directed walk} from \( v_0 \) to \( v_n \) in \( \mathcal{G} \) is of the form  
    \[
    v_0 \rightarrow v_1 \rightarrow \cdots \rightarrow v_{n-1} \rightarrow v_n,
    \]  
    for some \( n \geq 0 \), where all arrowheads point toward \( v_n \), and no arrowheads point backward.  

    \item A walk is called a \emph{path} if no node appears more than once.  
\end{enumerate}  
\end{definition}  

\begin{definition} [Family Relationships] \label{family relationships}
Let \( \mathcal{G} = (\mathcal{V},\mathcal{E},\mathcal{L}) \) be a DMG, and let \( v \in \mathcal{V} \) and \( A \subseteq \mathcal{V} \) be a subset of nodes. We define:  
\begin{enumerate}  
    \item The set of \emph{parents} of \( \mathcal{V} \) in \( \mathcal{G} \):  
    \[
    \text{Pa}_\mathcal{G}(v) := \{ w \in \mathcal{V} \mid w \rightarrow v \in \mathcal{E} \}.
    \]
    The set of \emph{parents} of \( A \) in \( \mathcal{G} \):  
    \[
    \text{Pa}_\mathcal{G}(A) := \bigcup_{v \in A} \text{Pa}_\mathcal{G}(v).
    \]
    
    \item The set of \emph{children} of \( \mathcal{V} \) in \( \mathcal{G} \):  
    \[
    \text{Ch}_\mathcal{G}(v) := \{ w \in \mathcal{V} \mid v \rightarrow w \in \mathcal{E} \}.
    \]
    The set of \emph{children} of \( A \) in \( \mathcal{G} \):  
    \[
    \text{Ch}_\mathcal{G}(A) := \bigcup_{v \in A} \text{Ch}_\mathcal{G}(v).
    \]

    \item The set of \emph{siblings} of \( \mathcal{V} \) in \( \mathcal{G} \):  
    \[
    \text{Sib}_\mathcal{G}(v) := \{ w \in \mathcal{V} \mid v \huh  w \in \mathcal{L} \}.
    \]
    The set of \emph{siblings} of \( A \) in \( \mathcal{G} \):  
    \[
    \text{Sib}_\mathcal{G}(A) := \bigcup_{v \in A} \text{Sib}_\mathcal{G}(v).
    \]

    \item The set of \emph{ancestors} of \( \mathcal{V} \) in \( \mathcal{G} \):  
    \[
    \anc_\mathcal{G}(v) := \{ w \in \mathcal{V} \mid \exists \text{ a directed walk } w \rightarrow \cdots \rightarrow v \text{ in } \mathcal{G} \}.
    \]
    The set of \emph{ancestors} of \( A \) in \( \mathcal{G} \):  
    \[
    \anc_\mathcal{G}(A) := \bigcup_{v \in A} \anc_\mathcal{G}(v).
    \]

    \item The set of \emph{descendants} of \( \mathcal{V} \) in \( \mathcal{G} \):  
    \[
    \text{Desc}_\mathcal{G}(v) := \{ w \in \mathcal{V} \mid \exists \text{ a directed walk } v \rightarrow \cdots \rightarrow w \text{ in } \mathcal{G} \}.
    \]
    The set of \emph{descendants} of \( A \) in \( \mathcal{G} \):  
    \[
    \text{Desc}_\mathcal{G}(A) := \bigcup_{v \in A} \text{Desc}_\mathcal{G}(v).
    \]

    \item The \emph{strongly connected component} of \( \mathcal{V} \) in \( \mathcal{G} \):  
    \[
    \text{Sc}_\mathcal{G}(v) := \anc_\mathcal{G}(v) \cap \text{Desc}_\mathcal{G}(v).
    \]
    The \emph{strongly connected component} of \( A \) in \( \mathcal{G} \):  
    \[
    \text{Sc}_\mathcal{G}(A) := \bigcup_{v \in A} \text{Sc}_\mathcal{G}(v).
    \]
\end{enumerate}  
\end{definition}

\begin{definition}[Directed Graphs (DGs)]
A \emph{Directed Graph (DG)} is a DMG that contains no bidirected edges, i.e., $\mathcal{L} = \emptyset$.
\end{definition}

A DMG can be regarded as the result of marginalizing latent variables in a DG, as formalized below:

\begin{definition}[Marginalization on DMGs \citep{BFPM21}]\label{margilization}
Let $\mathcal{G} = (\mathcal{V}, \mathcal{E}, \mathcal{L})$ be a DMG and let $\mathcal{W} \subseteq \mathcal{V}$ be a subset of nodes. The \emph{marginalization} of $\mathcal{G}$ with respect to $\mathcal{W}$ is the DMG
\[
\mathcal{G}^{\mathcal{V} \setminus \mathcal{W}} := \mathcal{G}^{\setminus \mathcal{W}} := (\mathcal{V}^{\setminus \mathcal{W}}, \mathcal{E}^{\setminus \mathcal{W}}, \mathcal{L}^{\setminus \mathcal{W}}),
\]
where:
\begin{enumerate}
    \item $\mathcal{V}^{\setminus \mathcal{W}} = \mathcal{V} \setminus \mathcal{W}$;
    \item $\mathcal{E}^{\setminus \mathcal{W}}$ consists of all directed edges $a \tuh b$ with $a, b \in \mathcal{V} \setminus \mathcal{W}$ for which there exists a directed walk in $\mathcal{G}$ of the form
    \[
    a \tuh w_0 \tuh \cdots \tuh w_n \tuh b,
    \]
    where all intermediate nodes $w_0, \ldots, w_n \in \mathcal{W}$ (if any);
    \item $\mathcal{L}^{\setminus \mathcal{W}}$ consists of all bidirected edges $a \huh b$ with $a, b \in \mathcal{V} \setminus \mathcal{W}$, $a \neq b$, for which there exists a bifurcation in $\mathcal{G}$ of the form
    \[
    a \hut w_0 \hut \cdots \hut w_{k-1} \hus w_k \tuh \cdots \tuh w_n \tuh b,
    \]
    where all intermediate nodes $w_0, \ldots, w_n \in \mathcal{W}$ (if any).
\end{enumerate}
\end{definition}

\subsection{\texorpdfstring{$\sigma$-separation}{sigma-Separation}}
\begin{definition}[Colliders and Non-Colliders]  
Let \( \mathcal{G} = (\mathcal{V},\mathcal{E},\mathcal{L}) \) be a DMG, and let \( \pi \) be a walk between \( v_0 \) and \( v_n \) in \( \mathcal{G} \). A node \( v_k \) on \( \pi \) is classified as follows:  
\begin{enumerate}  
    \item A \emph{non-collider} on \( \pi \) if there is at most one arrowhead pointing towards \( v_k \), i.e., if \( v_k \) is an endpoint of \( \pi \) or of the form  
    \[
    v_{k-1} \sus  v_k \rightarrow v_{k+1} \quad \text{or} \quad v_{k-1} \leftarrow v_k \sus  v_{k+1}.
    \]
    
    \item A \emph{collider} on \( \pi \) if there are two arrowheads pointing towards \( v_k \), i.e., if \( v_k \) is of the form  
    \[
    v_{k-1} \suh v_k \hus  v_{k+1}.
    \]
\end{enumerate}  
\end{definition}

\begin{definition}[Blockable and Unblockable Non-Colliders]  
Let \( \mathcal{G} = (\mathcal{V},\mathcal{E},\mathcal{L}) \) be a DMG, and let \( \pi \) be a walk between \( v_0 \) and \( v_n \) in \( \mathcal{G} \). A non-collider \( v_k \) on \( \pi \) is called an \emph{unblockable non-collider} if it is not an endpoint and has only outgoing edges on \( \pi \) leading to nodes within the same strongly connected component of \( \mathcal{G} \). Specifically, \( v_k \) follows one of the patterns:
\begin{align*}  
    v_{k-1} \leftarrow v_k \hus  v_{k+1}, & \quad \text{where } v_{k-1} \in \text{Sc}_\mathcal{G}(v_k), \\  
    v_{k-1} \suh v_k \rightarrow v_{k+1}, & \quad \text{where } v_{k+1} \in \text{Sc}_\mathcal{G}(v_k), \\  
    v_{k-1} \leftarrow v_k \rightarrow v_{k+1}, & \quad \text{where } v_{k-1}, v_{k+1} \in \text{Sc}_\mathcal{G}(v_k).
\end{align*}  

Otherwise, \( v_k \) is called a \emph{blockable non-collider} on \( \pi \), meaning that either:
\begin{itemize}  
    \item \( v_k \) is an endpoint, or  
    \item \( v_k \) has at least one outgoing edge \( v_k \rightarrow v_{k\pm1} \) where \( v_{k\pm1} \) lies in a different strongly connected component than \( v_k \), i.e., \( v_{k\pm1} \notin \text{Sc}_\mathcal{G}(v_k) \).
\end{itemize}  
\end{definition}

\begin{definition}[$\sigma$-blocked Walks (Node Version)]  
Let \( \mathcal{G} = (\mathcal{V},\mathcal{E},\mathcal{L}) \) be a DMG, and let \( Z \subseteq \mathcal{V} \) be a subset of nodes. Consider a walk \( \pi \) between \( v_0 \) and \( v_n \) in \( \mathcal{G} \). We define:  
\begin{enumerate}  
    \item The walk \( \pi \) is \emph{\( Z \)-\( \sigma \)-open} (or \emph{\( \sigma \)-open given \( Z \)}) if and only if:  
    \begin{enumerate}  
        \item All colliders \( v_k \) on \( \pi \) belong to \( \anc_\mathcal{G}(Z) \), and  
        \item All blockable non-colliders \( v_k \) on \( \pi \) do not belong to \( Z \).  
    \end{enumerate}  
    \item The walk \( \pi \) is \emph{\( Z \)-\( \sigma \)-blocked} (or \emph{\( \sigma \)-blocked by \( Z \)}) if and only if:  
    \begin{enumerate}  
        \item There exists a collider \( v_k \) on \( \pi \) that is not in \( \anc_\mathcal{G}(Z) \), or
        \item There exists a blockable non-collider \( v_k \) on \( \pi \) that belongs to \( Z \).  
    \end{enumerate}  
\end{enumerate}  
\end{definition}

Instead of presenting the well-established concept of $d$-separation, originally introduced by \citet{pearl1985constraint}, we provide its non-trivial generalization, referred to as $\sigma$-separation.

\begin{definition}[$\sigma$-separation (Node Version) \citep{forre2017markov}] \label{sigma separation}
Let \( \mathcal{G} = (\mathcal{V},\mathcal{E},\mathcal{L}) \) be a DMG, and let \( X, Y, Z \subseteq \mathcal{V} \) be subsets of nodes. We define:  
\begin{enumerate}  
    \item \( X \) is \emph{\(\sigma\)-separated from \( Y \) given \( Z \) in \( \mathcal{G} \)}, denoted as  
    \[
    X\overset{\sigma}{\underset{\mathcal{G}}{\perp}}Y \mid Z,
    \]  
    if every walk from a node in \( X \) to a node in \( Y \) is \( Z \)-\( \sigma \)-blocked.  
    \item If this condition does not hold, we write:  
    \[
    X\not\overset{\sigma}{\underset{\mathcal{G}}{\perp}}Y \mid Z.
    \]  
\end{enumerate}  
\end{definition}

\begin{definition}[$\sigma$-separation (Segment Version) \citep{forre2017markov}] \label{segmentsigma}
Let $\mathcal{G}=(\mathcal{V},\mathcal{E},\mathcal{L})$ be a DMG and $X,Y,Z\subseteq \mathcal{V}$ subsets of the nodes.
\begin{enumerate}
    \item Consider a path in $\mathcal{G}$ with $n\geq1$ nodes:
    $$
    v_0\sus\cdots\sus v_n.
    $$
    Then this path can be uniquely partitioned according to the strongly connected components of $\mathcal{G}$:
    $$
    v_{i-1}\sus v_i\sus v_{i+1}\sus\cdots\sus v_{j-1}\sus v_j\sus v_{j+1},
    $$
    with $v_i,\ldots, v_j\in \text{Sc}_\mathcal{G}(v_i)$ and $v_{i-1},v_{j+1}\notin \text{Sc}_\mathcal{G}(v_i)$. Note that $v_{i-1}$ or $v_{j+1}$ might not appear if $v_i$ or $v_j$ is an endpoint of the path. We will call the subpath $\sigma_k$ (given by the nodes $v_i,\ldots,v_j$ and its corresponding edges) a segment of the path. We abbreviate the left an right endpoint of $\sigma_k$ with $\sigma_{k,l}=v_i$ and $\sigma_{k,r}=v_j$. The path can then uniquely be written with its segments:
    $$
    \sigma_1\sus\cdots\sus\sigma_m.
    $$
    We will call $\sigma_1$ and $\sigma_m$ the end-segments of the path.
    \item Such a path will be called \emph{Z-$\sigma$-blocked} or \emph{$\sigma$-blocked by $Z$} if:
    \begin{enumerate}
        \item at least one of the endpoints $v_0=\sigma_{1,l}$, $v_n=\sigma_{m,r}$ is in $Z$, or 
        \item there is a segment $\sigma_k$ with an outgoing directed edge in the path and its corresponding endpoint lies in $Z$, or
        \item there is a segment $\sigma_k$ with two adjacent edges that form a collider $\suh\sigma_k\hus$ and $\text{Sc}_\mathcal{G}(\sigma_k)\cap\anc_\mathcal{G}(Z)=\emptyset$.
    \end{enumerate}
    If none of the above holds then the path is called \emph{Z-$\sigma$-open} or \emph{$\sigma$-open given $Z$}.
    \item We say that $X$ is $\sigma$-separated from $Y$ given $Z$ if every path in $\mathcal{G}$ with one endpoint in $X$ and one endpoint in $Y$ is $\sigma$-blocked by $Z$. In symbols this will be written as follows:
    $$
    X\overset{\sigma}{\underset{\mathcal{G}}{\perp}} Y\mid Z.
    $$
\end{enumerate}
\end{definition}

\begin{lemma} \label{lemma3.3.10}
Let $\mathcal{G}=(\mathcal{V},\mathcal{E},\mathcal{L})$ be a DMG, $Z\subseteq \mathcal{V}$, and $\pi$ is a $Z$-$\sigma$-open walk between $v_0$ and $v_n$ in $\mathcal{G}$. Suppose $v_i\in\text{Sc}_\mathcal{G}(v_j)$ for some $i,j\in\{0,\ldots,n\}$ with $i<j$. If we then replace the subwalk $v_i\sus\cdots\sus v_j$ of $\pi$ by
\begin{enumerate}
    \item a shortest directed path $v_i\tuh\cdots\tuh v_j$ in $\mathcal{G}$ if $j=n$ or if $v_j\tuh v_{j+1}$ on $\pi$, or
    \item a shortest directed path $v_i\hut\cdots\hut v_j$ in $\mathcal{G}$ otherwise,
\end{enumerate}
then this new subwalk is entirely within $\text{Sc}_\mathcal{G}(v_j)$ and the modified walk $\pi'$ is still $Z$-$\sigma$-open.
\end{lemma}
\begin{proof}
$\pi^{\prime}$ cannot become $Z$-$\sigma$-blocked at one of the initial nodes $v_{0},\ldots,v_{i-1}$ or at one of the final nodes $v_{j+1},\ldots,v_{n}$ on $\pi^{\prime}$, since these nodes occur in the same local configuration on $\pi$ and are not $Z$-$\sigma$-blocked on $\pi$ by assumption. Furthermore, $\pi^{\prime}$ cannot become $Z$-$\sigma$-blocked at one of the nodes strictly between $v_{i}$ and $v_{j}$ on $\pi^{\prime}$ (if there are any), since these nodes are all non-endnode non-colliders that only point to nodes in the same strongly connected component $\mathrm{Sc}_\mathcal{G}(v_{j})$. It is also worth noting that $\pi^{\prime}$ cannot become $Z$-$\sigma$-blocked at any of its endnodes, which could be $v_{i}$ or $v_{j}$ or both, because those are the same in $\pi$. So in the following we can w.l.o.g. assume that both $v_{i}$ and $v_{j}$ are non-endnodes of $\pi$ and thus $\pi^{\prime}$.

Case 1: By assumption $v_{j}$ is in the subwalk $v_{j-1}\sus v_j\tuh v_{j+1}$ (or the right endnode) on $\pi$ that is $Z$-$\sigma$-open. Since the same blocking criteria apply to $v_{j}$ on $\pi^{\prime}$ it remains $Z$-$\sigma$-open on $\pi^{\prime}$. If $v_{i}=v_{j}$ then also $v_{i}$ is $Z$-$\sigma$-open on $\pi^{\prime}$ (if $v_{i}$ is the left endnode or not). If $v_{i}\neq v_{j}$, then the new directed path $v_{i}\tuh\cdots\tuh v_{j}$ in $\pi^{\prime}$ is $Z$-$\sigma$-open at $v_{i}$ because all nodes in between lie in the same strongly connected component $\mathrm{Sc}_\mathcal{G}(v_{i})$ (or $v_{i}$ is the left endnode anyways).

Case 2: Since case 1 is solved we can assume that we have $j<n$ with $v_{j}\hus v_{j+1}$ on $\pi$. If $v_{i-1}\hus v_{i}$ on $\pi^{\prime}$ (or $v_{i}$ the left endnode) then this case is analogous to case 1. So we can also assume that we have $i>0$ and $v_{i-1}\suh v_{i}$ on $\pi$. So $\pi$ looks as follows:
\[
\pi:\qquad\cdots v_{i-1}\suh v_{i}\sus\cdots \sus v_{j}\hus v_{j+1}\cdots.
\]
So there must be a smallest number $k\in\{i,\ldots,j\}$ such that a collider appears at $v_{k}$ on $\pi$:
\[
\pi:\qquad\cdots v_{i-1}\suh v_{i}\tuh\cdots \tuh v_{k}\hus\cdots\sus v_{j}\hus v_{j+1}\cdots.
\]
Since $\pi$ is $Z$-$\sigma$-open we have $v_{k}\in\mathrm{Anc}_\mathcal{G}(Z)$. Since $v_{i}\in\mathrm{Anc}_\mathcal{G}(v_{k})$ (otherwise $v_{k}$ would not be the first collider appearing after $v_{i}$) we thus have that also $v_{i}\in\mathrm{Anc}_\mathcal{G}(Z)$. So if we replace the subwalk $v_{i}\sus\cdots\sus v_{j}$ of $\pi$ by the shortest directed path $v_{i}\hut\cdots\hut v_{j}$ in $\mathcal{G}$ we then get for $\pi^{\prime}$ the following situation:
\[
\pi^{\prime}:\qquad\cdots v_{i-1}\suh v_{i}\hut\cdots \hut v_{j}\hus v_{j+1}\cdots,
\]
which is then $Z$-$\sigma$-open at $v_{i}$ as $v_{i}\in\mathrm{Anc}_\mathcal{G}(Z)$. Note that this holds also when $v_{i}=v_{j}$. If $v_{i}\neq v_{j}$ then $v_{j}$ is also $Z$-$\sigma$-open on $\pi^{\prime}$ as $v_{j}$ points left to a node in the same strongly connected component as $v_{j}$.

So in all cases $\pi^{\prime}$ stays $Z$-$\sigma$-open.
\end{proof}

\begin{proposition} \label{proposition3.3.6}
Let $\mathcal{G}=(\mathcal{V},\mathcal{E},\mathcal{L})$ be a DMG. For $Z\subseteq \mathcal{V}$, and $v_i,v_j\in\mathcal{V}$, the following are equivalent:
\begin{enumerate}
    \item there exists a $Z$-$\sigma$-open path between $v_i$ and $v_j$ in $\mathcal{G}$;
    \item there exists a $Z$-$\sigma$-open walk between $v_i$ and $v_j$ in $\mathcal{G}$.
\end{enumerate}
\end{proposition}
\begin{proof}
1 $\implies$ 2 is trivial.

2 $\implies$ 1: Let $\pi$ be a $Z-\sigma$-open walk between $v_0$ and $v_n$ in $\mathcal{G}$. If a node $w$ occurs more than once on $\pi$, let $v_i$ be the first node on $\pi$ and $v_j$ be the last node on $\pi$ that are in $\text{Sc}_\mathcal{G}(w)$. We now use Lemma~\ref{lemma3.3.10} to construct a new walk $\pi'$ from $\pi$ by replacing the subwalk between $v_i$ and $v_j$ of $\pi$ by a particular directed path in $\text{Sc}_\mathcal{G}(w)$ between $v_i$ and $v_j$ in such a way that $\pi'$ is still $Z$-$\sigma$-open. On $\pi'$, the number of nodes that occurs more than once is at least one less than on $\pi$, and all nodes within $\text{Sc}_\mathcal{G}(w)$ occur within a single segment. This replacement procedure can be repeated until no nodes occur more than once. We have then obtained a $Z$-$\sigma$-open path between $v_0$ and $v_n$. 
\end{proof}

\subsection{\texorpdfstring{$\sigma$-Markov Equivalence}{sigma-Markov Equivalence}}
With selection nodes $\mathcal{S}$, a $\sigma$-inducing path generalizes the notion of inducing paths in DAGs and plays a crucial role in Definition~\ref{representing}.

\begin{definition} [$\sigma$-inducing Paths \citep{mooij2020constraint}] \label{sigmaSeparation}
Let \( \mathcal{G} = (\mathcal{V}^+,\mathcal{E},\mathcal{L}) \) be a DMG with nodes $\mathcal{V}^+=\mathcal{V}\cup\mathcal{S}$ and let $v_i,v_j\in\mathcal{V}$ be distinct nodes. A walk $\pi$ in $\mathcal{G}$ between $v_i$ and $v_j$ is called \emph{$\sigma$-inducing given $\mathcal{S}$} if each collider on $\pi$ is in $\anc_\mathcal{G}(\{v_i,v_j\}\cup\mathcal{S})$, and each non-endpoint non-collider on $\pi$ is unblockable. If it is a path, it is called a \emph{$\sigma$-inducing path given $\mathcal{S}$} between $v_i$ and $v_j$. 
\end{definition}

\begin{proposition}[Proposition 1 in \citep{mooij2020constraint}] \label{proposition12.2.3}
Let $\mathcal{G}=(\mathcal{V}^+,\mathcal{E},\mathcal{L})$ be a DMG with nodes $\mathcal{V}^+=\mathcal{V}\cup\mathcal{S}$ and let $v_i,v_j\in\mathcal{V}$ be distinct nodes. Then the following are equivalent:
\begin{enumerate}
    \item There is a $\sigma$-inducing path given $\mathcal{S}$ in $\mathcal{G}$ between $v_i$ and $v_j$;
    \item There is a $\sigma$-inducing walk given $\mathcal{S}$ in $\mathcal{G}$ between $v_i$ and $v_j$;
    \item $v_i\overset{\sigma}{\underset{\mathcal{G}}{\perp}} v_j\mid (Z\cup\mathcal{S})$ for all $Z\subseteq\mathcal{V}\backslash\{v_i,v_j\}$;
    \item $v_i\overset{\sigma}{\underset{\mathcal{G}}{\perp}} v_j\mid (Z\cup\mathcal{S})$ for all $Z=(\anc_\mathcal{G}(\{v_i,v_j\}\cup\mathcal{S}))\backslash\{v_i,v_j\}$.
\end{enumerate}
\end{proposition}
\begin{proof}
The proof is similar to that of Theorem 4.2 in \citep{richardson2002ancestral}.

  1 $\implies$ 2 is trivial.

  2 $\implies$ 3: Assume the existence of a $\sigma$-inducing walk given $\mathcal{S}$ between $v_i$ and $v_j$ in $\mathcal{G}$. 
  Let $Z \subseteq \mathcal{V}  \backslash\{v_i,v_j\}$. 
  Consider all walks in $\mathcal{G}$ between $v_i$ and $v_j$ with the property that all colliders on it are in $\anc_{\mathcal{G}}({\{v_i,v_j\} \cup \mathcal{S} \cup Z)}$, and each non-endpoint non-collider on it is not in $\mathcal{S} \cup Z$ or is unblockable.
  Such walks exist, since the $\sigma$-inducing walk is one. 
  Let $\mu$ be such a walk with a minimal number of colliders.
  We show that all colliders on $\mu$ must be in $\anc_{\mathcal{G}}(\mathcal{S} \cup Z)$. 
  Suppose on the contrary the existence of a collider $v_k$ on $\mu$ that is not ancestor of $\mathcal{S} \cup Z$.
  It is either ancestor of $v_i$ or of $v_j$, by assumption.
  If $v_j \in J$, it cannot be ancestor of $v_j$, and hence must be ancestor of $v_i$.
  Otherwise, we can assume it to be ancestor of $v_i$ without loss of generality.
  Then there is a directed path $\pi$ from $v_k$ to $v_i$ in $\mathcal{G}$ that does not pass through any node of $\mathcal{S} \cup Z$.
  Then the subwalk of $\mu$ between $v_k$ and $v_j$ can be concatenated with the directed path $\pi$
  into a walk between $v_i$ and $v_j$ that has the property, but has fewer colliders than $\mu$: a contradiction. 
  Therefore, $\mu$ is $\sigma$-open given $\mathcal{S} \cup Z$.
  Hence, $v_i$ and $v_j$ are $\sigma$-connected given $\mathcal{S} \cup Z$.

  3 $\implies$ 4 is trivial.

  4 $\implies$ 1:
  Suppose that 
  $v_i$ and $v_j$ are $\sigma$-connected given $Z = (\anc_{\mathcal{G}}(\{v_i,v_j\} \cup \mathcal{S}) ) \backslash \{v_i,v_j\}$. 
  Let $\pi$ be a path between $v_i$ and $\{v_j\}$ that is $\sigma$-open given $Z$. 
  We show that $\pi$ must be a $\sigma$-inducing path given $\mathcal{S}$. 
  First, all colliders on $\pi$ are in $\anc_\mathcal{G}(Z)$, and hence in $\anc_{\mathcal{G}}(\{v_i,v_j\} \cup \mathcal{S})$.
  Second, let $v_k$ be any non-endpoint non-collider on $\pi$.
  Then there must be a directed subpath of $\pi$ starting at $v_k$ that ends either at the first collider on $\pi$ next to $v_k$ or at an end node of $\pi$, and hence $v_k$ must be in $Z$. 
  Since $\pi$ is $\sigma$-open given $Z$, $v_k$ must be unblockable.
  Hence, all non-endpoint non-colliders on $\pi$ must be unblockable.
\end{proof}

\begin{lemma}\label{lemma12.2.5}
Let $\mathcal{G}=(\mathcal{V}^+,\mathcal{E},\mathcal{L})$ be a DMG with nodes $\mathcal{V}^+=\mathcal{V}\cup\mathcal{S}$ and let $v_i,v_j\in\mathcal{V}$ be distinct nodes. If there exists a $\sigma$-inducing path given $\mathcal{S}$ between $v_i$ and $v_j$ in $\mathcal{G}$, and all $\sigma$-inducing paths given $\mathcal{S}$ in $\mathcal{G}$ between $v_i$ and $v_j$ are out of $v_j$, then $v_j\in \anc_\mathcal{G}(\{v_i\}\cup\mathcal{S})$.
\end{lemma}
\begin{proof}
Let $\mu$ be a $\sigma$-inducing path given $\mathcal{S}$ between $v_i$ and $v_j$ in $\mathcal{G}$. 
It must be of the form $v_i \cdots v_k \hut v_j$ (with possibly $v_k = v_i$). 
First we show that $v_k$ cannot be in $\text{Sc}_\mathcal{G}(v_j)$.
If $v_k \in \text{Sc}_\mathcal{G}(v_j)$, then let $\pi$ be a directed path in $\mathcal{G}$ from $v_k$ to $v_j$ that is entirely contained in $\text{Sc}_\mathcal{G}(v_j)$.
Let $m$ be the node on $\mu$ closest to $v_i$ that is also on $\pi$ (possibly $m = v_k$).
The subpath of $\pi$ between $v_j$ and $m$ can be concatenated with the subpath of $\mu$ between $m$ and $v_i$ into a walk between $v_j$ and $v_i$.
This must be a $\sigma$-inducing path given $\mathcal{S}$ between $v_i$ and $v_j$ that is into $v_j$ by construction: contradiction.
Hence $v_k$ cannot be in $\text{Sc}_\mathcal{G}(v_j)$.

If $\mu$ is a directed path all the way to $v_i$, then clearly, $v_j \in \anc_\mathcal{G}(\{v_i\} \cup \mathcal{S})$.
Otherwise, it must contain a collider.
Let $v_l$ be the collider on $\mu$ closest to $v_j$.
$v_l$ must be ancestor of $v_i$ or $v_j$ or $\mathcal{S}$.
In the first and third cases, clearly $v_j \in \anc_\mathcal{G}(\{v_i\} \cup \mathcal{S})$.
In the second case, all nodes on the subpath of $\mu$ between $v_j$ and $v_l$ must be in $\text{Sc}_\mathcal{G}(v_j)$, a contradiction.
\end{proof}

\begin{lemma}\label{lemma12.2.6}
Let ${\mathcal{G}}=(\mathcal{V}^+,\mathcal{E},\mathcal{L})$ be a DMG with nodes $\mathcal{V}^+=\mathcal{V}\cup{\mathcal{S}}$ and let $v_i,v_j\in\mathcal{V}$ be distinct nodes. If there exists a $\sigma$-inducing path given ${\mathcal{S}}$ between $v_i$ and $v_j$ in ${\mathcal{G}}$ into $v_j$, and $v_i\notin \anc_{\mathcal{G}}(\{v_j\}\cup{\mathcal{S}})$, then there exists a $\sigma$-inducing path given ${\mathcal{S}}$ between $v_i$ and $v_j$ in ${\mathcal{G}}$ that is both into $v_i$ and into $v_j$.
\end{lemma}
\begin{proof}
Let $\mu$ be a $\sigma$-inducing path given $\mathcal{S}$ between $v_i$ and $v_j$ in $\mathcal{G}$ into $v_j$.
  If $\mu$ is into $v_i$, we are done.
  Therefore, suppose it is of the form $v_i \tuh \cdots \suh v_j$. 
  It cannot be a directed path, since $v_i \notin \anc_\mathcal{G}(\{v_j\} \cup \mathcal{S})$.
  Therefore, there must be a collider $v_k$ on $\mu$ such that $\mu$ is of the form $v_i \tuh \dots \tuh v_k \hus \cdots \suh v_j$
  (with the subpath between $v_i$ and $v_k$ directed).
  Then $v_k \in \anc_\mathcal{G}(\{v_i\})$, and hence all nodes on $\mu$ between $v_i$ and $v_k$ must be in $\text{Sc}_\mathcal{G}(v_i)$.
  Let $\pi$ be a directed path in $\mathcal{G}$ from $v_k$ to $v_i$ that is entirely contained in $\text{Sc}_\mathcal{G}(v_i)$.
  Let $v_l$ be the node on $\mu$ closest to $v_j$ that is also on $\pi$ (possibly $v_l = v_k$).
  Then $v_l \ne v_j$, because otherwise $v_j \in \text{Sc}_\mathcal{G}(v_i)$, contradicting $v_i \notin \anc_\mathcal{G}(\{v_j\} \cup \mathcal{S})$.
  The non-trivial subpath of $\pi$ between $v_i$ and $v_l$ can be concatenated with the non-trivial subpath of $\mu$ between $v_l$ and $v_j$ into a walk between $v_i$ and $v_j$.
  This must be a $\sigma$-inducing path given $\mathcal{S}$ between $v_i$ and $v_j$ that is both into $v_i$ and into $v_j$.
\end{proof}

\begin{definition}[$\sigma$-Markov Equivalence Given $\mathcal{S}$] \label{sigmamarkov}
Two DMGs \( \mathcal{G}_1, \mathcal{G}_2 \) with the same node set \( \mathcal{V}^+ = \mathcal{V} \cup \mathcal{S} \) are said to be \emph{\(\sigma\)-Markov equivalent given \( \mathcal{S} \)} if, for any three subsets of nodes \( X, Y, Z \subseteq \mathcal{V} \), it holds that  
\[
X \text{ is } \sigma\text{-separated from } Y \text{ given } Z \cup \mathcal{S} \text{ in } \mathcal{G}_1
\]
if and only if  
\[
X \text{ is } \sigma\text{-separated from } Y \text{ given } Z \cup \mathcal{S} \text{ in } \mathcal{G}_2.
\]
\end{definition}

\section{Proofs}
\subsection{Proofs in Section~\ref{section:1}}
\lemmadef*
\begin{proof}
By the \( \sigma \)-completeness of \( \sigma \)-MAGs, \( a \) and \( c \) are adjacent in \( \mathcal{H} \). Suppose the edge between \( a \) and \( b \) is \( a \tuh b \). Since \( \mathcal{H} \) is ancestral and contains a path \( a \tuh b \tut c \), it follows that \( \mathcal{H} \) does not contain an edge of the form \( a \hus c \). Therefore, \( \mathcal{H} \) must contain either \( a \tuh c \) or \( a \tut c \). If \( \mathcal{H} \) contains the edge \( a \tut c \), then from the path \( b \tut c \tut a \), we conclude that \( \mathcal{H} \) should not contain the edge \( b \hus a \), which is a contradiction. Thus, the edge between \( a \) and \( c \) must be \( a \tuh c \) as well.

Now we suppose the edge between \( a \) and \( b \) is \( a \huh b \). 
\begin{enumerate}
\item If \( \mathcal{H} \) contains the edge \( a \tut c \), then from the path \( b \tut c \tut a \), it follows that \( \mathcal{H} \) would not contain an edge of the form \( a \suh b \), which is a contradiction.
\item If \( \mathcal{H} \) contains the edge \( a \tuh c \), then from the path \( a \tuh c \tut b \), it follows that \( \mathcal{H} \) would not contain an edge of the form \( b \suh a \), which is a contradiction.
\item If \( \mathcal{H} \) contains the edge \( a \hut c \), then from the path \( b \tut c \tuh a \), it follows that \( \mathcal{H} \) would not contain an edge of the form \( a \suh b \), which is a contradiction.
\end{enumerate}
Thus, we conclude that \( \mathcal{H} \) must contain the edge \( a \huh c \).

If \( \mathcal{H} \) also contains \( b \tut d \), then by the \( \sigma \)-maximality of \( \sigma \)-MAGs, \( c \) and \( d \) must be adjacent in \( \mathcal{H} \). Since \( \mathcal{H} \) contains both paths \( d \tut b \tut c \) and \( c \tut b \tut d \), it follows that \( \mathcal{H} \) cannot contain an edge of the form \( c \suh d \) or \( d \suh c \). Thus, the edge between \( c \) and \( d \) must be \( c \tut d \). Hence, we conclude that $\text{Nbh}_\mathcal{H}(b)$ is complete. Similarly, since \( \mathcal{H} \) also contains the triple \( a \suh c \tut b \), as we have shown, it follows by symmetry that the neighborhood of \( c \) must also be complete.
\end{proof}


\lemmaant*
\begin{proof}
Start with \( v_2 \). If the edge between \( v_1 \) and \( v_2 \) is directed, then clearly \( v_0 \in \anc_\mathcal{H}(v_2) \). If instead \( v_1 - v_2 \) is present, then by Lemma~\ref{oldDef}, \( \mathcal{H} \) also contains the directed edge \( v_0 \rightarrow v_2 \), given the triple \( v_0 \tuh v_1 \tut v_2 \). Now, suppose that for some \( 2 \leq k \leq n \), we have \( v_0 \in \anc_\mathcal{H}(v_k) \). If the edge between \( v_k \) and \( v_{k+1} \) is directed, then it directly follows that \( v_0 \in \anc_\mathcal{H}(v_{k+1}) \). Since \( v_0 \in \anc_\mathcal{H}(v_k) \), there exists a directed path from \( v_0 \) to \( v_k \) in \( \mathcal{H} \). Let the last edge on this path be \( a \rightarrow v_k \). If the edge between \( v_k \) and \( v_{k+1} \) is undirected, then from the triple \( a \tuh v_k \tut v_{k+1} \), it follows that \( a \tuh v_{k+1} \). Consequently, \( \mathcal{H} \) contains a directed path from \( v_0 \) to \( v_{k+1} \), so we conclude that \( v_0 \in \anc_\mathcal{H}(v_{k+1}) \). Thus, by induction, \( v_0 \in \anc_\mathcal{H}(v_n) \), completing the proof.
\end{proof}


\lemmaone*
\begin{proof}
We construct the DMG \( \mathcal{G} \) as follows. It has nodes \( \mathcal{V}^+ = \mathcal{V} \cup \mathcal{S} \), where  
\[
\mathcal{S} = \{s_{\{a,b\}} : a\tut b \in \mathcal{H}, \text{Nbh}_\mathcal{H}(a)  \text{ or } \text{Nbh}_\mathcal{H}(b) \text{ is incomplete}\}.
\]  

First, we include all (bi)directed edges \( a \suh b \) for \(a,b \in \mathcal{V} \) that are present in \( \mathcal{H} \) as edges in \( \mathcal{G} \). For undirected edges \( a\tut b \) in \( \mathcal{H} \), we treat them as follows: If the neighborhood of \( a \) or $b$ is incomplete, we introduce a new node \( s_{\{a, b\}} \) and add the edges \( a \rightarrow s_{\{a, b\}} \leftarrow b \). On the other hand, if both of the neighborhoods of \( a \) and $b$ are complete, we replace the undirected edge \( a - b \) with the directed edges \( a \rightarrow b \) and \( a \leftarrow b \). 

We need to show that \( \mathcal{H} \) represents \( \mathcal{G} \) given \( \mathcal{S} \). For that we need to show:
\begin{enumerate}
    \item Two distinct nodes \(a, b \in \mathcal{V}\) are adjacent in \(\mathcal{H}\) if and only if there exists a \(\sigma\)-inducing path between \(a\) and \(b\) given \(\mathcal{S}\) in \(\mathcal{G}\).
    \item If \(a \hus b\) in \(\mathcal{H}\), then \(a \notin \anc_\mathcal{G}(\{b\} \cup \mathcal{S})\).
    \item If \(a \tus b\) in \(\mathcal{H}\), then \(a \in \anc_\mathcal{G}(\{b\} \cup \mathcal{S})\).
\end{enumerate}

We now demonstrate that for \( a \hus b \in \mathcal{H} \), it holds that \( a \notin \anc_{\mathcal{G}}(\mathcal{S}) \). Assume, for the sake of contradiction, that \( a \in \anc_{\mathcal{G}}(\mathcal{S}) \), which implies the existence of a shortest directed path from $a$ to some $s\in\mathcal{S}$ in $\mathcal{G}$. Suppose this path has more than one edge, meaning it takes the form \( a \rightarrow \cdots \rightarrow c \rightarrow d \rightarrow s \), where $a$ and $c$ might be the same, and all nodes except $s$ on the path belong to \( \mathcal{V} \). By the construction of \( \mathcal{G} \), there exists a node \( e \in \mathcal{V} \) such that \( d \rightarrow s \leftarrow e \) with \( s = s_{\{d, e\}} \) in \( \mathcal{G} \), the edge \( d \tut e \) is present in \( \mathcal{H} \), and the neighborhood of \( d \) or \( e \) is incomplete. If \( c \rightarrow d \) is also present in \( \mathcal{H} \), then by Lemma~\ref{oldDef}, \( \text{Nbh}_\mathcal{H}(d) \) and \( \text{Nbh}_\mathcal{H}(e) \) must be complete, which is a contradiction. Thus, the edge between \( c \) and \( d \) must be undirected in \( \mathcal{H} \). Moreover, by the construction of \( \mathcal{G} \), we only have \( c \tuh d \) in \( \mathcal{G} \) given \( c \tut d \) in \( \mathcal{H} \) when the neighborhoods of \( c \) and \( d \) are complete. So, \( c \) and \( e \) are connected by an undirected edge in \( \mathcal{H} \), and there exists a node \( s_{\{c,e\}} \in \mathcal{S} \) such that \( \mathcal{G} \) contains \( c \rightarrow s_{\{c,e\}} \leftarrow e \) because \( \text{Nbh}_\mathcal{H}(e) \) is incomplete. This leads to a shorter directed path \( a \rightarrow \cdots \rightarrow c \rightarrow s_{\{c,e\}} \) from \( a \) to \( \mathcal{S} \), which is a contradiction. The only remaining possibility is that we have \( a \rightarrow s \) in \( \mathcal{G} \), where \( s \in \mathcal{S} \). By the construction of \( \mathcal{G} \), there exists a node \( f \in \mathcal{V} \) such that \( a \rightarrow s \leftarrow f \) in \( \mathcal{G} \) with \( s = s_{\{a, f\}} \), the edge \( a - f \) is present in \( \mathcal{H} \), and the neighborhood of \( a \) or \( f \) is incomplete. Combining this with the edge \( b \suh a \) in \( \mathcal{H} \), we obtain a triple \( b \suh a \tut f \) in \( \mathcal{H} \), which, by Lemma~\ref{oldDef}, implies that \( \text{Nbh}_\mathcal{H}(a) \) and \( \text{Nbh}_\mathcal{H}(f) \) must be complete, leading to a contradiction. Thus, we conclude that \( a \notin \anc_{\mathcal{G}}(\mathcal{S}) \).

Next, we will show that for \( a \hus b \in \mathcal{H} \), it holds that \( a \notin \anc_{\mathcal{G}}(b) \). Assume the contrary, that \( a \in \anc_\mathcal{G}(b) \). This implies the existence of a shortest directed path from \( a \) to \( b \) in \( \mathcal{G} \), which corresponds to an anterior path from \( a \) to \( b \) in \( \mathcal{H} \). However, \( \mathcal{H} \) also contains the edge \( b \suh a \), which contradicts the definition of \( \sigma \)-MAGs. Therefore, \( a \notin \anc_\mathcal{G}(b) \). Thus, we conclude that if \( a \hus b \in \mathcal{H} \), then \( a \notin \anc_\mathcal{G}(\{b\} \cup \mathcal{S}) \).

Now suppose \( a \tus b \) in \( \mathcal{H} \). Then either \( a \tut b \) in \( \mathcal{H} \) or \( a \tuh b \) in \( \mathcal{H} \). In the first case, by construction of \( \mathcal{G} \), \( a \) must be an ancestor of either \( b \) or \( \mathcal{S} \) in \( \mathcal{G} \). In the second case, by construction of \( \mathcal{G} \), \( a \) must be an ancestor of \( b \). Hence, \( a \in \anc_{\mathcal{G}}(\{b\} \cup \mathcal{S}) \).

Finally, we consider adjacency. If two distinct nodes \( v_0, v_n \in \mathcal{V} \) are adjacent in \( \mathcal{H} \), then \( v_0, v_n \) must also be adjacent in \( \mathcal{G} \), or there must be a triple \( v_0 \rightarrow s_{v_0, v_n} \leftarrow v_n \) in \( \mathcal{G} \). In both cases, there exists a \( \sigma \)-inducing path given $\mathcal{S}$ between \( v_0 \) and \( v_n \) in \( \mathcal{G} \). Next, suppose \( v_0, v_n \in \mathcal{V} \) are not adjacent in \( \mathcal{H} \). This implies that there should not be a \( \sigma \)-inducing path between \( v_0 \) and \( v_n \) given \( \mathcal{S} \) in \( \mathcal{G} \). We assume the opposite and let \( \pi \) be a shortest \( \sigma \)-inducing path given $\mathcal{S}$ with the smallest collider distance sum to \( \{v_0, v_n\} \cup \mathcal{S} \) between \( v_0 \) and \( v_n \) in \( \mathcal{G} \). Notice that there is no other \( \sigma \)-inducing path \( \mu \) between \( v_0 \) and \( v_n \) in \( \mathcal{G} \) given $\mathcal{S}$ such that \( \mu \) has fewer edges than \( \pi \), or \( \mu \) has the same number of edges but a smaller collider distance sum to \( \{v_0, v_n\} \cup \mathcal{S} \) than \( \pi \).

Suppose \( \pi \) contains a node from \( \mathcal{S} \). Then \( \pi \) must include a subpath of the form \( v_k \rightarrow s_{\{v_k,v_{k+1}\}} \leftarrow v_{k+1} \). In this case, \( v_k \) and \( v_{k+1} \) are non-colliders, and both are blockable. This implies that \( \pi \) reduces to \( v_0 \rightarrow s_{\{v_0,v_n\}} \leftarrow v_n \), and \( v_0 - v_n \) is an edge in \( \mathcal{H} \), which contradicts the assumption that \( v_0 \) and \( v_n \) are not adjacent in \( \mathcal{H} \). Hence, \( \pi \) must only contain nodes from \( \mathcal{V} \).

Suppose \( v_k \) is a non-endpoint non-collider on \( \pi \). W.O.L.G., assume \( v_{k-1} \sus v_k \rightarrow v_{k+1} \) is present on \( \pi \). By the definition of a \( \sigma \)-inducing path, we know that \( v_k \) and \( v_{k+1} \) belong to the same strongly connected component in \( \mathcal{G} \), which implies that there exists a directed path from \( v_{k+1} \) to \( v_k \) in \( \mathcal{G} \). If the edge \( v_k \rightarrow v_{k+1} \) is present in \( \mathcal{H} \), then by the second property we have established, $v_{k+1}\notin \anc_\mathcal{G}(v_k)$, which leads to a contradiction. Therefore, the undirected edge \( v_k - v_{k+1} \) must be present in \( \mathcal{H} \). 

If we have \( v_{k-1} \leftarrow v_k \rightarrow v_{k+1} \) on \( \pi \), then similarly, we obtain \( v_{k-1} - v_k \) in \( \mathcal{H} \). By the construction of \( \mathcal{G} \), \( \text{Nbh}_\mathcal{H}(\{v_k\}) \) is complete, implying that \( v_{k-1} \) and \( v_{k+1} \) are connected by an undirected edge in \( \mathcal{H} \) and their neighborhoods are both complete. Consequently, we can replace the triple \( v_{k-1} \leftarrow v_k \rightarrow v_{k+1} \) on \( \pi \) with \( v_{k-1} \rightarrow v_{k+1} \). The new path remains \( \sigma \)-inducing since \( v_{k+1} \) retains the same edge mark, and \( v_{k-1} \) now serves as an unblockable non-collider on the new path, regardless of whether it was a collider or a non-collider on \( \pi \). This is contradictory, as we obtain a shorter \( \sigma \)-inducing path than \( \pi \).

If \( v_{k-1} \suh v_k \) is present in \( \mathcal{H} \), then by the construction of \( \mathcal{G} \), we also have \( v_{k-1} \suh v_k \) in \( \mathcal{G} \), and by Lemma~\ref{oldDef}, \( \mathcal{G} \) contains the edge \( v_{k-1} \suh v_{k+1} \). Thus, we can replace the triple \( v_{k-1} \suh v_k \tuh v_{k+1} \) on \( \pi \) with \( v_{k-1} \suh v_{k+1} \), obtaining a shorter \( \sigma \)-inducing path than \( \pi \), which is contradictory. Now, consider the last case where \( v_{k-1} \rightarrow v_k \) on \( \pi \) and \( v_{k-1} \tut v_k \) in \( \mathcal{H} \). Similarly as for the case $v_{k-1}\hut v_k\tuh v_{k+1}$ on $\pi$, \( v_{k-1} \) and \( v_{k+1} \) are connected by an undirected edge in \( \mathcal{H} \) and their neighborhoods are both complete, so $\mathcal{G}$ contains the directed edge $v_{k-1}\tuh v_{k+1}$. Replacing \( v_{k-1} \rightarrow v_k \rightarrow v_{k+1} \) on \( \pi \) with \( v_{k-1} \rightarrow v_{k+1} \) yields a shorter \( \sigma \)-inducing path, which is contradictory. Thus, all non-end nodes on \( \pi \) must be colliders.

Given that \( \pi \) is a shortest \( \sigma \)-inducing path with the smallest collider distance to \( \{v_0, v_n\} \cup \mathcal{S} \), all edges except possibly for the two end-edges on it are also present in \( \mathcal{H} \), because they are bidirected. There are two cases:

\begin{enumerate}    
    \item Suppose \( \pi \) has more than one collider. Then the two end-edges must also be present in \( \mathcal{H} \); otherwise, W.L.O.G., assume \( v_0 \tut v_1 \huh v_2 \) appears in $\mathcal{H}$. By definition, this triple on $\pi$ can be replaced with \( v_0 \huh v_2 \), leading to a shorter \( \sigma \)-inducing path, contradicting minimality. Furthermore, each collider on \( \pi \) is not in \( \anc_\mathcal{G}(\mathcal{S}) \) (as they have arrowheads pointing toward themselves in \( \mathcal{H} \)) but is in \( \anc_\mathcal{G}(\{v_0, v_n\}) \). Thus, for any collider \( v_k \) in the triple \( v_{k-1} \suh v_k \hus v_{k+1} \) on \( \pi \), there exists a directed path from \( v_k \) to \( v_0 \) or \( v_n \) in \( \mathcal{G} \). W.L.O.G., suppose there exists a shortest directed path \( v_k \rightarrow q \rightarrow \cdots \rightarrow v_0 \) from $v_k$ to $v_0$ in \( \mathcal{G} \). The corresponding path in \( \mathcal{H} \) is \( v_k \tus q \tus \cdots \tus v_0 \). If \( v_k - q \) is present in \( \mathcal{H} \), then by definition, \( \mathcal{H} \) also contains the path \( v_{k-1}\suh q \hus v_{k+1} \). Replacing the subpath of \( \pi \) between \( v_{k-1} \) and \( v_{k+1} \) with this path results in a \( \sigma \)-inducing path with a smaller collider distance sum to \( \{v_i, v_j\} \cup \mathcal{S} \), which is contradictory. Thus, we have \( v_k \rightarrow q \) in \( \mathcal{H} \). Then, by Lemma~\ref{lemmaant}, there exists a directed path from \( v_k \) to \( v_0 \) in \( \mathcal{H} \). Notice that we also consider the case where there is a path consisting of a single edge \( v_k \rightarrow v_0 \) in \( \mathcal{G} \). If \( v_k \tut v_0 \) is present in \( \mathcal{H} \), then by Lemma~\ref{oldDef} $\mathcal{H}$ also contains the edge $v_0\hus v_{k+1}$, implying that the subpath of \( \pi \) from \( v_0 \) to \( v_{k+1} \) can be replaced by \( v_0 \hus v_{k+1} \), contradicting the minimality of \( \pi \). Hence, \( v_k \rightarrow v_0 \) must be present in \( \mathcal{H} \).

    \item Suppose \( \pi \) is \( v_0 \suh v_1 \hus v_2 \) with \( v_n = v_2 \). If either of the edges \( v_0 \suh v_1 \) or $v_1\hus v_2$ is present in \( \mathcal{H} \), then the other must also be present in $\mathcal{H}$, since if \( v_0 \tut v_1 \) or \( v_1 \tut v_2 \) is in \( \mathcal{H} \), Lemma~\ref{oldDef} ensures that \( v_0 \) and \( v_2 \) must be adjacent, which is contradictory. Moreover, if \( v_0 \tut v_1 \tut v_2 \) exists in \( \mathcal{H} \), then by the construction of \( \mathcal{G} \), \( v_0 \) and \( v_2 \) must be connected by an undirected edge in \( \mathcal{H} \), again leading to a contradiction. Therefore, we only need to consider the case where the triple \( v_0 \suh v_1 \hus v_2 \) is explicitly present in \( \mathcal{H} \). Additionally, by the second property, we have $v_1\notin\anc_\mathcal{G}(\{v_0,v_2\}\cup\mathcal{S})$, which contradicts the assumption that $\pi$ is $\sigma$-inducing given $\mathcal{S}$.
\end{enumerate}

Hence, $\pi$ must contain more than one collider, and all colliders on \( \pi \) are ancestors of \( \{v_0, v_n\} \) in \( \mathcal{H} \). Consequently, \( \pi \) corresponds to an inducing path between \( v_0 \) and \( v_n \) in \( \mathcal{H} \), contradicting the maximality of \( \mathcal{H} \).

Consequently, we conclude that \( \mathcal{H} \) represents \( \mathcal{G} \) given \( \mathcal{S} \).
\end{proof}

\lemmazero*
\begin{proof}
\begin{enumerate}
    \item Suppose \( \mathcal{H} \) contains a directed path \( a = v_0 \tus \cdots \tus v_n = b \). Note that for all \( k = 0, 1, \ldots, n-1 \), we have \( v_k \in \anc_{\mathcal{G}}(\{v_{k+1}\} \cup \mathcal{S}) \). By induction, it follows that \( a \in \anc_{\mathcal{G}}(\{b\} \cup \mathcal{S}) \).
    \item Since \( a \) and \( b \) are adjacent in \( \mathcal{H} \), and \( \mathcal{H} \) represents \( \mathcal{G} \) given \( \mathcal{S} \), there exists a \( \sigma \)-inducing path given \( \mathcal{S} \) between \( a \) and \( b \) in \( \mathcal{G} \). Suppose, for contradiction, that all such \( \sigma \)-inducing paths given \( \mathcal{S} \) between \( a \) and \( b \) in \( \mathcal{G} \) are out of \( b \). Then, by Lemma~\ref{lemma12.2.5}, it would follow that \( b \in \anc_{\mathcal{G}}(\{a\} \cup \mathcal{S}) \), contradicting the orientation \( a \suh b \) in \( \mathcal{H} \). Therefore, there must exist a \( \sigma \)-inducing path given \( \mathcal{S} \) between \( a \) and \( b \) in \( \mathcal{G} \) that is into \( b \).
    \item Similarly, there exists a \( \sigma \)-inducing path given \( \mathcal{S} \) between \( a \) and \( b \) in \( \mathcal{G} \). Applying Lemma~\ref{lemma12.2.6}, if \( a \huh b \) is present in \( \mathcal{H} \) so that $a\notin \anc_\mathcal{G}(\{b\}\cup \mathcal{S})$, then there exists a \( \sigma \)-inducing path given \( \mathcal{S} \) in \( \mathcal{G} \) between \( a \) and \( b \) that is into both \( a \) and \( b \).
\end{enumerate}
\end{proof}

\lemmatwo*
\begin{proof}
By the assumption that \( \mathcal{H} \) represents \( \mathcal{G} \) given \( \mathcal{S} \), we know that between any two distinct nodes in \( \mathcal{H} \), there is at most one edge, and no node is adjacent to itself.

Now we show that $\mathcal{H}$ is ancestral. Suppose \( \mathcal{H} \) contains an anterior path from \( v_0 \) to \( v_n \), namely \( v_0 \tus \cdots \tus v_n \). By Lemma~\ref{lemma0}, \( v_0 \in \anc_\mathcal{G}(\{v_n\} \cup \mathcal{S}) \). If \( \mathcal{H} \) also contains an edge \( v_n \suh v_0 \), it would imply that \( v_0 \notin \anc_\mathcal{G}(\{v_n\} \cup \mathcal{S}) \), which is a contradiction. Therefore, such edges cannot exist, meaning that \( \mathcal{H} \) is ancestral.

We now continue to show that \( \mathcal{H} \) is maximal, meaning that there is no inducing path between any two distinct non-adjacent nodes. Suppose, for contradiction, that there exists an inducing path \( \pi \) between two distinct non-adjacent nodes \( v_0, v_n \in \mathcal{V} \). For every edge \( v_{k} \sus v_{k+1} \) on \( \pi \), where \( k = 0, \dots, n-1 \), there exists a \( \sigma \)-inducing path \( \mu_k \) between \( v_{k} \) and \( v_{k+1} \) given \( \mathcal{S} \) in \( \mathcal{G} \). By Lemma~\ref{lemma0}, these \( \sigma \)-inducing paths can be chosen to be into $v_k$ if the edge is \( v_{k} \hus v_{k+1} \), into \( v_{k+1} \) if the edge is \( v_{k} \suh v_{k+1} \), and both into \( v_{k} \) and \( v_{k+1} \) if the edge is \( v_{k} \huh v_{k+1} \). Now, concatenate all \( \mu_k \) in the order of the edges in \( \pi \) to form a walk \( \mu \) in \( \mathcal{G} \) between \( v_0 \) and \( v_n \). By definition, every non-endpoint node on \( \pi \) is a collider, and by the construction of \( \mu \), these nodes remain colliders on \( \mu \). Since all colliders on \( \pi \) belong to \( \anc_\mathcal{H}(\{v_0, v_n\}) \), they must also be in \( \anc_\mathcal{G}(\{v_0, v_n\} \cup \mathcal{S}) \) in \( \mathcal{G} \). Similarly, all non-endpoint colliders on any \( \mu_k \) are in \( \anc_\mathcal{G}(\{v_{k}, v_{k+1}\} \cup \mathcal{S}) \) and, therefore, in \( \anc_\mathcal{G}(\{v_0, v_n\} \cup \mathcal{S}) \). Thus, all colliders on \( \mu \) are in \( \anc_\mathcal{G}(\{v_0, v_n\} \cup \mathcal{S}) \). Additionally, all non-endpoint non-colliders on each \( \mu_k \) are unblockable, meaning all non-endpoint non-colliders on \( \mu \) are also unblockable. Hence, \( \mu \) is a \( \sigma \)-inducing walk given \( \mathcal{S} \) in \( \mathcal{G} \). Consequently, by Proposition~\ref{proposition12.2.3}, there must exist a \( \sigma \)-inducing path given \( \mathcal{S} \) in \( \mathcal{G} \) between \( v_0 \) and \( v_n \). Since \( \mathcal{H} \) represents \( \mathcal{G} \) given \( \mathcal{S} \), it follows that \( v_0 \) and \( v_n \) must be adjacent in \( \mathcal{H} \), contradicting our assumption.

The last step is to show that $\mathcal{H}$ is $\sigma$-complete. Suppose \( \mathcal{H} \) contains a triple of the form \( a\suh b\tut c \). Since \( b \in \anc_\mathcal{G}(\{c\} \cup \mathcal{S}) \) and \( b \notin \anc_\mathcal{G}(\{a\} \cup \mathcal{S}) \), we must have \( b \in \anc_\mathcal{G}(c) \). Moreover, we also know that \( c \in \anc_\mathcal{G}(\{b\} \cup \mathcal{S}) \). Since \( b \in \anc_\mathcal{G}(c) \) and \( b \notin \anc_\mathcal{G}(\mathcal{S}) \), it follows that \( c \in \anc_\mathcal{G}(b) \). Thus, \( b \) and \( c \) are in the same strongly connected component of \( \mathcal{G} \). Next, we extend a \( \sigma \)-inducing path given \( \mathcal{S} \) in \( \mathcal{G} \) between \( a \) and \( b \) that is into \( b \) (which exists by Lemma~\ref{lemma0}) by appending the directed path from \( b \) to \( c \) within \( \text{Sc}_{\mathcal{G}}(b) \), thereby forming a \( \sigma \)-inducing walk given \( \mathcal{S} \) between \( a \) and \( c \). By Proposition~\ref{proposition12.2.3}, there exists a \( \sigma \)-inducing path given \( \mathcal{S} \) between \( a \) and \( c \), which implies that \( a \) and \( c \) are adjacent. Moreover, if there exists a node \( d \in \mathcal{V} \) such that \( b - d \) is present in \( \mathcal{H} \), then by similar reasoning, from the triple \( a \suh  b - d \), we can infer that \( d \in \text{Sc}_\mathcal{G}(b) \), so \( c \) and \( d \) are in the same strongly connected component in \( \mathcal{G} \). Thus, there exists a directed path from \( c \) to \( d \) in \( \mathcal{G} \), which is also a \( \sigma \)-inducing path given \( \mathcal{S} \), implying that \( c \) and \( d \) are adjacent. Therefore, we conclude that $\mathcal{H}$ is $\sigma$-complete.
\end{proof}

\thmone*
\begin{proof}
Obviously obtained by Lemma~\ref{lemma1} and Lemma~\ref{lemma2}.
\end{proof}

\subsection{Proofs in Section~\ref{section:2}}
\walkpath*
\begin{proof}
$1\implies 2$: Suppose $\pi$ is a $Z$-$m$-open walk between $v_0=a$ and $v_n=b$ in $\mathcal{H}$. If a node $v$ appears more than once on $\pi$, let $v_i = v$ be its first occurrence and $v_j = v$ its last occurrence on $\pi$. Consider removing the subwalk between $v_i$ and $v_j$ by replacing the subwalk
\[
v_{i-1} \sus v_i \sus v_{i+1} \sus \cdots \sus v_{j-1} \sus v_j \sus v_{j+1}
\]
(where $v_{i-1}$ and $v_{j+1}$ may not exist if $v_i$ or $v_j$ is an endpoint) with the shorter path $v_{i-1} \sus v \sus v_{j+1}$.

\begin{enumerate}
    \item Suppose one of $v_i$ and $v_j$ is an endpoint of $\pi$; without loss of generality, let $v_i = v_0 = v$ be the starting node of $\pi$. Then the subwalk
    \[
    v_0 \sus v_1 \sus \cdots \sus v_{j-1} \sus v_j \sus v_{j+1}
    \]
    contains a repeated node $v = v_0 = v_j$. We can simplify this by replacing the entire subwalk from \( v_0 \) to \( v_j \) with the final edge of the subwalk, namely:
    \[
    v_0 = v_j \sus v_{j+1}.
    \]
    This replacement preserves the edge marks at $v_{j+1}$. Moreover, since $v_0$ is a non-collider and $v \notin Z$, the resulting walk remains $Z$-$m$-open.

    \item Assume $v_i$ and $v_j$ are non-endpoint nodes on $\pi$.
    \begin{enumerate}
        \item If both \( v_i \) and \( v_j \) are colliders on \( \pi \), then we can simplify the walk by replacing the subwalk between \( v_{i-1} \) and \( v_{j+1} \) with the shorter walk
        \[
        v_{i-1} \suh v_i = v_j \hus v_{j+1}.
        \]
        The resulting walk remains \( Z \)-\( m \)-open since the repeated node \( v \) is still a collider in the new path and, by assumption, is in \( \anc_\mathcal{H}(Z) \).

        \item If \( v_i \) and \( v_j \) are both non-colliders on \( \pi \), then we can attempt a similar simplification using the shorter walk \( v_{i-1} \sus v_i = v_j \sus v_{j+1} \). If the repeated node \( v \) remains a non-collider on the new path and \( v \notin Z \), then the new walk appears to be \( Z \)-\( m \)-open. However, we must account for the possibility that the shorter walk takes the form \( v_{i-1} \tut v_i = v_j \hus v_{j+1} \) or \( v_{i-1} \suh v_i = v_j \tut v_{j+1} \), which could block the path due to the edge marks. In such cases, a more careful analysis is needed to ensure the resulting walk remains \( Z \)-\( m \)-open. Without loss of generality, suppose the first case holds, i.e., \( v_{i-1} \tut v_i = v_j \hus v_{j+1} \) exists. Then the original walk \( \pi \) must contain the subwalk
        \[
        v_{k-1} \hut v_k \tut \cdots \tut v_{i-1} \tut v_i \tus \cdots \hut v_j \hus v_{j+1},
        \]
        where \( v_{k-1} \) may not exist if \( v_k = v_0 \) (i.e., the walk starts at \( v_k \)), and possibly \( v_k = v_{i-1} \).
        By Lemma~\ref{oldDef}, since \( \mathcal{H} \) contains the undirected walk \( v_k \tut \cdots \tut v_i = v_j \), it must also contain the edge \( v_k \hus v_{j+1} \). Therefore, we can simplify \( \pi \) by replacing the subwalk between \( v_k \) and \( v_{j+1} \) with the edge \( v_k \hus v_{j+1} \). This results in a shorter walk that remains \( Z \)-\( m \)-open, because \( v_k \), being a non-collider on both the original and new walks, is not in \( Z \), and the edge marks at other nodes are preserved.

        Now consider another special case in which the repeated node \( v \) becomes a collider on the new path, if we replace the subwalk of \( \pi \) between \( v_{i-1} \) and \( v_{j+1} \) with \( v_{i-1} \sus v_i = v_j \sus v_{j+1} \). In this case, we must have:
        \[
        v_{i-1} \suh v_i \tuh v_{i+1} \quad \text{and} \quad v_{j-1} \hut v_j \hus v_{j+1}.
        \]
        Then, the original walk \( \pi \) must contain a directed subwalk from \( v_i \) to some collider \( v_l \) for some \( i < l < j \) such that \( v_l \in \anc_{\mathcal{H}}(Z) \). Consequently, \( v_i \) (and thus \( v_j \)) is also an ancestor of \( Z \). Therefore, the repeated node \( v \) satisfies the collider condition for \( m \)-openness, and the new path remains \( Z \)-\( m \)-open.
                
        \item If \( v_i \) and \( v_j \) have different collider statuses on \( \pi \), without loss of generality, assume \( v_i \) is a non-collider and \( v_j \) is a collider on \( \pi \). Then we have either
        \[
        v_{i-1} \sus v_i \tus v_{i+1} \quad \text{or} \quad v_{i-1} \sut v_i \sus v_{i+1},
        \]
        and
        \[
        v_{j-1} \suh v_j \hus v_{j+1}.
        \]
        We attempt to simplify the walk by replacing the subwalk of \( \pi \) between \( v_{i-1} \) and \( v_{j+1} \) with the shorter walk $v_{i-1} \sus v_i = v_j \hus v_{j+1}$. This replacement preserves \( Z \)-\( m \)-openness because the repeated node \( v \) (i.e., \( v_i = v_j \)) was a non-collider at one occurrence and a collider at the other, satisfying both conditions: \( v \notin Z \) and \( v \in \anc_\mathcal{H}(Z) \). However, if \( \mathcal{H} \) contains the path
        \[
        v_{i-1} \tut v_i = v_j \hus v_{j+1},
        \]
        then the new walk would be blocked. In this case, similarly to the earlier scenario where both \( v_i \) and \( v_j \) are non-colliders, there exists a node \( v_k \) such that the subwalk of \( \pi \) from \( v_k \) to \( v_i \) is undirected:
        \[
        v_k \tut \cdots \tut v_i = v_j,
        \]
        and by Lemma~\ref{oldDef}, \( \mathcal{H} \) must contain the edge \( v_k \hus v_{j+1} \). Therefore, we can replace the entire subwalk of $\pi$ from \( v_k \) to \( v_{j+1} \) with the single edge \( v_k \hus v_{j+1} \), preserving the \( Z \)-\( m \)-openness of the walk.
    \end{enumerate}
\end{enumerate}

This replacement procedure can be iteratively applied until no node occurs more than once on the walk. And each replacement preserves $Z$-$m$-openness, so the final path remains $Z$-$m$-open. The resulting walk is then a $Z$-$m$-open \emph{path} between \( v_0 \) and \( v_n \).

$2 \implies 1$ is trivial since paths are walks.
\end{proof}

\lemmaseventeena*
\begin{proof}
W.L.O.G., assume $Z\cap\{a,b\}=\emptyset$. Suppose $\pi$ is in the form $a=v_0\sus\cdots\sus v_n=b$, and let \( v_k \) be an arbitrary node on the path \(\pi\). We will prove \( v_k \in \anc_\mathcal{G}(\{v_0, v_n\} \cup Z \cup \mathcal{S}) \) based on its position and properties.

If \( v_k \) is an endpoint, then \( v_k \in \{v_0, v_n\} \), which is straightforward.

If \( v_k \) is a collider, it follows that \( v_k \in \anc_\mathcal{H}(Z) \). Consequently, \( v_k \in \anc_\mathcal{G}(Z \cup \mathcal{S}) \).

Now consider the case where \( v_k \) is a non-collider. Without loss of generality, assume the subpath around \( v_k \) takes the form 
\[
v_0\sus \cdots\sus  v_{k-1}\sus v_k\tus v_{k+1}\sus \cdots\sus  v_n.
\]
If there is no edge of the form \(\hus\) on the subpath of \(\pi\) between \( v_k \) and \( v_n \), then \( v_k \in \anc_\mathcal{G}(\{v_n\} \cup \mathcal{S}) \).

Suppose there exists such an edge, with the corresponding node \( v_l \) in the configuration \( v_{l-1}\tus v_l\hus v_{l+1} \) on the subpath of \(\pi\) between \( v_k \) and \( v_n \). In this case, \( v_k \in \anc_\mathcal{G}(\{v_l\} \cup \mathcal{S}) \). Since \(\pi\) is \( Z \)-$m$-open, we have \( v_{l-1}\rightarrow v_l\hus v_{l+1} \). Then \( v_l \) is a collider, so \( v_l \in \anc_\mathcal{H}(Z) \), implying \( v_k \in \anc_\mathcal{G}(Z \cup \mathcal{S}) \).

From these cases, we conclude that every node on the path \(\pi\) belongs to \( \anc_\mathcal{G}(\{v_0, v_n\} \cup Z \cup \mathcal{S}) \).
\end{proof}

\lemmaseventeenb*
\begin{proof}
W.L.O.G., assume $Z\cap\{a,b\}=\emptyset$. Let $\pi$ be a $Z$-$m$-open path between $a$ and $b$ in $\mathcal{H}$ as follows:
$$
a=v_0\sus v_1\sus\cdots\sus v_{n-1}\sus v_n=b.
$$
For each $i = 1, \ldots, n$, let $\mu_i$ be a $\sigma$-inducing path given $\mathcal{S}$ in $\mathcal{G}$ between $v_{i-1}$ and $v_i$, where $\mu_i$ is into $v_{i-1}$ if $v_{i-1} \hus  v_i$ on $\pi$, and into $v_i$ if $v_{i-1} \suh v_i$ on $\pi$ (Lemma~\ref{lemma0} guarantees this). Concatenating all paths $(\mu_i)_{i=1,\ldots,n}$ gives a walk $\mu$. 

We aim to show that there exists a walk in $\mathcal{G}$ between $v_0$ and $v_n$ that satisfies the following properties:
\begin{enumerate}
    \item All colliders on the walk are in $\anc_\mathcal{G}(\{v_0, v_n\} \cup Z \cup \mathcal{S})$.
    \item All non-colliders on the walk are either not in $Z \cup \mathcal{S}$ or are unblockable.
\end{enumerate}
Such a walk exists because $\mu$ satisfies these properties, as we now verify.

By Lemma~\ref{lemma17a}, $v_i \in \anc_\mathcal{G}(\{v_0, v_n\} \cup Z \cup \mathcal{S})$ for all $i = 1, \ldots, n$. This holds in particular for all $v_i$ that are colliders on $\mu$. Furthermore, every non-endpoint collider on $\mu_i$ is in $\anc_\mathcal{G}(\{v_{i-1}, v_i\} \cup \mathcal{S})$, and hence also in $\anc_\mathcal{G}(\{v_0, v_n\} \cup Z \cup \mathcal{S})$. For non-colliders, observe that all non-endpoint non-colliders on $\mu_i$ are unblockable. Now consider $v_i$ that are non-colliders on $\mu$. By construction, such $v_i$ are also non-colliders on $\pi$. Since $\pi$ is $Z$-$m$-open, we have $v_i \notin Z$, and obviously $v_i \notin \mathcal{S}$. Therefore, all non-colliders are either not in $Z \cup \mathcal{S}$ or are unblockable.

Let $\nu$ be a walk satisfying the above properties, with the minimal number of colliders, and such that $v_0$ and $v_n$ appear only once on $\nu$. We claim that all colliders on $\nu$ must be in $\anc_\mathcal{G}(Z \cup \mathcal{S})$.

Suppose, for contradiction, that there exists a collider $v_k$ on $\nu$ such that $v_k \notin \anc_\mathcal{G}(Z \cup \mathcal{S})$ but $v_k \in \anc_\mathcal{G}(\{v_0, v_n\})$. Then there must exist a directed path from $v_k$ to $v_0$ that does not pass through $v_n$, or a directed path from $v_k$ to $v_n$ that does not pass through $v_0$. Without loss of generality, assume the former. Notice that every node on the directed path from $v_k$ to $v_0$ is not in $Z \cup \mathcal{S}$. Replacing the subpath of $\nu$ from $v_k$ to $v_0$ with this directed path would result in a walk with fewer colliders between $v_0$ and $v_n$ that still satisfies the conditions, contradicting the minimality of $\nu$.

Thus, all colliders on $\nu$ must be in $\anc_\mathcal{G}(Z \cup \mathcal{S})$. Consequently, $\nu$ is a $\sigma$-open walk given $Z \cup \mathcal{S}$ in $\mathcal{G}$ between $v_0$ and $v_n$. Therefore, by Proposition~\ref{proposition3.3.6}, there must exist a $Z \cup \mathcal{S}$-$\sigma$-open path between $a$ and $b$ in $\mathcal{G}$.
\end{proof}

\lemmaeighteen*
\begin{proof}
W.L.O.G., assume $Z\cap\{a,b\}=\emptyset$. This proof follows the approach used in proving Lemma 18 from \citep{spirtes1996polynomial} and incorporates the segment-based $\sigma$-separation introduced in \citep{forre2017markov} (Definition~\ref{segmentsigma}).

To begin, given the strongly connected components of $\mathcal{G}$, we pick a $\sigma$-open path $\pi$ conditioned on $Z \cup \mathcal{S}$, uniquely expressed in segment form as:
\[
\sigma_0 \sus \sigma_1 \sus \cdots \sus \sigma_{n-1} \sus \sigma_n,
\]
where $\sigma_{0,l} = a$ and $\sigma_{n,r} = b$. By the definition of segment-based $\sigma$-separation:
\begin{enumerate}
    \item $a, b \notin Z \cup \mathcal{S}$.
    \item For all non-collider segments, their endpoints corresponding to outgoing directed edges are not in $Z \cup \mathcal{S}$.
    \item For all collider segments, their nodes intersect non-trivially with $\anc_\mathcal{G}(Z \cup \mathcal{S})$.
\end{enumerate}

We construct a sequence of nodes $Q_0$ based on the segments of $\pi$. Initialize $Q_0(1) = a$. For each segment $\sigma_i$: If $\sigma_i$ is a non-collider segment, include its endpoints corresponding to outgoing directed edges in $Q_0$. If $\sigma_i$ is a collider segment and its intersection with $\anc_\mathcal{G}(\mathcal{S})$ is empty, add an arbitrary endpoint from its ends to $Q_0$. Finally, add \( b \) to \( Q_0 \) if it is not already included, and set \( Q_0(m) = b \).

We now show that for $i = 1, \dots, m - 1$, the nodes $Q_0(i)$ and $Q_0(i+1)$ are adjacent in $\mathcal{H}$. By construction, the subpath of $\pi$ between $Q_0(i)$ and $Q_0(i+1)$ includes only unblockable non-colliders or colliders that are ancestors of $\mathcal{S}$. Such a path is $\sigma$-inducing given $\mathcal{S}$, which guarantees that $Q_0(i)$ and $Q_0(i+1)$ are adjacent in $\mathcal{H}$. Consequently, $Q_0$ forms a path $\pi_0$ from $a$ to $b$ in $\mathcal{H}$:
$$
a = Q_0(1)\sus\cdots\sus Q_0(m) = b.
$$
Next, we show that if $Q_0(i)$ originates from a non-collider segment on $\pi$, it remains a non-collider on $\pi_0$. Without loss of generality, assume $Q_0(i)$ is the right endpoint of a non-collider segment, with an outgoing directed edge to the right. In this case, $Q_0(i)$ is either an ancestor of the right endpoint of the next non-collider segment or a collider segment's endpoint, which belongs to $\anc_\mathcal{G}(Z \cup \mathcal{S})$. Thus, $Q_0(i)$ and $Q_0(i+1)$ are connected by an edge $Q_0(i) \tus  Q_0(i+1)$ in $\mathcal{H}$. Moreover, if \( Q_0(i) \notin \anc_\mathcal{G}(\mathcal{S}) \), meaning that \( Q_0(i) \) has a directed edge into the segment of \( \pi \) containing \( Q_0(i+1) \), then \( Q_0(i+1) \) cannot be an ancestor of \( Q_0(i) \), as they lie in different strongly connected components. Similarly, suppose \( Q_0(i) \) is the left endpoint of a non-collider segment on \( \pi \), with an outgoing directed edge \( \hut \) to the left. Then \( Q_0(i-1) \sut Q_0(i) \) holds in \( \mathcal{H} \). Moreover, if \( Q_0(i) \notin \mathcal{S} \), it follows that \( Q_0(i-1) \notin \anc_\mathcal{G}(Q_0(i)) \).

For nodes $Q_0(i)$ originating from collider segments, we cannot guarantee that $Q_0(i)$ will remain a collider on $\pi_0$. To resolve this, we employ the following algorithm to remove problematic nodes from the sequence $Q_0$:

\begin{algorithm}
\begin{algorithmic}[1]
\STATE $j \leftarrow 0$
\REPEAT
    \STATE $I\leftarrow$ \{$0<i<m$: $Q_j(i)\in\anc_\mathcal{G}(\{Q_j(i-1),Q_j(i+1)\})$ and $Q_j(i)$ is from a collider segment on $\pi$\}
    \IF{$I\neq\emptyset$} 
        \STATE form sequence $Q_{j+1}$ from $Q_j$ by removing some $Q_j(i)$ with $i\in I$
        \STATE $j \leftarrow j + 1$, $m\leftarrow m-1$
    \ENDIF
\UNTIL{$I=\emptyset$}
\STATE $k\leftarrow j$
\end{algorithmic}
\end{algorithm}

We now show that, in each intermediate sequence \( Q_j \), every pair of consecutive nodes are adjacent in \( \mathcal{H} \), thus forming a valid path \( \pi_j \) in \( \mathcal{H} \). This property has already been shown for \( Q_0 \). Now, assume the property holds for \( Q_j \), and consider the \((j+1)\)-th step in which the node \( Q_j(i) \) is removed. In this case, we know that \( Q_j(i) \in \anc_\mathcal{G}(\{Q_j(i-1), Q_j(i+1)\}) \). The subpath of \( \pi \) between \( Q_j(i-1) \) and \( Q_j(i+1) \) is \(\sigma\)-inducing given \( \mathcal{S} \), since \( Q_j(i) \) is either a collider in \( \anc_\mathcal{G}(\{Q_j(i-1), Q_j(i+1)\}) \), or an unblockable non-collider on \( \pi \), and all other nodes on this subpath are either unblockable non-colliders or colliders in \( \anc_\mathcal{G}(\{Q_j(i-1), Q_j(i+1)\} \cup \mathcal{S}) \). Consequently, we conclude that \( Q_{j+1}(i-1) \equiv Q_j(i-1) \) and \( Q_{j+1}(i) \equiv Q_j(i+1) \) are adjacent in \( \mathcal{H} \).

We now show that if \( Q_k(i) \) comes from a non-collider segment on \( \pi \), it remains a non-collider on \( \pi_k \). Moreover, if \( Q_k(i) \notin \anc_\mathcal{G}(\mathcal{S}) \), then \( Q_k(i+1) \notin \anc_\mathcal{G}(Q_k(i)) \) if it is the right endpoint with edge \( \tuh \), and \( Q_k(i-1) \notin \anc_\mathcal{G}(Q_k(i)) \) if it is the left endpoint with edge \( \hut \). This has been proved for \( Q_0 \). Suppose the result holds for \( Q_j \). Then, if \( Q_j(i) \) is from a non-collider segment on \( \pi \), it remains a non-collider on \( \pi_j \). Furthermore, \( Q_j(i) \) retains the same edge marks on \( \pi_{j+1} \), unless \( Q_j(i)\tus Q_j(i+1) \) is present on \( \pi_j \) and \( Q_j(i+1) \) is removed at the \( (j+1) \)-th step, or \( Q_j(i-1)\sut Q_j(i) \) and \( Q_j(i-1) \) is removed at the \( (j+1) \)-th step. Without loss of generality, we assume the former case. Since \( Q_j(i)\tus Q_j(i+1) \) is in \( \mathcal{H} \), $Q_j(i)\in\anc_\mathcal{G}(\{Q_j(i+1)\}\cup\mathcal{S})$. If \( Q_j(i) \) is an ancestor of \( \mathcal{S} \), it is obviously a non-collider on \( \pi_{j+1} \). Therefore, we need only consider the case where \( Q_j(i) \notin \anc_\mathcal{G}(\mathcal{S}) \) but \( Q_j(i) \in \anc_\mathcal{G}(Q_j(i+1)) \). Since \( Q_j(i+1) \) is removed at the \( (j+1) \)-th step, we know that \( Q_j(i+1) \) is an ancestor of either \( Q_j(i) \) or \( Q_j(i+2) \). Given that \( Q_j(i+1) \notin \anc_\mathcal{G}( Q_j(i) ) \), we conclude that \( Q_j(i+1) \in \anc_\mathcal{G}( Q_j(i+2) ) \). Thus, \( Q_{j+1}(i) \equiv Q_j(i) \) is also an ancestor of \( Q_{j+1}(i+1) \equiv Q_{j}(i+2) \), meaning \( Q_{j+1}(i) \) is a non-collider on \( \pi_{j+1} \). Additionally, if \( Q_{j+1}(i) \notin \anc_\mathcal{G}(\mathcal{S}) \), we have \( Q_{j+1}(i+1) \notin \anc_\mathcal{G}( Q_{j+1}(i) ) \). Otherwise, this would imply that \( Q_j(i+1) \in \anc_\mathcal{G}( Q_j(i) ) \), which leads to a contradiction.

Given that all nodes from non-collider segments on \( \pi \) remain non-colliders on \( \pi_k \) and lie outside \( Z \cup \mathcal{S} \), we now consider nodes on \( \pi_k \) originating from collider segments. For any such node \( Q_k(i) \), we have \( Q_k(i) \notin \anc_\mathcal{G}(\{Q_k(i-1), Q_k(i+1)\}) \); otherwise, it would have been removed by the algorithm. Moreover, since all such \( Q_k(i) \) lie in collider segments on \( \pi \), we have \( Q_k(i) \in \anc_\mathcal{G}(Z) \setminus \anc_\mathcal{G}(\mathcal{S}) \). Therefore, \( Q_k(i) \) must be a collider on \( \pi_k \), and in particular, satisfies \( Q_k(i) \in \ant_\mathcal{H}(Z) \).

So far, we have shown that all non-colliders on \( \pi_k \) are not in \( Z \). To prove that \( \pi_k \) is \( Z \)-$m$-open, we only need to check all potentially present subpaths of \( \pi_k \) in the form of \( Q_k(i-1) \suh Q_k(i) \tut Q_k(i+1) \) or \( Q_k(i-1) \tut Q_k(i) \hus Q_k(i+1) \). Suppose it contained a subpath of the form \( Q_k(i-1) \suh Q_k(i) \tut Q_k(i+1) \). Then $Q_k(i),Q_k(i+1)$ must be in the same strongly connected component in $\mathcal{G}$, and both are not in $\anc_\mathcal{G}(\mathcal{S})$. Since \( Q_k(i) \) lies in a non-collider segment on \( \pi \), but we have shown that if \( Q_k(i) \notin \anc_\mathcal{G}(\mathcal{S}) \), then \( Q_k(i+1) \notin \anc_\mathcal{G}(Q_k(i)) \), this leads to a contradiction. One can show in a similar way that it cannot contain a subpath \( Q_k(i-1) \tut Q_k(i) \hus Q_k(i+1) \) either.

We now have shown the existence of a path $\pi_k$ in $\mathcal{H}$ between $a$ and $b$ that almost qualifies for being $m$-open given $Z$, except that some of its colliders may not be in $\anc_\mathcal{H}(Z)$ (yet all colliders are in $\ant_\mathcal{H}(Z)$. By Lemma~\ref{oldDef}, there exist nodes in \( \mathcal{H} \) that lie in \( \anc_\mathcal{H}(Z) \), allowing us to replace the problematic colliders on \( \pi_k \) with these nodes while preserving their collider status on the path. As a result, we obtain an \( m \)-open path \( \pi_k' \) between \( a \) and \( b \) in \( \mathcal{H} \).
\end{proof}

\thmtwo*
\begin{proof}
Obviously obtained by Lemma~\ref{lemma17b} and Lemma~\ref{lemma18}.
\end{proof}

\inducingPropertyOne*
\begin{proof}
$1 \implies 2$ is trivial.

$2 \implies 3$: Assume there exists an inducing walk between \( a \) and \( b \) in \( \mathcal{H} \). Let \( Z \subseteq  \mathcal{V} \setminus \{a,b\} \). Consider all walks in \( \mathcal{H} \) between \( a \) and \( b \) such that all colliders on the walk lie in \( \anc_{\mathcal{H}}(\{a,b\} \cup Z) \), all non-endpoint non-colliders are not in \( Z \), and there is no arrowhead into an undirected edge along the walk. Such walks exist, since the inducing walk is one of them. Let \( \mu \) be such a walk with a minimal number of colliders. We claim that all colliders on \( \mu \) must lie in \( \anc_{\mathcal{H}}(Z) \). Suppose, for contradiction, that there exists a collider \( c \) on \( \mu \) that is not an ancestor of \( Z \). By assumption, \( c \) must be an ancestor of either \( a \) or \( b \). Without loss of generality, assume \( c \in \anc_{\mathcal{H}}(a) \). Therefore, there exists a directed walk \( \pi \) from \( c \) to \( a \) in \( \mathcal{H} \) that does not pass through any node in \( Z \). Concatenating this walk with the subwalk of \( \mu \) between \( c \) and \( b \) yields another walk between \( a \) and \( b \) with the same properties but fewer colliders, contradicting the minimality of \( \mu \). Hence, all colliders on \( \mu \) must lie in \( \anc_{\mathcal{H}}(Z) \), implying that \( \mu \) is \( m \)-open given \( Z \). By Proposition~\ref{walk eq path}, this ensures the existence of a \( Z \)-\( m \)-open path between \( a \) and \( b \), and thus \( a \) and \( b \) are \( m \)-connected given \( Z \).

$3 \implies 4$ is trivial.

$4 \implies 1$: Suppose that \( a \) and \( b \) are \( m \)-connected given \( Z = \anc_{\mathcal{H}}(\{a,b\})  \setminus \{a,b\} \). Let \( \pi \) be a path between \( a \) and \( b \) that is \( m \)-open given \( Z \). We claim that \( \pi \) must be an inducing path. First, all colliders on \( \pi \) are in \( \anc_{\mathcal{H}}(Z) \), and hence must be in \( \anc_{\mathcal{H}}(\{a,b\}) \). Second, all non-endpoint nodes on \( \pi \) must be colliders. Suppose, for contradiction, that there exists a non-collider \( c \) on \( \pi \). Then there must be a directed subpath of \( \pi \) starting at \( c \) and ending either at the first collider after \( c \) or at one of the endpoints \( a \) or \( b \). This implies that \( c \in \anc_{\mathcal{H}}(\{a,b\}) \), and hence \( c \in Z \). But this contradicts the assumption that \( \pi \) is \( m \)-open given \( Z \), as it would contain a non-collider in \( Z \). Therefore, all non-endpoint nodes on \( \pi \) are colliders in \( \anc_{\mathcal{H}}(\{a,b\}) \), and thus \( \pi \) is an inducing path.
\end{proof}

\inducingPropertyTwo*
\begin{proof}
Let \( \mu \) be an inducing path between \( a \) and \( b \) in \( \mathcal{H} \). If \( \mu \) contains only the two endpoints \( a \) and \( b \), then \( \mu \) must be of the form \( a \sut b \), which implies that \( b \in \ant_{\mathcal{H}}(a) \). If \( \mu \) contains more than two nodes, then it must be of the form \( a \sus \cdots \sus c \hut b \), where \( c \) is a collider lying in \( \anc_{\mathcal{H}}(\{a,b\}) \). Since \( \mathcal{H} \) is ancestral, \( c \) cannot be an ancestor of \( b \); hence, it must be that \( c \in \anc_{\mathcal{H}}(a) \), which implies \( b \in \anc_{\mathcal{H}}(a) \).
\end{proof}

\inducingPropertyThree*
\begin{proof}
Let \( \mu \) be an inducing path between \( a \) and \( b \) in \( \mathcal{H} \) that is into \( b \). If \( \mu \) contains only the two endpoints \( a \) and \( b \), then it must be of the form \( a \suh b \). Since \( a \notin \anc_{\mathcal{H}}(b) \), \( \mu \) must be of the form \( a \huh b \). If \( \mu \) contains more than two nodes, then it must be of the form \( a \suh c \huh \cdots \huh b \), where \( c \) is a collider that lies in \( \anc_{\mathcal{H}}(\{a, b\}) \). Since \( \mathcal{H} \) is ancestral, \( c \) cannot be an ancestor of \( a \). Moreover, if \( \mu \) is not into \( a \), then \( c \) cannot be an ancestor of \( b \) either, because \( a \notin \anc_{\mathcal{H}}(b) \), which contradicts the assumption that \( \mu \) is an inducing path. Therefore, we conclude that \( \mu \) is into both \( a \) and \( b \).
\end{proof}

\subsection{Proofs in Section~\ref{section:3}}
\lemmasix*
\begin{proof}
Suppose $Z\subseteq\mathcal{V}\setminus\{a,c\}$ is a subset of nodes such that $a$ and $c$ are $m$-separated given $Z$. By the definition of discriminating paths, $v_0$ is a collider on $\pi$ and has a directed edge into $c$. Then, ${\mathcal{H}}$ contains a path:
\[
a \suh v_0 \tuh c.
\]
Since $v_0$ is a non-collider on this path, $Z$ must contain $v_0$ to $m$-block the path. Furthermore, ${\mathcal{H}}$ contains another path:
\[
a \suh v_0 \huh v_1 \tuh c.
\]
Since $v_0 \in Z$, we have $v_0 \in \anc_{\mathcal{H}}(Z)$. Moreover, $v_1$ is a non-collider on this path, implying that $Z$ must also contain $v_1$. By induction, all nodes $v_k$ for $k=0,\dots,n$ belong to $Z$. Consequently, all nodes on $\pi$, except for $b$, do not $m$-block $\pi$, so $b$ must $m$-block $\pi$. Thus, we conclude:
\begin{enumerate}
    \item If $b$ is a \emph{collider} on $\pi$, then $b \notin Z$.
    \item If $b$ is a \emph{non-collider} on $\pi$, then $b \in Z$.
\end{enumerate}
\end{proof}

\thmthreea*
\begin{proof}
Suppose that $\mathcal{H}_1$ and $\mathcal{H}_2$ do not satisfy Condition~\ref{condition}. Assume that $a$ and $b$ are adjacent in $\mathcal{H}_1$ but not in $\mathcal{H}_2$. Since the edge between $a$ and $b$ in $\mathcal{H}_1$ is an inducing path, by Proposition~\ref{inducing property one}, it follows that $a \overset{m}{\underset{\mathcal{H}_1}{\not\perp}} b \mid Z$ for $Z = \anc_{\mathcal{H}_1}(\{a,b\}) \setminus \{a,b\}$. As $\mathcal{H}_1$ and $\mathcal{H}_2$ are $m$-Markov equivalent, we also have $a \overset{m}{\underset{\mathcal{H}_2}{\not\perp}} b \mid Z$ for the same set $Z$. Applying Proposition~\ref{inducing property one} again, this implies the existence of an inducing path between $a$ and $b$ in $\mathcal{H}_2$, contradicting the maximality of $\mathcal{H}_2$. Therefore, $\mathcal{H}_1$ and $\mathcal{H}_2$ must have the same adjacencies.

Now assume that \( v_k \) is an unshielded collider of the form \( v_{k-1} \suh v_k \hus v_{k+1} \) in \( \mathcal{H}_1 \), but not in \( \mathcal{H}_2 \). Since \( v_{k-1} \) and \( v_{k+1} \) are non-adjacent in \( \mathcal{H}_2 \), there exists no inducing path between them. By Proposition~\ref{inducing property one}, we have \(v_{k-1} \overset{m}{\underset{\mathcal{H}_2}{\perp}} v_{k+1} \mid Z,\) where \( Z = \anc_{\mathcal{H}_2}(\{v_{k-1}, v_{k+1}\}) \setminus \{v_{k-1}, v_{k+1}\} \). Since \( \mathcal{H}_1 \) and \( \mathcal{H}_2 \) are \( m \)-Markov equivalent, it follows that \(v_{k-1} \overset{m}{\underset{\mathcal{H}_1}{\perp}} v_{k+1} \mid Z\) holds for the same set \( Z \). In \( \mathcal{H}_2 \), the path between \( v_{k-1} \) and \( v_{k+1} \) passes through \( v_k \) as a non-collider (either \( v_{k-1} \sus v_k \tus v_{k+1} \) or \( v_{k-1} \sut v_k \sus v_{k+1} \)), which implies that \( v_k \in Z \). However, in \( \mathcal{H}_1 \), \( v_k \) is a collider on the path \( v_{k-1} \suh v_k \hus v_{k+1} \), which implies \( v_k \notin Z \). This contradiction shows that \( \mathcal{H}_1 \) and \( \mathcal{H}_2 \) must have the same unshielded colliders.

Finally, assume \( \pi \) is a discriminating path between \( a \) and \( c \) for a node \( b \) in \( \mathcal{H}_1 \), and let \( \pi' \) be the corresponding path in \( \mathcal{H}_2 \), which is also a discriminating path for \( b \). By similar reasoning, we have $a \overset{m}{\underset{\mathcal{H}_1,\mathcal{H}_2}{\not\perp}} c \mid Z$, for \( Z = \anc_{\mathcal{H}_1}(\{a, c\}) \setminus \{a, c\} \), since \( a \) and \( c \) are not adjacent. If \( b \) is a collider on \( \pi \) but not on \( \pi' \), then by Lemma~\ref{lemma6}, it follows that \( b \notin Z \) in \( \mathcal{H}_1 \), whereas \( b \in Z \) in \( \mathcal{H}_2 \). This contradiction implies that \( b \) is a collider on \( \pi \) if and only if it is a collider on \( \pi' \).
\end{proof}


\lemmanine*
\begin{proof}
This proof is inspired by the Lemma $9$ in \citep{spirtes1996polynomial}, with modifications to adapt it to \( \sigma \)-MAGs.

If \( v_j \) is a covered node on \( \pi \), then \( \mathcal{H} \) must contain one of the six subgraphs described in Lemma~\ref{lemma8}. In particular, if \( v_{j-1} \rightarrow v_{j+1} \), then \( \pi \) contains the subgraph Figure~\ref{lemma9.1}. Similarly, if \( v_{j-1} \hut v_{j+1} \), then \( \pi \) contains the subpath Figure~\ref{lemma9.2}. Since these two cases are symmetric, we assume W.L.O.G. the former.

\begin{figure}[t]
\centering
\subfloat[One possible case when \( v_{j-1} \rightarrow v_{j+1} \) \label{lemma9.1}]{
\begin{tikzpicture}
    \begin{scope}
      \node[ndout] (-2) at (-3,0) {$v_{j-2}$};
      \node[ndout] (-1) at (-1,0) {$v_{j-1}$};
      \node[ndout] (0) at (1,0) {$v_j$};
      \node[ndout] (1) at (3,0) {$v_{j+1}$};
      \draw[suh] (-2) edge (-1);
      \draw[hus] (-1) edge (0);
      \draw[suh] (0) edge (1);
      \draw[tuh, bend left = 40] (-1) edge (1);
    \end{scope}
  \end{tikzpicture}
}\\

\subfloat[The other possible case when \( v_{j-1} \hut v_{j+1} \) \label{lemma9.2}]{
\begin{tikzpicture}
    \begin{scope}
      \node[ndout] (-1) at (-3,0) {$v_{j-1}$};
      \node[ndout] (0) at (-1,0) {$v_{j}$};
      \node[ndout] (1) at (1,0) {$v_{j+1}$};
      \node[ndout] (2) at (3,0) {$v_{j+2}$};
      \draw[hus] (-1) edge (0);
      \draw[suh] (0) edge (1);
      \draw[hus] (1) edge (2);
      \draw[hut, bend left = 40] (-1) edge (1);
    \end{scope}
  \end{tikzpicture}
}
\caption{Illustration for Proposition~\ref{lemma9}.}
\end{figure}
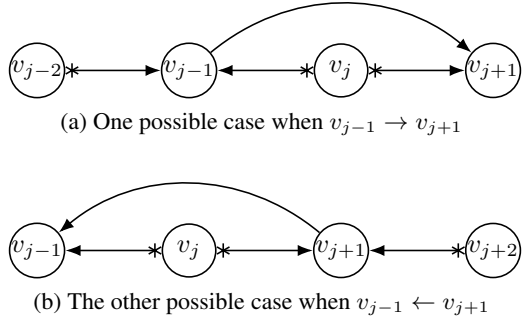

We aim to find a node \( v_k \) with $k\leq j-1$ such that the subpath of \( \pi \) between \( v_k \) and \( v_{j+1} \) consists of at least \( n \geq 2 \) edges, is into \( v_k \), all non-endpoint nodes on the subpath (except the endpoints and $v_j$) are colliders, and every node on the subpath (excluding \( v_j \) and \( v_{j+1} \)) has a directed edge into \( v_{j+1} \) in \( \mathcal{H} \). Such nodes do exist, as the subpath of \( \pi \) between \( v_{j-1} \) and \( v_{j+1} \) satisfies these conditions for \( n = 2 \). We will show that the subpath of \( \pi \) between \( v_{k-1} \) and \( v_{j+1} \) is either a discriminating path for \( v_j \), or a subpath with \( n+1 \) edges that satisfies the required conditions, in which \( q_{k-1} \) is a collider on \( \pi \) and \( \mathcal{H} \) contains the edge \( q_{k-1} \tuh q_{j+1} \).

By Lemma~\ref{lemma7}, \( v_k \) and \( v_{j+1} \) are adjacent in \( \mathcal{H} \) but not on \( \pi \), and we have \( v_k \rightarrow v_{j+1} \), so there must be an edge \( v_{k-1} \suh v_k \) on \( \pi \). Notice that all non-endpoint nodes between \( v_{k-1} \) and \( v_j \) are colliders on \( \pi \) and have directed edges into \( v_{j+1} \). If the subpath of \( \pi \) between \( v_{k-1} \) and \( v_{j+1} \) is not a discriminating path for \( v_j \), then \( v_{k-1} \) and \( v_{j+1} \) must be adjacent. Applying Lemma~\ref{lemma7} again, the edge between \( v_{k-1} \) and \( v_{j+1} \) must be directed. If we have \( v_{k-1} \leftarrow v_{j+1} \) in \( \mathcal{H} \), then a directed cycle or an almost directed cycle would exist, such as \( v_k \rightarrow v_{j+1} \rightarrow v_{k-1} \suh v_k \), which leads to a contradiction. Therefore, we conclude that \( v_{k-1} \rightarrow v_{j+1} \), and there is an edge \( v_{k-2} \suh v_{k-1} \huh  v_k \) on \( \pi \).

If the subpath of \( \pi \) between \( v_{k-1} \) and \( v_{j+1} \) is not a discriminating path for \( v_j \), then the subpath between \( v_{k-2} \) and \( v_{j+1} \) might be a discriminating path for \( v_j \). If it is not, we move to the next node on the left. Notice that if the subpath of \( \pi \) between \( v_1 \) and \( v_{j+1} \) satisfies the conditions above, then \( v_0 \) is not adjacent to \( v_{j+1} \). Otherwise, by Lemma~\ref{lemma7}, the edge between \( v_0 \) and \( v_{j+1} \) must be \( v_0 \leftarrow v_{j+1} \), which would create a directed or almost directed cycle \( v_1 \rightarrow v_{j+1} \rightarrow v_0 \suh v_1 \) in \( \mathcal{H} \). Since \( \pi \) is of finite length, it follows that there must exist a node \( v_i \) between \( v_0 \) and \( v_j \) such that the subpath of \( \pi \) between \( v_i \) and \( v_{j+1} \) is a discriminating path for \( v_j \).

We will show that this discriminating path is unique. Let \( \pi_i \) denote this path. No subpath of \( \pi_i \) can be a discriminating path for \( v_j \) since all nodes on \( \pi_i \), except its endpoints, are adjacent to \( v_{j+1} \). Moreover, no subpath of \( \pi \) that contains \( \pi_i \) can be a discriminating path for \( v_j \) since \( v_i \) is not adjacent to \( v_{j+1} \). Therefore, we conclude that \( \pi_i \) is unique.
\end{proof}

\lemmaseven*
\begin{proof}
This proof is inspired by the Lemma $7$ in \citep{spirtes1996polynomial}, with modifications to adapt it to \( \sigma \)-MAGs.

Since $\pi$ is a shortest $m$-open path between $v_0$ and $v_n$ given $Z$, it follows that $\pi'$ is $m$-blocked by $Z$. Therefore, one of the following must hold:
\begin{enumerate}
    \item There exists a subpath of the form:
    \begin{itemize}
        \item \( v_{i-1} \suh v_i - v_j \), or
        \item \( v_{i-1} - v_i \hus  v_j \), or
        \item \( v_i - v_j \hus  v_{j+1} \), or
        \item \( v_i \suh v_j - v_{j+1} \).
    \end{itemize}
    \item One of \( v_i \) or \( v_j \) is either a non-collider on \( \pi' \) that lies in \( Z \), or a collider on \( \pi' \) that does not belong to \( \anc_H(Z) \).
\end{enumerate}

Assume one of the above subpaths is present on \( \pi' \). Suppose we have \( v_{i-1} \suh v_i - v_j \) on \( \pi' \). If \( v_j \) is an endpoint of \( \pi' \), then we can substitute \( v_{i-1} \suh v_i - v_j \) with \( v_{i-1} \suh v_j \) and obtain a shorter \( Z \)-$m$-open path, leading to a contradiction. If we have \( v_j \hus  v_{j+1} \) on \( \pi' \), then we can substitute \( v_{i-1} \suh v_i - v_j \hus  v_{j+1} \) with \( v_{i-1} \suh v_i \hus  v_{j+1} \), which results in a \( Z \)-$m$-open path because \( v_i \) is either a collider or has a directed subpath into a collider on \( \pi \). If we have \( v_j \tus  v_{j+1} \) on \( \pi' \), then we can substitute \( v_{i-1} \suh v_i - v_j \tus  v_{j+1} \) with \( v_{i-1} \suh v_j \tus  v_{j+1} \), which results in a shorter \( Z \)-$m$-open path because \( v_j \) is a non-collider on $\pi$ and \( \pi' \). The case \( v_i - v_j \hus  v_{j+1} \) is analogous. Suppose we have \( v_{i-1} - v_i \hus  v_j \) on \( \pi' \). Since there exists a node \( v_k \) on \( \pi' \) such that \( v_{k-1} \leftarrow v_k - \cdots - v_i \hus  v_j \) (where \( v_{k-1} \) may not appear), we can substitute \( v_k - \cdots - v_i \hus  v_j \) with \( v_k \hus  v_j \), which produces a shorter \( Z \)-$m$-open path since \( v_k \) is a non-collider on \( \pi \). The case \( v_i \suh v_j - v_{j+1} \) is similar. Hence, one of \( v_i \) and \( v_j \) must block \( \pi' \).

So the above subpaths cannot exist on \( \pi' \). If \( v_i \) is a non-collider on \( \pi \) but not on \( \pi' \), then we have \( v_{i-1} \suh v_i \rightarrow v_{i+1} \) on \( \pi \) and \( v_i \hus  v_j \) on $\pi'$. If there is no collider on the subpath between \( v_i \) and \( v_j \) on \( \pi \), then \( v_i \) is an ancestor of \( v_j \), which leads to a directed cycle or an almost directed cycle in \( \mathcal{H} \) since we have \( v_i \hus  v_j \). It follows that there must be a collider on the subpath, and \( v_i \) is an ancestor of the first collider on the subpath, implying \( v_i \in \anc_\mathcal{H} (Z) \). Hence, \( v_i \) does not block \( \pi' \). 

Similarly, if \( v_j \) is a non-collider on \( \pi \) but not on \( \pi' \), then \( v_j \) does not block \( \pi' \). Therefore, either \( v_i \) is a collider on \( \pi \) but not on \( \pi' \), or \( v_j \) is a collider on \( \pi \) but not on \( \pi' \). Hence, we have \( v_{i-1}\suh v_i \hus  v_{i+1} \) on \( \pi \) and \( v_i \rightarrow v_j \) in \( \mathcal{H} \), or we have \( v_{j-1} \suh v_j\hus v_{j+1} \) on \( \pi \) and \( v_i \leftarrow v_j \) in \( \mathcal{H} \).
\end{proof}

\lemmaeight*
\begin{proof}
Directly applying Lemma~\ref{lemma7}, and noting that \( \mathcal{H} \) does not contain a directed cycle or an almost directed cycle, we obtain the remaining six cases. 
\end{proof}


\lemmaten*
\begin{proof}
This proof is inspired by the Lemma 10 in \citep{spirtes1996polynomial}, with modifications to adapt it to \( \sigma \)-MAGs.

Suppose, for the sake of contradiction, that \( v_b \) is a covered node on the discriminating subpath of \( \pi \) between \( v_i \) and \( v_k \) for \( v_j \) (denoted by \( \pi_j \), where \( v_j \) is adjacent to \( v_k \)), and that \( v_j \) is a covered node on the discriminating subpath of \( \pi \) between \( v_a \) and \( v_c \) for \( v_b \) (denoted by \( \pi_b \), where \( v_b \) is adjacent to \( v_c \)).

W.L.O.G., assume that $a < b = c - 1$. Since $v_b$ and $v_j$ are two distinct non-endpoint nodes on $\pi_j$ and $\pi_b$, respectively, we must have the ordering $k+1 = j < i$. Otherwise, a contradiction arises: namely, $i < b < j$ and $a < j < b$ cannot simultaneously hold. Therefore, the paths $\pi_b$ and $\pi_j$ are illustrated in Figures~\ref{lemma10.1} and \ref{lemma10.2}.

\begin{figure}[t]
    \centering
    \subfloat[$v_j$ on the path $\pi_b$\label{lemma10.1}]{
    \begin{tikzpicture}
            \node[ndout] (va) at (-4,0) {$v_a$};
            \node (dots1) at (-2.5,0) {$\cdots$};
            \node[ndout] (vj) at (-1,0) {$v_j$};
            \node (dots2) at (0.5,0) {$\cdots$};
            \node[ndout] (vb) at (2,0) {$v_b$};
            \node[ndout] (vc) at (3.5,0) {$v_c$ };
            \draw[suh] (va) edge (dots1);
            \draw[huh] (dots1) edge (vj);
            \draw[huh] (vj) edge (dots2);
            \draw[hus] (dots2) edge (vb);
            \draw[sus] (vb) edge (vc);
            \draw[tuh, bend left=20] (vj) edge (vc);
    \end{tikzpicture}
    }\\

    \subfloat[$v_b$ on the path $\pi_j$\label{lemma10.2}]{
    \begin{tikzpicture}
            \node[ndout] (vk) at (-4,0) {$v_k$};
            \node[ndout] (vj) at (-2.5,0) {$v_j$};
            \node (dots1) at (-1,0) {$\cdots$};
            \node[ndout] (vb) at (0.5,0) {$v_b$};
            \node (dots2) at (2,0) {$\cdots$};
            \node[ndout] (vi) at (3.5,0) {$v_i$};
            \draw[sus] (vk) edge (vj);
            \draw[suh] (vj) edge (dots1);
            \draw[huh] (dots1) edge (vb);
            \draw[huh] (vb) edge (dots2);
            \draw[hus] (dots2) edge (vi);
            \draw[hut, bend left = 20] (vk) edge (vb);
    \end{tikzpicture}
    }\\

    \subfloat[Subpath of $\pi$ between $v_a$ ($v_k$) and $v_c$ ($v_i$)\label{lemma10.3}]{
    \begin{tikzpicture}
        \begin{scope} 
            \node[ndout, minimum size=1cm] (va) at (-4,0) {$v_a$ ($v_k$)};
            \node[ndout, minimum size=1cm] (vj) at (-2,0) {$v_j$};
            \node (dots) at (0,0) {$\cdots$};
            \node[ndout, minimum size=1cm] (vb) at (2,0) {$v_b$};
            \node[ndout, minimum size=1cm] (vc) at (4,0) {$v_c$ ($v_i$)};
            \draw[suh] (va) edge (vj);
            \draw[huh] (vj) edge (dots);
            \draw[huh] (dots) edge (vb);
            \draw[hus] (vb) edge (vc);
            \draw[tuh, bend left = 20] (vj) edge (vc);
            \draw[hut, bend right = 20] (va) edge (vb);
        \end{scope}
    \end{tikzpicture}
    }
    \caption{Illustrations for Lemma~\ref{lemma10}.}
\end{figure}
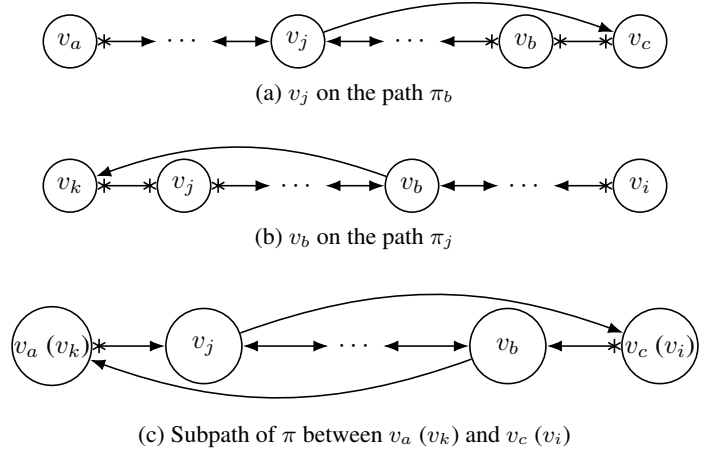

Since \( v_b \) is a non-endpoint node on \( \pi_j \) and distinct from \( v_j \), and since \( v_b \) and \( v_c \) are adjacent on \( \pi \), it follows that \( v_c \) also lies on \( \pi_j \). Consequently, we must have \( v_c \neq v_k \). Moreover, if \( v_c = v_j \), then \( v_j \) would be the endpoint of \( \pi_b \), $v_k$, without being a covered node, contradicting the assumption. Thus, we conclude that \( v_c \neq v_j \).  

Now, suppose that \( v_c = v_i \). Since \( v_j \) must be a non-endpoint node between \( v_a \) and \( v_b \) on \( \pi_b \) (where \( v_j \neq v_a \) and \( v_j \neq v_b \)), and since \( v_k \) is adjacent to \( v_j \), it follows that \( v_k \) lies between \( v_a \) and \( v_b \) on \( \pi_b \) ($v_k\neq v_b$). Moreover, because all nodes between \( v_a \) and \( v_b \), except for \( v_a \), are adjacent to \( v_c = v_i \), we must have \( v_k = v_a \); otherwise, \( v_k \) would be adjacent to \( v_i \), contradicting the assumption that \( \pi_j \) is a discriminating path. Thus, we have \( v_c = v_i \) and \( v_k = v_a \) simultaneously. Then the subpath of $\pi$ between $v_a$ ($v_k$) and $v_c$ ($v_i$) is in the form of Figure~\ref{lemma10.3}.

Consequently, all nodes between \( v_a \) and \( v_c \) on \( \pi \) are colliders that belong to the ancestors of \( \{v_a, v_c\} \). This implies the existence of an inducing path between \( v_a \) and \( v_c \) in \( \mathcal{H} \), contradicting the maximality of \( \sigma \)-MAGs.  

Since \( v_c \) lies on \( \pi_j \) but is distinct from \( v_i, v_j, \) and \( v_k \), it follows that \( v_c \rightarrow v_k \) in \( \mathcal{H} \). Similarly, we obtain \( v_k \rightarrow v_c \) in \( \mathcal{H} \), leading to a directed cycle in \( \mathcal{H} \), which is a contradiction.
\end{proof}

\lemmaeleven*
\begin{proof}
This proof follows from Lemma 11 in \citep{spirtes1996polynomial}, which we reproduce here for the reader's convenience.

Suppose, for the sake of contradiction, that there exists such a triple satisfying all conditions. Let \( \pi_j \) and \( \pi_k \) be the discriminating subpaths of \( \pi \) for \( v_j \) and \( v_k \), respectively. Since \( v_i \) is a covered node on \( \pi_j \), every node between \( v_i \) and \( v_j \) must also lie on \( \pi_j \). Consequently, \( v_k \) is also on \( \pi_j \). Since both \( v_i \) and \( v_j \) are non-endpoint nodes on \( \pi_j \), it follows that \( v_k \) is also a non-endpoint node on \( \pi_j \). Thus, \( v_{k-1} \) and \( v_{k+1} \), which are adjacent to \( v_k \) on $\pi$, must both be on \( \pi_j \), meaning that \( v_k \) remains a covered node on \( \pi_j \). By hypothesis, \( v_j \) is a covered node on \( \pi_k \), which contradicts Lemma~\ref{lemma10}.  
\end{proof}

\lemmatwelve*
\begin{proof}
This proof follows from Lemma 12 in \citep{spirtes1996polynomial}, which we reproduce here for the reader's convenience.

Suppose, for the sake of contradiction, that there exists such a quadruple satisfying all conditions. Let \( \pi_j \) and \( \pi_l \) be the discriminating subpaths of \( \pi \) for \( v_j \) and \( v_l \), respectively. Since \( v_l \) lies between \( v_i \) and \( v_j \), and \( v_i \) is on \( \pi_j \), it follows that \( v_l \) is also a non-endpoint node on \( \pi_j \). Moreover, since \( v_l \) is a covered node on \( \pi \), both of the nodes adjacent to \( v_l \) must be on \( \pi_j \), implying that \( v_l \) is a covered node on \( \pi_j \). Similarly, \( v_j \) is a covered node on \( \pi_l \), contradicting Lemma~\ref{lemma10}.
\end{proof}


\lemmathirteen*
\begin{proof}
This proof follows the same structure as Lemma 13 in \citep{spirtes1996polynomial}, with modifications to adapt it to the framework of \( \sigma \)-MAGs.

Suppose, for the sake of contradiction, that there exists such a sequence of nodes satisfying all conditions. Without loss of generality, assume that \( v_0 \) is to the right of \( v_n \) on \( \pi \). Let \( k \) be the largest index such that \( v_k \) is to the right of \( v_0 \) on \( \pi \), if such a node exists; otherwise, let \( k = 0 \). Now, we will show that \( v_{k+1} \) ($k<n$) is to the left of \( v_n \). If \( k = 0 \), then every other node in the sequence lies to the left of \( v_0 \) on \( \pi \), so \( v_1 \) is to the left of \( v_0 \). By Lemma~\ref{lemma10}, we know that \( v_1 \neq v_n \). Furthermore, by Lemma~\ref{lemma11}, \( v_1 \) is not between \( v_0 \) and \( v_n \), which implies that \( v_1 = v_{k+1} \) is on the left side of \( v_n \) on \( \pi \). If \( k \neq 0 \), then we must have \( v_{k+1} \neq v_n \), as otherwise, this would contradict Lemma~\ref{lemma11}. Moreover, applying Lemma~\ref{lemma12}, we conclude that \( v_{k+1} \) is not between \( v_0 \) and \( v_n \). Hence, \( v_{k+1} \) lies to the left of \( v_n \) on \( \pi \).  

Now, we will show that there exists a node \( v_l \) with \( l \geq k+1 \) on \( \pi \) such that \( v_l \) is to the right of \( v_k \), leading to a contradiction since \( k \) is the largest index for which \( v_k \) is to the right of \( v_0 \). By Lemma~\ref{lemma11}, \( v_{k+2} \) is not between \( v_{k+1} \) and \( v_k \) on \( \pi \), implying that \( v_{k+2} \neq v_n \) (\( k+2 \leq n \), as we have already established that \( v_{k+1} \neq v_n \)). If \( v_{k+2} \) is to the right of \( v_k \), then we are done. Otherwise, consider the case where \( v_{k+2} \) is to the left of \( v_{k+1} \) on \( \pi \). Since \( v_n \) is to the right of \( v_{k+1} \) and $k+2<m$, there must exist a node \( v_l \) with \( l \geq k+3 \) such that \( v_{l-1} \) is to the left of \( v_{k+1} \) and \( v_l \) is to the right of \( v_{k+1} \) on \( \pi \). By Lemma~\ref{lemma12}, \( v_l \) is not between \( v_{k+1} \) and \( v_k \), implying that \( v_l \neq v_n \). Thus, \( v_l \) is to the right of \( v_k \) on \( \pi \), leading to a contradiction. 
\end{proof}


\lemmafifteen*
\begin{proof}
This proof is inspired by the Lemma 15 in \citep{spirtes1996polynomial}, with modifications to adapt it to \( \sigma \)-MAGs.

By definition, in \( \mathcal{H}_1 \), \( a \) is not adjacent to \( c \), \( v_0 \) is a collider on \( \pi \), and \( v_0 \) is an unshielded non-collider on \( a \suh v_0 \rightarrow c \). Since \( \mathcal{H}_1 \) and \( \mathcal{H}_2 \) have the same adjacencies, it follows that in \( \mathcal{H}_2 \), \( a \) is also not adjacent to \( c \). Additionally, by hypothesis, \( v_0 \) remains a collider on \( \pi' \), and \( v_0 \) is an unshielded non-collider on \( a \suh v_0 \tus  c \), as \( \mathcal{H}_1 \) and \( \mathcal{H}_2 \) share the same unshielded colliders. The edge between \( v_0 \) and \( c \) cannot be undirected; otherwise, \( a \) and \( c \) would be adjacent by definition. Thus, we conclude that \( v_0 \rightarrow c \) in \( \mathcal{H}_2 \).  

Now, suppose that for \( 0 \leq i \leq j-1 \) with \(1\leq j \leq n \), we have \( v_i \rightarrow c \) in \( \mathcal{H}_2 \), and by hypothesis, \( v_i \) is a collider on \( \pi' \). Let \( \pi_j' \) be the concatenation of the subpath of \( \pi' \) between \( a \) and \( v_j \) and the edge between \( v_j \) and \( c \) in \( \mathcal{H}_2 \) (since \( \mathcal{H}_1 \) and \( \mathcal{H}_2 \) have the same adjacencies, \( v_j \) and \( c \) must be adjacent). Every node on \( \pi_j' \) between \( v_0 \) and \( v_{j-1} \) is a collider and a parent of \( c \) in \( \mathcal{H}_2 \), implying that \( \pi_j' \) is a discriminating path for \( v_j \) in \( \mathcal{H}_2 \). Let \( \pi_j \) be the corresponding path in \( \mathcal{H}_1 \); then \( \pi_j \) is also a discriminating path for \( v_j \), and \( v_j \) is a non-collider on \( \pi_j \) (since we have \( v_j \rightarrow c \) on \( \pi_j \)). Hence, \( v_j \) must also be a non-collider on \( \pi_j' \) because \( v_j \) is a collider on \( \pi_j' \) in \( \mathcal{H}_2 \) if and only if it is a collider on \( \pi_j \) in \( \mathcal{H}_1 \). It follows that \( v_j \tus  c \) holds in \( \mathcal{H}_2 \), given that \( v_j \) is a collider on \( \pi' \) by hypothesis. Moreover, \( v_j - c \) is not possible; otherwise, we would have \( v_{j-1} \huh  c \), which is contradictory. Therefore, we conclude that \( v_j \rightarrow c \) in \( \mathcal{H}_2 \).  

By induction, all nodes $v_0,\ldots,v_n$ have a directed edge into \( c \) in $\mathcal{H}_2$. Given that every node on \( \pi' \), except for the endpoints and \( b \), is a collider, it follows that \( \pi' \) is a discriminating path for \( b \) in \( \mathcal{H}_2 \).
\end{proof}


\lemmasixteen*
\begin{proof}
This proof is inspired by the Lemma 16 in \citep{spirtes1996polynomial}, with modifications to adapt it to \( \sigma \)-MAGs.

If \( v_k \) is not a covered node on \( \pi \), then since \( \mathcal{H}_1 \) and \( \mathcal{H}_2 \) have the same unshielded colliders, \( v_k \) is a collider on \( \pi \) if and only if \( v_k \) is a collider on \( \pi' \).  

Suppose \( v_k \) is a covered node on \( \pi \). By Proposition~\ref{lemma9}, \( \pi \) contains a unique discriminating subpath for \( v_k \). Let \( \pi_k \) denote the discriminating subpath of \( \pi \) for \( v_k \).  

Now, suppose \( v_k \) is a zero-order covered node on \( \pi \). Then all nodes on \( \pi_k \), except for the endpoints and \( v_k \), are unshielded colliders in \( \mathcal{H}_1 \). Since \( \mathcal{H}_1 \) and \( \mathcal{H}_2 \) have the same unshielded colliders, all unshielded colliders on \( \pi_k \) remain unshielded colliders on \( \pi_k' \), the corresponding path to \( \pi_k \) in \( \mathcal{H}_2 \). Hence, by Lemma~\ref{lemma15}, \( \pi_k' \) is a discriminating path in \( \mathcal{H}_2 \). It follows that \( v_k \) is a collider on \( \pi \) if and only if \( v_k \) is a collider on \( \pi' \), by the third hypothesis of Theorem~\ref{thm3}.  

Now, suppose that for \( 0 \leq i < j \), the \( i^{\text{th}} \)-order covered nodes on \( \pi \) are oriented in the same way as on \( \pi' \). Suppose $v_k$ is a $j^\text{th}$-order covered node on $\pi$. By the induction hypothesis, all colliders on \( \pi_k \), possibly except for \( v_k \), are either not covered or are covered nodes of order less than \( j \). Consequently, these nodes remain colliders on \( \pi_k' \). Hence, by Lemma~\ref{lemma15}, \( \pi_k' \) is also a discriminating path for \( v_k \) in \( \mathcal{H}_2 \). It follows that \( v_k \) is a collider on \( \pi \) if and only if \( v_k \) is a collider on \( \pi' \).
\end{proof}


\lemmanineteen*
\begin{proof}
This proof is inspired by the Lemma 19 in \citep{spirtes1996polynomial}, with modifications to adapt it to \( \sigma \)-MAGs.

The edge between \( a \) and \( c \) cannot be undirected; otherwise, \( \mathcal{H} \) would also contain an edge \( b \rightarrow a \), which is contradictory. If the edge between \( a \) and \( c \) were \( a \leftarrow c \), the structure \( b \rightarrow c \rightarrow a \suh b \) would form a directed cycle or an almost directed cycle, which is also contradictory. Thus, we conclude that \( a \suh c \) must exist in \( \mathcal{H} \).  

If \( a \sus  c \) has a different edge mark at \( a \) than \( a \suh b \), then there are two possible cases:  
\begin{enumerate}  
    \item \( a \huh  b \) and \( a \rightarrow c \),  
    \item \( a \rightarrow b \) and \( a \huh  c \).  
\end{enumerate}  
The second case leads to an almost cycle \( a \rightarrow b \rightarrow c \huh  a \), which is contradictory. Therefore, we consider only the first case.
\end{proof}


{\renewcommand\footnote[1]{}\lemmatwenty*}
\begin{proof}
This proof follows the same structure as Lemma 20 in \citep{spirtes1996polynomial}, with modifications to adapt it to the framework of \( \sigma \)-MAGs.

If \( v_i \) is a collider on \( \pi \), then by Lemma~\ref{lemma16}, both \( \mathcal{H}_1 \) and \( \mathcal{H}_2 \) contain the structure \( v_{i-1} \suh v_i \hus  v_{i+1} \). If \( v_{i-1} \) and \( q \) are not adjacent in \( \mathcal{H}_1 \), then \( v_i \) is an unshielded non-collider on \( v_{i-1} \suh v_i \rightarrow q \) in \( \mathcal{H}_1 \), implying that \( v_i \) is also an unshielded non-collider on \( v_{i-1} \suh v_i \tus  q \) in \( \mathcal{H}_2 \). Since the edge between \( v_i \) and \( q \) cannot be undirected in \( \mathcal{H}_2 \) (otherwise, \( v_{i-1} \) and \( q \) would be adjacent in \( \mathcal{H}_2 \)), it follows that we must have \( v_i \rightarrow q \) in \( \mathcal{H}_2 \). Similarly, if \( v_{i+1} \) and \( q \) are not adjacent in \( \mathcal{H}_1 \), then we conclude that \( v_i \rightarrow q \) in \( \mathcal{H}_2 \) as well. Now, suppose that both \( v_{i-1} \) and \( v_{i+1} \) are adjacent to \( q \).  

There exists a node \( u \) on \( \pi \) between \( v_0 \) and \( v_{i-1} \) that satisfies at least one of the following conditions:  
\begin{enumerate}  
    \item[(i)] \( u \) is not adjacent to \( q \).  
    \item[(ii)] The edge between \( u \) and \( q \) is not into \( q \).  
    \item[(iii)] \( u \) has the same collider/non-collider status on \( \pi \) and on the concatenation of the subpath of \( \pi \) between \( v_0 \) and \( u \) and the edge between \( u \) and \( q \).  
\end{enumerate}  
Such a node must exist since \( v_0 \) satisfies either condition (i) if \( v_0 \) is not adjacent to \( q \) in \( \mathcal{H}_1 \), or condition (iii) if \( v_0 \) is adjacent to \( q \) in \( \mathcal{H}_1 \). Let \( v_{i-m_1} \) (where \( m_1 \geq 1 \)) be the closest such node to the left of \( v_i \). Similarly, there exists a node on \( \pi \) between \( v_{i+1} \) and \( v_n \) that satisfies at least one of conditions (i) or (ii), or the following condition:  
\begin{enumerate}  
    \item[(iii')] \( u \) has the same collider/non-collider status on \( \pi \) and on the concatenation of the subpath of \( \pi \) between \( v_n \) and \( u \) and the edge between \( u \) and \( q \).  
\end{enumerate}  
Let \( v_{i+m_2} \) (where \( m_2 \geq 1 \)) be the closest such node to the right of \( v_i \).  

We aim to show that every non-endpoint node between $v_{i - m_1}$ and $v_i$ on $\pi$ is a collider and has a directed edge into $q$. If $m_1 = 1$, this holds trivially since there are no non-endpoint nodes between \( v_{i - m_1} \) and \( v_i \). Thus, we assume \( m_1 \geq 2 \). We will show that for \( 1 \leq k \leq m_1 - 1 \), $v_{i-k}$ is a collider on $\pi$ and there exists an edge \( v_{i - k} \tuh q \) in \( \mathcal{H}_1 \). Since \( v_{i - 1} \) lies between \( v_{i - m_1} \) and \( v_i \) but does not satisfy the above conditions, $\mathcal{H}_1$ contains the edge $v_{i-1}\suh q$, and $v_{i-1}$ must have a different collider/non-collider status on $\pi$ compared to the concatenation of the subpath of $\pi$ from $v_0$ to $v_{i-1}$ together with the edge between $v_{i-1}$ and $q$. This implies that $\mathcal{H}_1$ also contains the edge $v_{i-2} \suh v_{i-1}$, and that the edge between \( v_{i - 1} \) and \( v_i \) has a different edge mark at \( v_{i - 1} \) compared to the edge between \( v_{i - 1} \) and \( q \) in \( \mathcal{H}_1 \). By Lemma~\ref{lemma19}, it follows that \( v_{i - 1} \tuh q \) and \( v_{i - 1} \huh v_i \). Thus, $v_{i-1}$ is a collider on $\pi$ and has a directed edge into $q$. This settles the case when $m_1 = 2$. Now assume $m_1 \geq 3$. Suppose that for \( 1 \leq l \leq k - 1 \), $v_{i-l}$ is a collider on $\pi$ and there exists an edge \( v_{i - l} \tuh q \) in \( \mathcal{H}_1 \). Since \( v_{i - l} \) lies between \( v_{i - m_1} \) and \( v_i \), we have \( v_{i - k} \suh q \), and \( v_{i - k} \) exhibits a different collider/non-collider status on \( \pi \) compared to the concatenation of the subpath of \( \pi \) from \( v_0 \) to \( v_{i - k} \) and the edge between \( v_{i - k} \) and \( q \). This implies that $\mathcal{H}_1$ contains the edge $v_{i - k - 1} \suh v_{i - k}$, and that the edge between \( v_{i - k} \) and \( v_{i - k + 1} \) has a different edge mark than the edge between \( v_{i - k} \) and \( q \) in \( \mathcal{H}_1 \). By Lemma~\ref{lemma19}, it follows that \( v_{i - k} \tuh q \) and \( v_{i - k - 1} \suh v_{i - k} \huh v_{i - k + 1} \). Hence, every non-endpoint node between \( v_{i - m_1} \) and \( v_i \) (if any) is a collider on \( \pi \) and has a directed edge into $q$. Similarly, every non-endpoint node between \( v_i \) and \( v_{i + m_2} \) is a collider on \( \pi \) and has a directed edge into \( q \).

Suppose \( v_{i-m_1} \) is adjacent to \( q \), and we have \( v_{i-m_1+1} \rightarrow q \) and \( v_{i-m_1} \suh v_{i-m_1+1} \) by induction. Then, by Lemma~\ref{lemma19}, it follows that \( v_{i-m_1} \suh q \). By hypothesis, \( v_{i-m_1} \) retains the same collider/non-collider status on \( \pi \) and on the concatenation of the subpath of \( \pi \) between \( v_0 \) and \( v_{i-m_1} \), along with the edge between \( v_{i-m_1} \) and \( q \) in \( \mathcal{H}_1 \). Similarly, if \( v_{i+m_2} \) is adjacent to \( q \), then \( v_{i+m_2} \) has the same collider/non-collider status on \( \pi \) and on the concatenation of the subpath of \( \pi \) between \( v_{i+m_2} \) and \( v_n \), along with the edge between \( v_{i+m_2} \) and \( q \). Then, denote the concatenation of the subpath of \( \pi \) between \( v_0 \) and \( v_{i-m_1} \), the edge between \( v_{i-m_1} \) and \( q \), the edge between \( q \) and \( v_{i+m_2} \), and the subpath of \( \pi \) between \( v_{i+m_2} \) and \( v_n \) by \( \mu \). Notice that there does exist subpaths of the form \( v_{i-m_1-1}-v_{i-m_1}\huh  q \) or \( q\huh  v_{i+m_2}-v_{i+m_2+1} \) on \( \mu \), since otherwise \( \mathcal{H}_1 \) contains an almost directed cycle \( v_{i-m_1}\rightarrow v_{i-m_1+1}\rightarrow q\huh  v_{i-m_1} \) or \( v_{i+m_2}\rightarrow v_{i+m_2+1}\rightarrow q\huh  v_{i+m_2} \). Thus, \( \mu \) is $m$-open given \( Z \) and is shorter than \( \pi \), which is a contradiction. The only exception occurs when \( m_1 = m_2 = 1 \), in which case the concatenated path has the same length as \( \pi \), but a smaller sum of distances from colliders to \( Z \), which is also a contradiction. It follows that at least one of \( v_{i-m_1} \) or \( v_{i+m_2} \) is not adjacent to \( q \).

W.L.O.G., suppose \( v_{i-m_1} \) is not adjacent to \( q \). Since we assume that \( v_{i-1} \) is adjacent to \( v_i \) in \( \mathcal{H}_1 \), it follows that \( m_1 \geq 2 \). Denote by \( \pi_i \) the concatenation of the subpath of \( \pi \) between \( v_{i-m_1} \) and \( v_i \) along with the edge \( v_i \rightarrow q \) in \( \mathcal{H}_1 \), and let \( \pi_i' \) be the corresponding path in \( \mathcal{H}_2 \). By definition, \( \pi_i \) is a discriminating path for \( v_i \) in \( \mathcal{H}_1 \), and \( v_i \) is a non-collider on this path. By Lemma~\ref{lemma16}, all colliders on \( \pi \) remain colliders on \( \pi' \), the corresponding path to \( \pi \) in \( \mathcal{H}_2 \). Furthermore, by Lemma~\ref{lemma15}, \( \pi_i' \) is a discriminating path for \( v_i \). Hence, \( v_i \) is a non-collider on \( \pi_i' \) by assumption. Since \( \mathcal{H}_2 \) contains the edges \( v_{i-1} \huh v_i \) and \( v_{i-1} \rightarrow q \), we cannot have \( v_i \tut q \) in \( \mathcal{H}_2 \). Therefore, we conclude that \( v_i \rightarrow q \) in \( \mathcal{H}_2 \).

\end{proof}

\lemmatwentyone*
\begin{proof}
This proof is inspired by the Lemma 21 in \citep{spirtes1996polynomial}, with modifications to adapt it to \( \sigma \)-MAGs.

Suppose $\mu$ is in the form \( v_k = u_0 \tuh u_1 \tuh \cdots \tuh u_m = q \). Let \( \pi' \) and \( \mu' \) be the corresponding paths to \( \pi \) and \( \mu \), respectively, in \( \mathcal{H}_2 \). By Lemma~\ref{lemma20}, the first edge on \( \mu' \) is directed and points out of \( u_0 \). Suppose there exists a subpath of \( \mu' \) in the form \( u_i \rightarrow u_{i+1} - \cdots - u_j \hus u_{j+1} \), with $j\geq i + 1$. By Lemma~\ref{oldDef}, this implies that \( \mathcal{H}_2 \) contains a triple \( u_i \rightarrow u_j \hus u_{j+1} \). Since \( \mu \) contains the edge \( u_j \rightarrow u_{j+1} \), and \( \mathcal{H}_1 \) and \( \mathcal{H}_2 \) share the same unshielded colliders, it follows that \( u_i \) and \( u_{j+1} \) must be adjacent in \( \mathcal{H}_2 \), and consequently, also adjacent in \( \mathcal{H}_1 \). Notice that \( \mu \) contains a subpath \( u_i \rightarrow \cdots \rightarrow u_j \rightarrow u_{j+1} \). If the edge between \( u_i \) and \( u_{j+1} \) were \( u_i \hus  u_{j+1} \), this would lead to a directed or almost directed cycle, which is a contradiction. If the edge were undirected, i.e., \( u_i - u_{j+1} \), then by Lemma~\ref{oldDef}, it would also contain the edge \( u_j \rightarrow u_i \), creating a directed cycle, which is again a contradiction. If the edge were \( u_i \rightarrow u_{j+1} \), replacing the subpath of \( \pi \) between \( u_i \) and \( u_{j+1} \) with \( u_i \rightarrow u_{j+1} \) would yield a shorter directed path from \( v_k \) to \( q \), which contradicts the assumption that \( \pi \) is the shortest path. Therefore, no such subpath exists in \( \mu' \), and we conclude that \( \mu' \) must be in the form \( u_0 \rightarrow u_1 \tus \cdots \tus u_m \). By Lemma~\ref{lemmaant}, this implies that \( v_k \) is an ancestor of \( q \) in \( \mathcal{H}_2 \).
\end{proof}

\thmthreeb*
\begin{proof}
Suppose $\mathcal{H}_1, \mathcal{H}_2$ satisfy Condition~\ref{condition}. For any subsets of the nodes $X, Y, Z \subseteq \mathcal{V}$, if we have $X \overset{m}{\underset{\mathcal{H}_1}{\not\perp}} Y \mid Z$ in $\mathcal{H}_1$, then there exists a shortest $m$-open path $\pi$ given $Z$ with the smallest collider distance sum to $Z$ from a node $v_0 \in X$ to a node $v_n \in Y$ in $\mathcal{H}_1$. Let $\pi'$ be the corresponding path in $\mathcal{H}_2$. By Lemma~\ref{lemma16}, we know that all colliders on $\pi$ remain colliders on $\pi'$. Furthermore, by Lemma~\ref{lemma21}, these colliders are still ancestors of $Z$ in $\mathcal{H}_2$. Additionally, by Lemma~\ref{lemma16}, all non-colliders on $\pi$ remain non-colliders on $\pi'$, and they do not belong to $Z$ since $\pi$ is $Z$-$m$-open. 

We now verify whether $\pi'$ contains a subpath of the form \( v_i \suh v_{i+1} \tut v_{i+2} \) or \( v_i \tut v_{i+1} \hus v_{i+2} \). Without loss of generality, assume the former case holds. By Lemma~\ref{oldDef}, $\mathcal{H}_2$ also contains an edge of the form $v_i \suh v_{i+2}$. Since $\mathcal{H}_1$ and $\mathcal{H}_2$ share the same adjacencies, $v_{i+1}$ is a covered node on $\pi$. Applying Lemma~\ref{lemma7} and Lemma~\ref{lemma16}, we deduce that $v_{i+2}$ is a non-collider on $\pi$, and $v_i$ is a collider on $\pi$ with a directed edge $v_i \rightarrow v_{i+2}$. By Lemma~\ref{lemma16}, $v_i$ must also be a collider on $\pi'$, implying the existence of the edges:
\[
v_i \huh v_{i+1} \quad \text{on } \pi', \quad \text{and} \quad v_i \huh v_{i+2} \quad \text{in } \mathcal{H}_2.
\]
Furthermore, by Proposition~\ref{lemma9}, there exists a node $v_j$ with $0 \leq j < i$ such that the subpath $\mu$ of $\pi$ between $v_j$ and $v_{i+2}$ is a discriminating path for $v_{i+1}$. Applying Lemma~\ref{lemma15}, the corresponding path $\mu'$ in $\mathcal{H}_2$ is also a discriminating path for $v_{i+1}$. However, this contradicts the presence of the edge $v_i \huh  v_{i+2}$ in $\mathcal{H}_2$. Thus, $\pi'$ cannot contain such subpaths, and it follows that $\pi'$ is $m$-open given $Z$ in $\mathcal{H}_2$. The case where \( X \overset{m}{\underset{\mathcal{H}_2}{\not\perp}} Y \mid Z \) in $\mathcal{H}_2$ is analogous. We therefore conclude that $\mathcal{H}_1$ and $\mathcal{H}_2$ are $m$-Markov equivalent.
\end{proof}

\thmthree*
\begin{proof}
Obviously obtained by Lemma~\ref{thm3a} and Lemma~\ref{thm3b}.
\end{proof}

\thmfour*
\begin{proof}
Obviously obtained by Theorem~\ref{thm2} and Theorem~\ref{thm3}.
\end{proof}

\end{document}